\documentclass[a4paper]{article}

\usepackage{graphicx}
\usepackage{latexsym,amsfonts,amsmath,amsthm,dsfont,mathrsfs}
	\DeclareMathOperator\arccot{arccot}
	\DeclareMathAlphabet\mathbfcal{OMS}{cmsy}{b}{n}	
\usepackage{authblk} 				
\usepackage[inline]{enumitem}  		
\usepackage{hyperref}				
\usepackage{rotating}
\usepackage{natbib}          			
\usepackage{setspace} 				
\usepackage[usenames,dvipsnames,svgnames,table]{xcolor} 	
\usepackage{soul} 			
\usepackage{framed}						
	\colorlet{shadecolor}{blue!20}

\theoremstyle{definition}
\newtheorem{thm}{Theorem} 

\begin{document}
%
%
\title{The Amplitude-Phase Decomposition for \\the Magnetotelluric Impedance Tensor}
%
%
	\author[$\ddagger$]{Maik Neukirch}
	\affil[$\ddagger$]{presently: independent researcher; previously: Barcelona Center for Subsurface Imaging, Institut de Ci\`encies del Mar, Barcelona, Spain}

	\author[$\,*$]{Daniel Rudolf}
	\affil[$\,*$]{Institute for Mathematical Stochastics, Georg-August-Universit\"at G\"ottingen, Germany}

	\author[$\dagger$]{Xavier Garcia}
	\affil[$\dagger$]{Barcelona Center for Subsurface Imaging, Institut de Ci\`encies del Mar, Barcelona, Spain}
	
	\author[$\ddagger$]{Savitri Galiana} 
%
%
{
\maketitle
\vspace{-.75cm}
{\begin{center} 
Short title: Amplitude-Phase Decomposition
\end{center}}
\begin{abstract}
The Phase Tensor marked a breakthrough in understanding and analysis of electric galvanic distortion. It can be used for (distortion free) dimensionality analysis, distortion analysis, mapping and subsurface model inversion. However, it does not contain any impedance amplitude information and therefore cannot quantify resistivity without complementary data. 
Any technique considering the impedance amplitude as well as the Phase Tensor will result in a more reliable interpretation. 
We formulate a complete impedance tensor decomposition into the Phase Tensor and a new Amplitude Tensor that is shown to be complementary and mathematically independent to the Phase Tensor. 
We show that for the special cases of 1D and 2D models, the geometric Amplitude Tensor parameters (strike and skew angles) converge to Phase Tensor parameters and the singular values of the Amplitude Tensor correspond to the impedance amplitudes of the transverse electric (TE) and transverse magnetic (TM) modes. 
In all cases, we show that the Amplitude Tensor contains both galvanic and inductive amplitudes, the latter of which is argued to be physically related to the inductive information of the Phase Tensor. 
The geometric parameters of the inductive Amplitude and the Phase Tensors represent the same geometry of the subsurface conductivity distribution that is affected by induction processes, and therefore we hypothesise that geometric Phase Tensor parameters can be used to approximate the inductive Amplitude Tensor. 
Then, this hypothesis leads to the estimation of the galvanic Amplitude Tensor which is equal to the galvanic electric distortion tensor at the lowest measured period. 
This estimation of the galvanic distortion departs from the common assumption to consider 1D or 2D regional structures and can be applied for general 3D subsurfaces. 
We demonstrate exemplarily with an explicit formulation how our hypothesis can be used to recover the galvanic electric anisotropic distortion for 2D subsurfaces, which was, until now, believed to be indeterminable for 2D data. 
Moreover, we illustrate the Amplitude Tensor as a mapping tool and we compare it to the Phase Tensor with both synthetic and real data examples. 
Lastly, we argue that the Amplitude Tensor can provide important non-redundant amplitude information to Phase Tensor inversions.

\begin{center}\rule{12cm}{0.02cm}\end{center}
\end{abstract}
}
%
%


\section{Introduction}

Galvanic distortion of the electric field remains a recognised complication for Magnetotelluric (MT) 3D data and, if present, can lead to inversion artefacts and incorrect interpretation of obtained subsurface images \citep[][and references therein]{jones:2012b}. It is caused by near-surface inhomogeneities below the size of the experiment's resolution defined by its sampling rate \citep[][and references therein]{berdichevsky:1976,Larsen:1977,Bahr:1988,Groom:1989,Jiracek:1990,Jones:2011}. 
Two main methods have emerged to deal with it. 
The first one models either near surface inhomogeneities or their effects, and removes them or accounts for them \citep{berdichevsky:1976,deGroot:1991,deGroot:1995,baba:2005,miensopust:2010,patro:2011,avdeeva:2015}. 
The second one uses analytic solutions either based on structural assumptions of the impedance tensor, which must be 1D or 2D  
\citep{Larsen:1977,Bahr:1988,Groom:1989,Chave:1994,McNeice:2001,garcia:2003,Bibby:2005,jones:2012}, 
or requiring a sufficiently dense spatial distribution of sites for general 3D data \citep{utada:2000,garcia:2002}. 
However, both groups that approached 3D data reported instabilities when they analysed the impedance tensors of real measurements. Besides, the modelling of distortion effects has grown more popular than analytic solutions since \cite{Caldwell:2004} introduced the distortion free Phase Tensor, $\mathbf{\Phi}$. 

The Phase Tensor is defined as the matrix multiplication between the inverse real and the imaginary parts of the impedance tensor, $\mathbf{Z}=\mathbf{X}+i\mathbf{Y}$, 
 \begin{equation}
\mathbf{\Phi}=\mathbf{X}^{-1}\mathbf{Y}, 
\label{eq:PhaseTensor}
\end{equation} 
so that for a distorted impedance $\mathbf{Z_d}=\mathbf{CZ}$, the Phase Tensor 
$\mathbf{\Phi_d}=\mathbf{X}^{-1}\mathbf{C}^{-1}\mathbf{C} \mathbf{Y}$ 
is free of the distortion $\mathbf{C}$. It contains four algebraic reformulations of the eight components (four real and four imaginary) of the MT impedance tensor such that it is dimensionless \citep{Caldwell:2004}. Moreover, the Phase Tensor has been widely accepted and is used for distortion free dimensionality analysis \citep{Marti:2014}, distortion analysis \citep{Bibby:2005}, tensor parameter mapping \citep{Booker:2014} and subsurface model inversion \citep{patro:2012,Tietze:2015}. 
Although it only indicates changes of subsurface conductivity, the recovery of absolute values by inversion has been demonstrated successfully \citep{patro:2012,Tietze:2015} if information of several sites overlaps sufficiently and an appropriate a-priori model has been chosen. 
Nevertheless, the authors acknowledge that amplitude information would be a great asset for Phase Tensor inversions in order to relax these demanding requirements. For a more complete treatment of the Phase Tensor and galvanic distortion the reader is referred to \cite{weidelt:2012b} and \cite{jones:2012b}. 

In this work, in analogy to the Phase Tensor definition and with the motivation to obtain an observable that contains amplitude information but is independent to the Phase Tensor, we derive a new algebraic matrix decomposition that is applicable to the full MT impedance tensor and contains the Phase Tensor by introducing the Amplitude Tensor. 
We prove the existence and uniqueness of the proposed Amplitude Tensor and its mathematical independence to the Phase Tensor. Additionally, we compare the information present in the Phase and Amplitude Tensors in terms of parameters that are directly related to the geometry of the subsurface conductivity distribution, like the strike angle, dimensionality and anisotropy. 
Furthermore, we show that electric galvanic distortion is contained in the Amplitude Tensor 
and can be expressed as a matrix multiplication. We postulate the hypothesis that the regional inductive Amplitude Tensor parameters 
can be approximated with Phase Tensor parameters. This hypothesis allows the estimation of local distortion up to a single constant, the static shift factor, without any assumptions on the dimensionality of the regional impedance. We demonstrate this finding explicitly on a synthetic example where we estimate anisotropic distortion at a single site, a problem that cannot be solved with any other available solution to date. We argue that this hypothesis can be further exploited to recover all distortion parameters except for the static shift. Then, we illustrate how Amplitude and Phase Tensor parameters complement each other with respect to subsurface information inferred from map plots of real and synthetic data. Lastly, we discuss that the Amplitude Tensor is particularly suited to be included in a Phase Tensor inversion scheme, because of its unique properties of containing purely impedance amplitude information and being mathematically independent to the Phase Tensor.

\section{The Amplitude Tensor}
\subsection{Notation and Preliminaries}
For natural numbers $n$ we denote by $\mathbb{C}^{n\times n}$ and $\mathbb{R}^{n\times n}$ the complex and real $n\times n$ matrix spaces, respectively. Let $\mathbf{\mathbb{I}}$ be the identity matrix and for $\mathbf{N}\in \mathbb{R}^{n\times n}$ let $\mathbf{N}^T$ be the transpose of $\mathbf{N}$. 
For $\lambda_1,\dots,\lambda_n \in \mathbb{R}$ let $\mathbf{\Sigma}_\lambda$ be a diagonal matrix with $\lambda_1,\dots,\lambda_n$ on its diagonal.

\cite{Booker:2014} proposes a (phase) tensor parameterization based on geometric
considerations, which is closely related to the singular value decomposition (SVD) of the Phase Tensor and decomposes the Phase Tensor in its singular values and two rotations related to the left and right orthogonal SVD matrices. The virtue of this parameterization is that the rotations are related to two undistorted, important observables: the strike angle that is representative for the geologic strike and the skew angle that is associated with subsurface dimensionality. Let us adopt this parameterization for general matrices $\mathbf{X}\in
\mathbb{R}^{2\times2}$ in the following:
 \begin{equation}
\mathbf{X} =
 \begin{pmatrix}
X_{1,1}	&X_{1,2}      \\
X_{2,1} 	&X_{2,2}
\end{pmatrix}
=  \mathbf{R}(-\theta_X)   \begin{pmatrix} x_1 & 0\\0 & x_2 \end{pmatrix}
  \mathbf{R}(\psi_X)   \mathbf{R}(\theta_X)=\mathbf{V}_X \mathbf{\Sigma}_X \mathbf{W}_X^T,
\label{eq:TensorParameterBooker}
\end{equation} 
where $\theta_X=\arccos\left(\mathrm{Tr}\left(\mathbf{V}_X\right)/2\right)$ is 
the angle between the cartesian and the matrix ellipse coordinates that, for the Phase Tensor, represents the geologic strike angle in 2D cases \citep{Booker:2014};
 \begin{equation}
	\label{psi_X}
  \psi_X=\begin{cases}
    \arctan\frac{X_{1,2}-X_{2,1}}{X_{1,1}+X_{2,2}}, & \text{if $0\le|X_{1,2}-X_{2,1}| \le |X_{1,1}+X_{2,2}| \ne 0$},\\
    \arccot\frac{X_{1,1}+X_{2,2}}{X_{1,2}-X_{2,1}}, & \text{if $|X_{1,2}-X_{2,1}| > |X_{1,1}+X_{2,2}|$ and $\frac{X_{1,1}+X_{2,2}}{X_{1,2}-X_{2,1}}\ge0$},\\
    \arccot\frac{X_{1,1}+X_{2,2}}{X_{1,2}-X_{2,1}}-\pi, & \text{if $|X_{1,2}-X_{2,1}| > |X_{1,1}+X_{2,2}|$ and $\frac{X_{1,1}+X_{2,2}}{X_{1,2}-X_{2,1}}<0$},
  \end{cases}
\end{equation} 
is the normalised matrix skew angle \citep{Booker:2014}; $x_1$ and $x_2$, in the diagonal matrix $\mathbf{\Sigma}_X$, are the singular values of $\mathbf{X}$; and $\mathbf{V}_X=\mathbf{R}(-\theta_X)$ and $\mathbf{W}_X=\mathbf{R}(-\psi_X-\theta_X)$ are the left and
right orthogonal SVD matrices, with the rotation by $\delta$ defined as:
 \begin{equation}
\mathbf{R} (\delta) = \begin{pmatrix} \cos\delta & \sin\delta\\
                            -\sin\delta & \cos\delta \end{pmatrix}.
\label{eq:RotMat}
\end{equation} 
The normalised matrix skew angle in \eqref{psi_X} ensures that the definition of $\hat{\psi}_X= \arctan\frac{X_{1,2}-X_{2,1}}{X_{1,1}+X_{2,2}}$ by \cite{Booker:2014} is defined for diagonal and anti-diagonal matrices $\mathbf{X}$ by expressing the limit of the arctangent function with arccotangent, viz.~$\lim_{y \to \pm 0} \arctan \frac{x}{y}=\pm\frac{\pi}{2}=\arccot\frac{y}{x}-\frac{\pi}{2}\pm\frac{\pi}{2}, \forall x>0$, and therewith, provides improved numerical stability for $\psi_X>45^\circ$.

\subsection{Amplitude-Phase Decomposition}
Let us introduce a new Amplitude-Phase decomposition of the impedance tensor as a generalisation of the polar form of the complex numbers, represented by a real-valued amplitude and phase, for complex matrices. The Amplitude-Phase decomposition is a new, to the best knowledge of the authors, mathematical decomposition of a complex matrix in real-valued Phase and Amplitude matrices, which is general for any complex matrix $\mathbf{M}\in\mathbb{C}^{n\times n}$.
We proof the existence of such a generalization and argue for the uniqueness of a principal solution in appendix \ref{app:proofs}. We derive this formulation in analogy to the polar form of complex numbers
 \begin{equation}
z = \varrho   \exp{(i \varphi)}=
    \varrho   (\cos{\varphi}+i\sin{\varphi}),
\label{eq:polForm}
\end{equation} 
using the property that the phase, $\varphi$, and the amplitude, $\varrho$, are real-valued, and that the Pythagorean identity: 
 \begin{equation}
\cos^2{\varphi}+\sin^2{\varphi}=1,
\end{equation} 
is satisfied. 

Thus, let us represent $\mathbf{M}\in \mathbb{C}^{n\times n}$ in terms of two real-valued matrices, $\mathbf{P}$ (read capital greek letter Rho) and $\mathbf{\Phi}$, associated with the amplitude ($\varrho$) and phase ($\varphi$) of the polar form in \eqref{eq:polForm}, with defined functions $c(\mathbf{\Phi})$, $s(\mathbf{\Phi})$ and $e(\mathbf{\Phi})$ analogue to $\cos\varphi$, $\sin\varphi$ and $\exp{(i\varphi)}$:
 \begin{equation}
\mathbf{M}=\mathbf{P}e(\mathbf{\Phi})=\mathbf{P}(c(\mathbf{\Phi})+is(\mathbf{\Phi})).
\label{eq:phaseFun}
\end{equation} 
In particular, the sum of squares of the amplitude independent terms shall be unity in analogy to the Pythagorean identity:
 \begin{equation}
c(\mathbf{\Phi})  c(\mathbf{\Phi})^T + s(\mathbf{\Phi})  s(\mathbf{\Phi})^T=\mathbb{I}.
\label{eq:trigIdent}
\end{equation} 
The definition of $\mathbf{\Phi}$ in \eqref{eq:PhaseTensor}, the decomposition in \eqref{eq:phaseFun} and the constraint in \eqref{eq:trigIdent} yield solutions for $e$ such that a real-valued amplitude matrix $\mathbf{P}$ always exists if $e(\mathbf{\Phi})$ is invertible. Further, it can be argued that one fundamental solution for $\mathbf{P}$ exists where $e(\mathbf{\Phi})$ shares SVD properties with the unique phase matrix (see appendix \ref{app:proofs}).

\subsection{Definition of the Amplitude Tensor}
Let us represent the invertible, complex-valued $2\times2$ MT impedance tensor,
$\mathbf{Z}=\mathbf{X}+i\mathbf{Y}$, using \eqref{eq:phaseFun} and constraint by \eqref{eq:trigIdent} as a
multiplication of a real-valued Amplitude Tensor, $\mathbf{P}$, with a function, $e$, of the
invertible, real-valued Phase Tensor, $\mathbf{\Phi}=\mathbf{X}^{-1}\mathbf{Y}$. Then, we obtain the following fundamental solutions (see appendix \ref{app:proofs}) for $c(\mathbf{\Phi})$, $s(\mathbf{\Phi})$ and $e(\mathbf{\Phi})$:
 \begin{equation}
c(\mathbf{\Phi}) = \left(\mathbb{I}+\mathbf{\Phi\Phi}^T\right)^{-\frac{1}{2}},\qquad
s(\mathbf{\Phi}) =\left(\mathbb{I}+\mathbf{\Phi\Phi}^T\right)^{-\frac{1}{2}} \mathbf{\Phi},\qquad
e(\mathbf{\Phi})=c(\mathbf{\Phi})+is(\mathbf{\Phi}),
\label{eq:SolCSE}
\end{equation} 
where the superscript $-\frac{1}{2}$ denotes the matrix inverse of the matrix square root. 
Hence, as simply follows from \eqref{eq:phaseFun}, an unique and real-valued Amplitude Tensor can be defined as:
 \begin{equation}
\mathbf{P} =\mathbf{Z}  (e(\mathbf{\Phi}))^{-1}.
\label{eq:ampTen}
\end{equation} 

\subsection{Mathematical Relation to the Phase Tensor}
\label{sec:AP_relation}
From the definitions of the Phase Tensor in \eqref{eq:PhaseTensor} and the Amplitude-Phase Decomposition in \eqref{eq:phaseFun} it is seen directly that the Phase Tensor is mathematically unrelated to the Amplitude Tensor:
 \begin{equation}
\mathbf{\Phi}=\mathbf{X}^{-1}\mathbf{Y}=(\mathbf{P}c(\mathbf{\Phi}))^{-1}\mathbf{P}s(\mathbf{\Phi})=(c(\mathbf{\Phi}))^{-1}\mathbf{P}^{-1}\mathbf{P}s(\mathbf{\Phi})=(c(\mathbf{\Phi}))^{-1}s(\mathbf{\Phi}),
\end{equation} 
and therefore, we argue that the Amplitude and Phase Tensors must 
contain distinct, complementary information. The potential of the proposed Amplitude-Phase decomposition lies just in this property of mathematical independence between both tensors since it guarantees that any similarity between them cannot originate from the mathematical construction of the decomposition but must be due to the underlying physics that result in the present impedance tensor. This result contrasts with alternative decompositions proposed in the literature that involve the Phase Tensor in combination with either the real or imaginary part of the impedance, i.e.~$\mathbf{Z}=\mathbf{X}(\mathbb{I}+i\mathbf{\Phi})=\mathbf{Y}(\mathbf{\Phi}^{-1}+i\mathbb{I})$, because in those ones there is no mathematical independence between the decomposed tensors since the Phase Tensor itself is always derived from the other.

Let us note here that we denote $\mathbf{P}$ as the MT Amplitude Tensor, because the MT Phase Tensor, $\mathbf{\Phi}$, contains only impedance phase information and $\mathbf{P}$ is, on one hand, mathematically independent to $\mathbf{\Phi}$ and thus, devoid of phase information, and, on the other hand, forms with $\mathbf{\Phi}$ a complete decomposition of the MT impedance tensor, so that all impedance amplitude information must be contained in $\mathbf{P}$. 

\subsection{Rotational Invariants}
The MT impedance can be represented by seven rotational invariant parameters and one rotational dependent parameter, the phase strike angle \citep{weaver:2006,lilley:2012}. This number of rotational invariant parameters must be preserved after the Amplitude-Phase decomposition. The Amplitude and Phase Tensors, parameterised by \eqref{eq:TensorParameterBooker}, contain a rotationally dependent strike angle and three rotational invariants: the skew angle and two singular values. Then, since the impedance only features one rotational variant, the phase and amplitude strike angles must be related by another rotational invariant. This one can be deduced by parameterising the Amplitude and Phase Tensors in \eqref{eq:phaseFun} with \eqref{eq:SolCSE} and using \eqref{eq:TensorParameterBooker}:
 \begin{equation}
\mathbf{Z}=\mathbf{P}e(\mathbf{\Phi})=\mathbf{R}(-\theta_\Phi)\mathbf{R}(-(\theta_P-\theta_\Phi))\mathbf{\Sigma}_P
\mathbf{R}(\psi_P)\mathbf{R}(\theta_P-\theta_\Phi)(\mathbf{\Sigma}_c+i\mathbf{\Sigma}_s\mathbf{R}(\psi_\Phi))\mathbf{R}(\theta_\Phi),
\label{eq:Z2APdecomposed}
\end{equation} 
where the subscript is used in all the expression to indicate if it is an Amplitude or a Phase Tensor parameter, but in $\mathbf{\Sigma}_c$ and $\mathbf{\Sigma}_s$, which are $\mathbf{\Sigma}_c=\left(\mathbb{I}+\mathbf{\Sigma}_\Phi\mathbf{\Sigma}_\Phi^T\right)^{-\frac{1}{2}}$ and $\mathbf{\Sigma}_s=\mathbf{\Sigma}_c\mathbf{\Sigma}_\Phi$, containing the phase singular values.
Note, that the phase strike angle rotation, $\mathbf{R}(\theta_\Phi)$, in \eqref{eq:Z2APdecomposed} brackets the seven other parameters and therefore, any coordinate rotation applied to the impedance, i.e.~a counter-clockwise coordinate rotation by angle $\alpha$ that yields $\mathbf{Z}_\mathrm{rot}=\mathbf{R}(-\alpha)\mathbf{ZR}(\alpha)$, can only affect the phase strike angle, i.e.~$\theta_{\Phi,\mathrm{rot}}=\theta_\Phi+\alpha$. Hence, $\theta_\Phi$ is the only rotationally dependent parameter and the strike angle difference $\theta_P-\theta_\Phi$ is the seventh rotational invariant in the parameterization of the Amplitude-Phase decomposition. 

\subsection{Properties According to Different Dimensionality Cases}
\label{sec:AmpDim}
\begin{figure}[t]\centering
		\noindent\includegraphics[width=.48\textwidth]{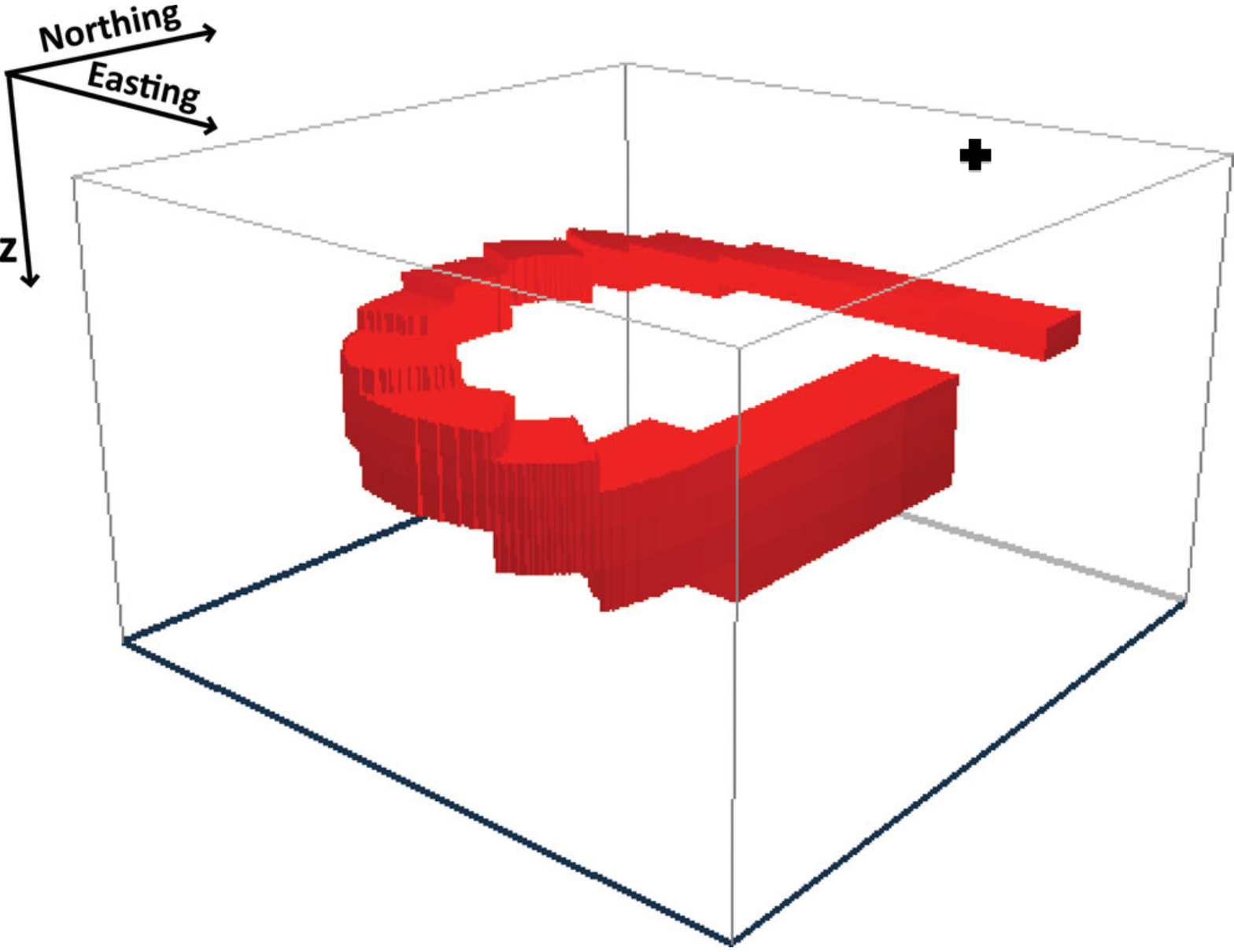}
	\caption{Dublin Secret Model 1 (DSM1), the red body represents a $1\,\Omega m$ conductive structure embedded in a $100\,\Omega m$ background. Station $A09$, marked by a black cross, is situated at the surface in the North-East corner on top of both, the shallowest and the deepest part of the conductor. Illustration adopted and modified from \cite{Miensopust:2013}. 	}
	\label{fig:DSM1model}
\end{figure}
\begin{figure}[t]\centering
		a) Period-scaled real part of impedance\\
		\noindent\includegraphics[width=.58\textwidth]{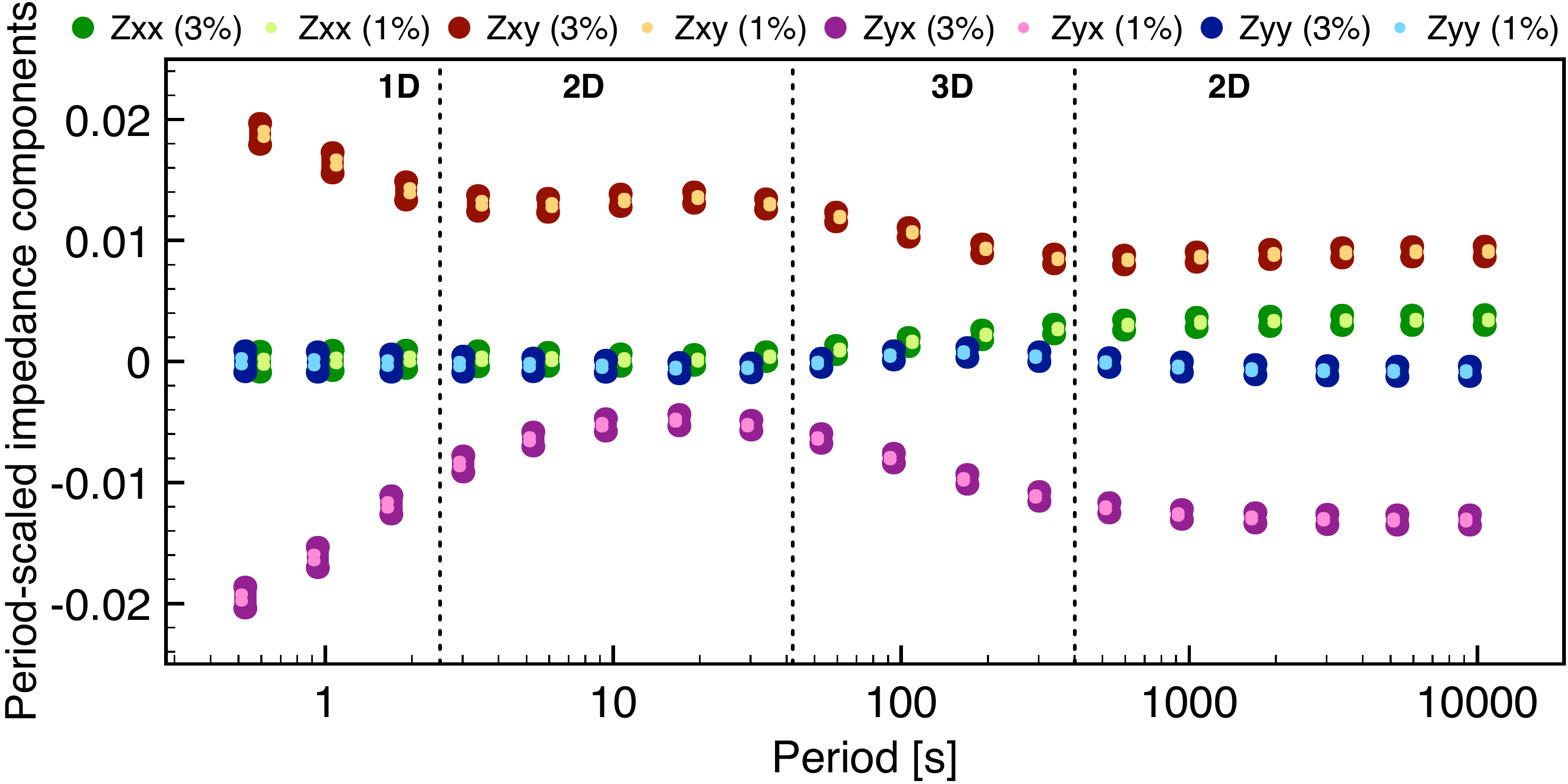}\\
		b) Period-scaled imaginary part of impedance\\
		\noindent\includegraphics[width=.58\textwidth]{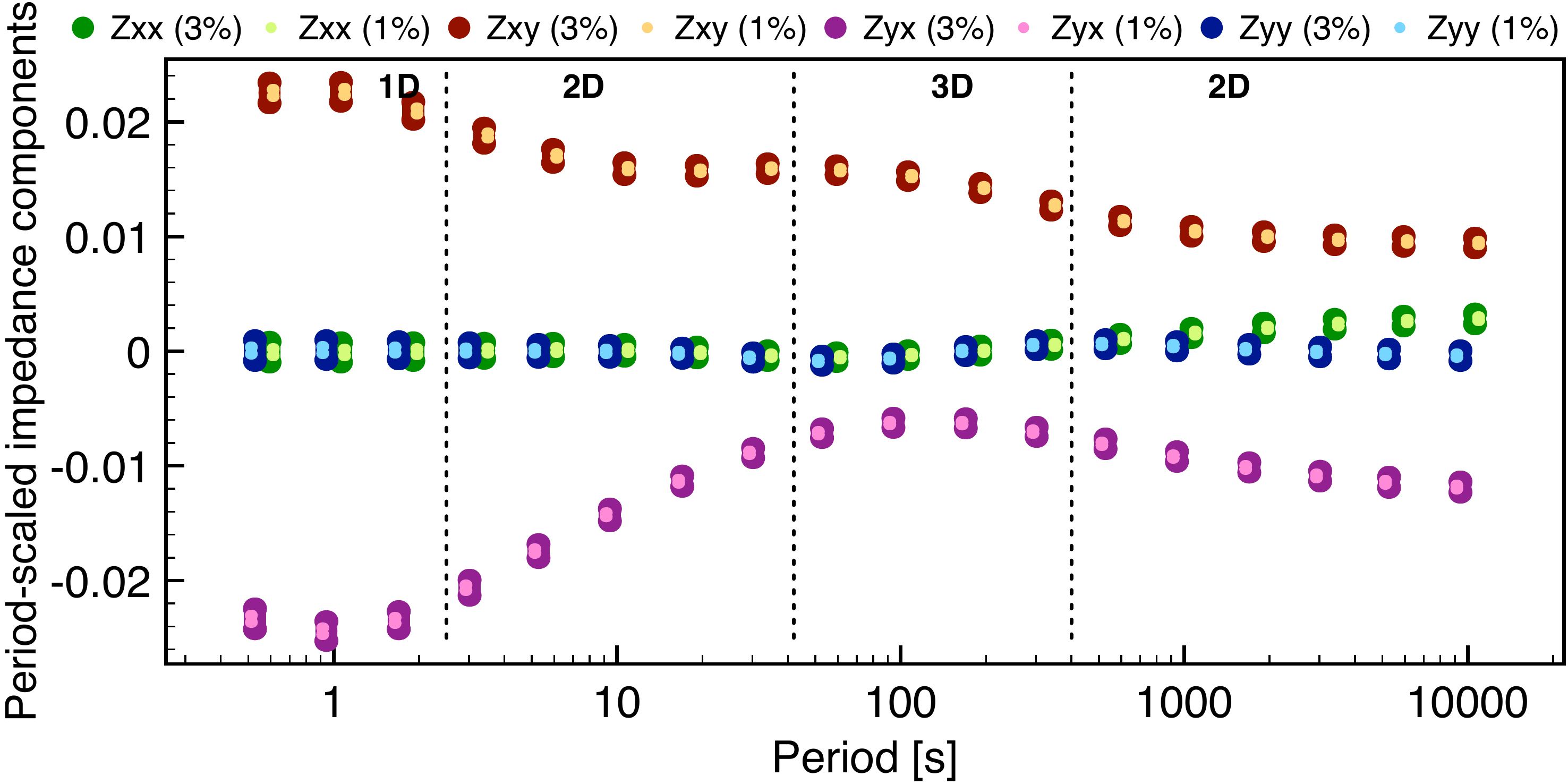}
	\caption{Impedance data of site \emph{A09} of the DSM1 \citep{Miensopust:2013}. A confidence limit of $1\%$ and $3\%$ was added to the impedance based on the impedance determinants square root. Periods of different data are slightly shifted for comparison's sake.}
	\label{fig:exampleImpedances}
\end{figure}
\begin{figure}[p]
		\centering
		a) Normalised skews\\
		\noindent\includegraphics[width=.58\textwidth]{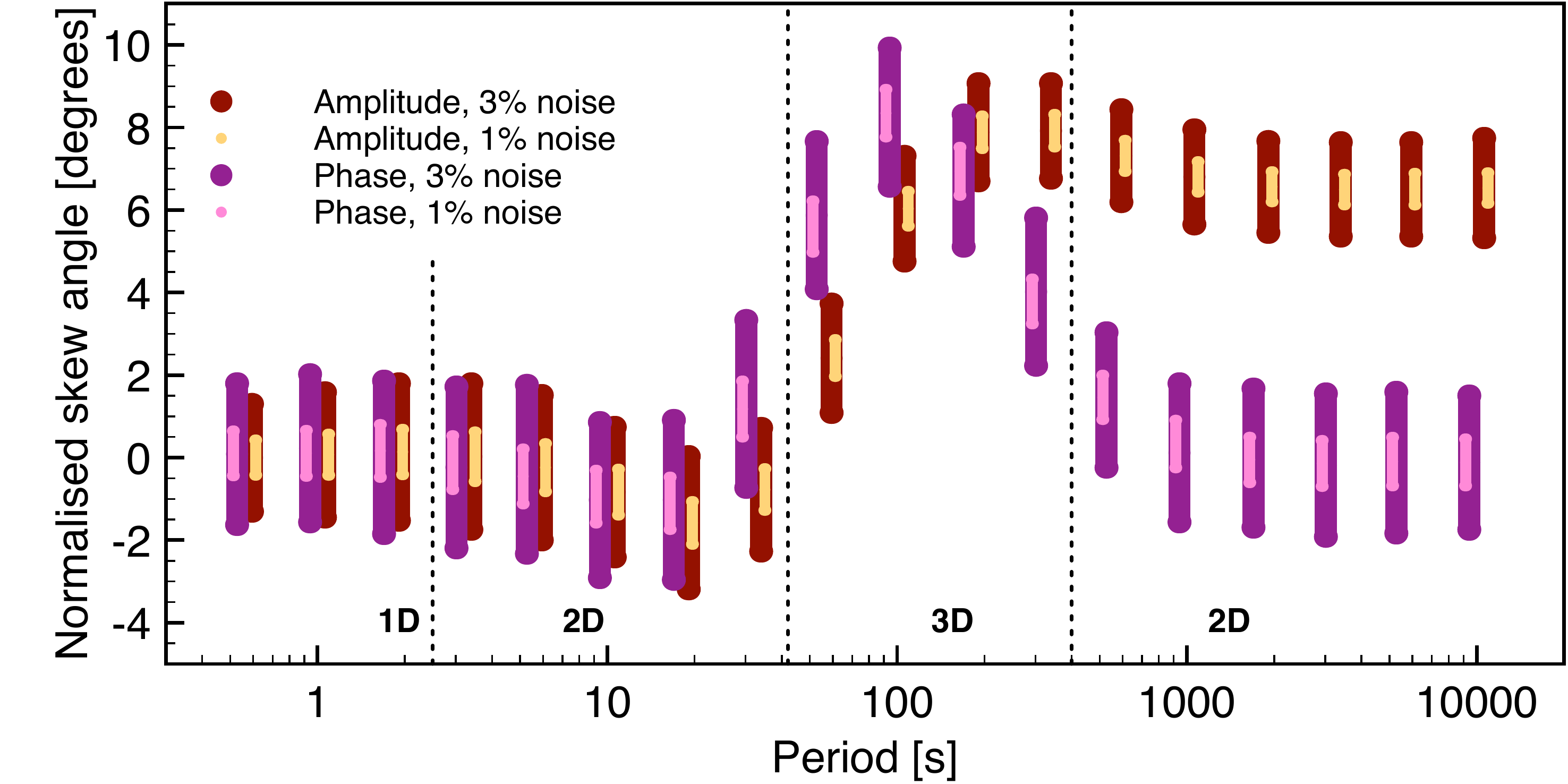}\\
		b) Strike angles\\
		\noindent\includegraphics[width=.58\textwidth]{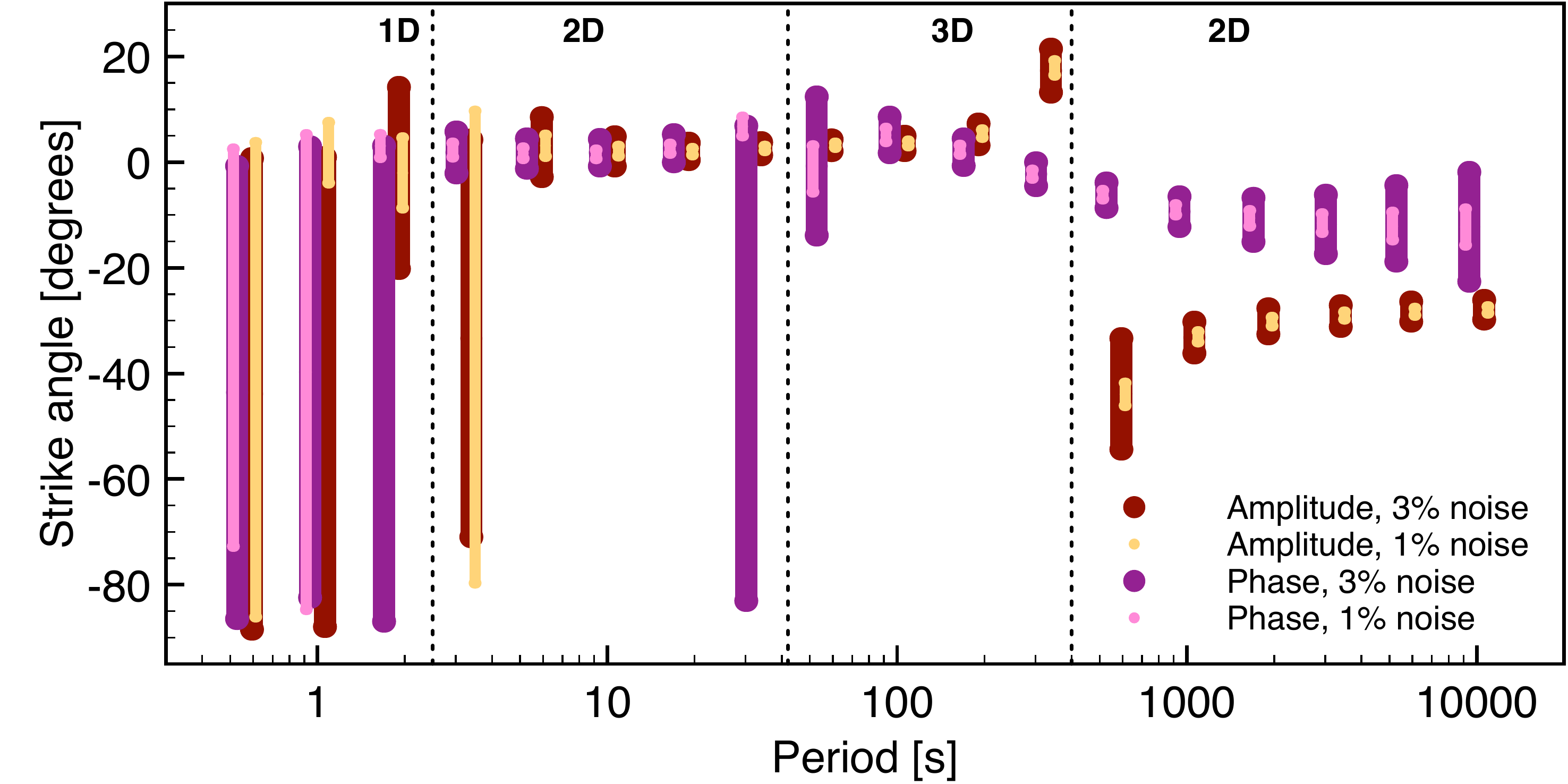}\\
		c) Period-scaled Amplitude Tensor singular values\\
		\includegraphics[width=.58\textwidth]{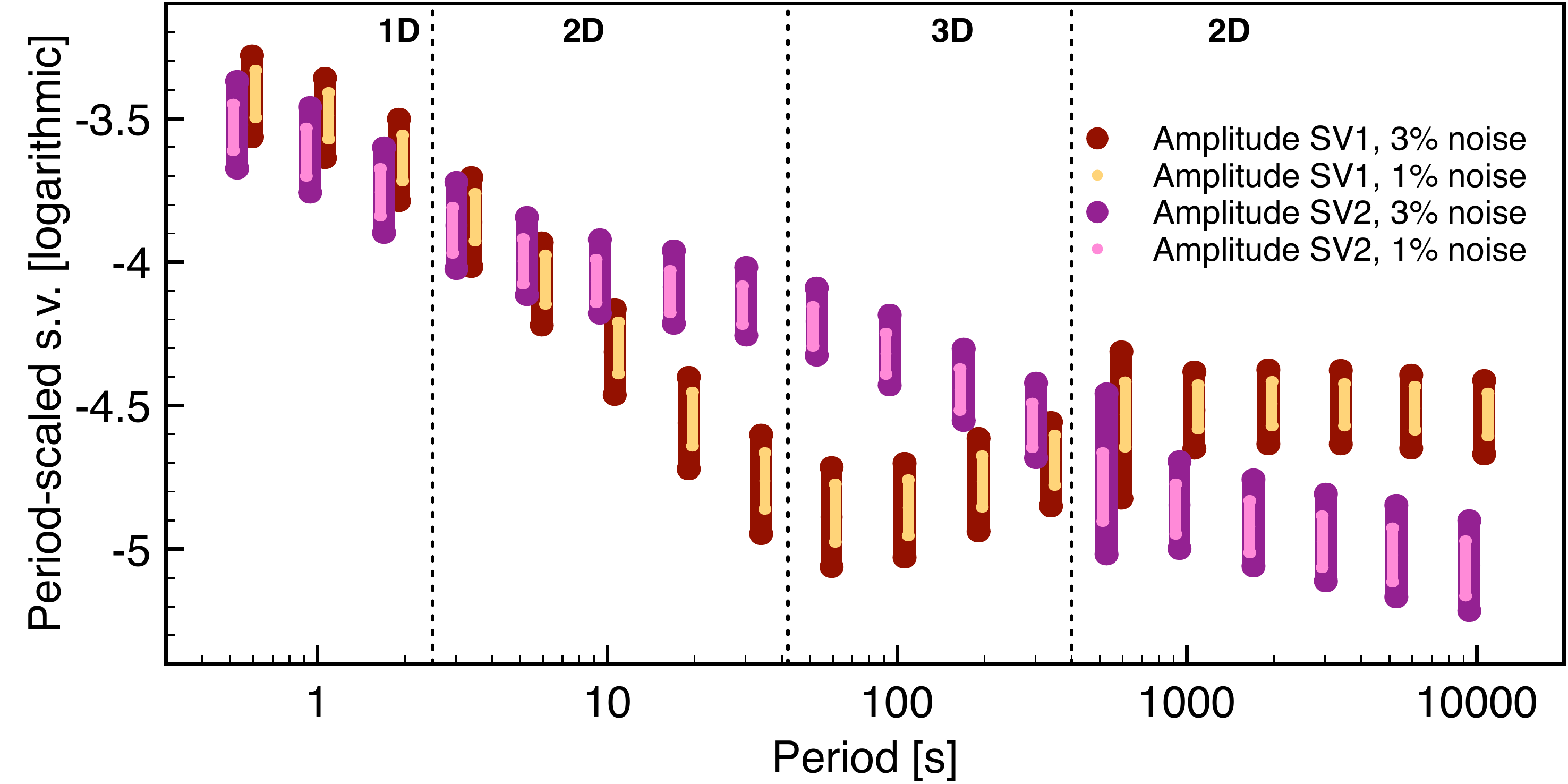}\\
		d) Phase Tensor singular values\\
		\includegraphics[width=.58\textwidth]{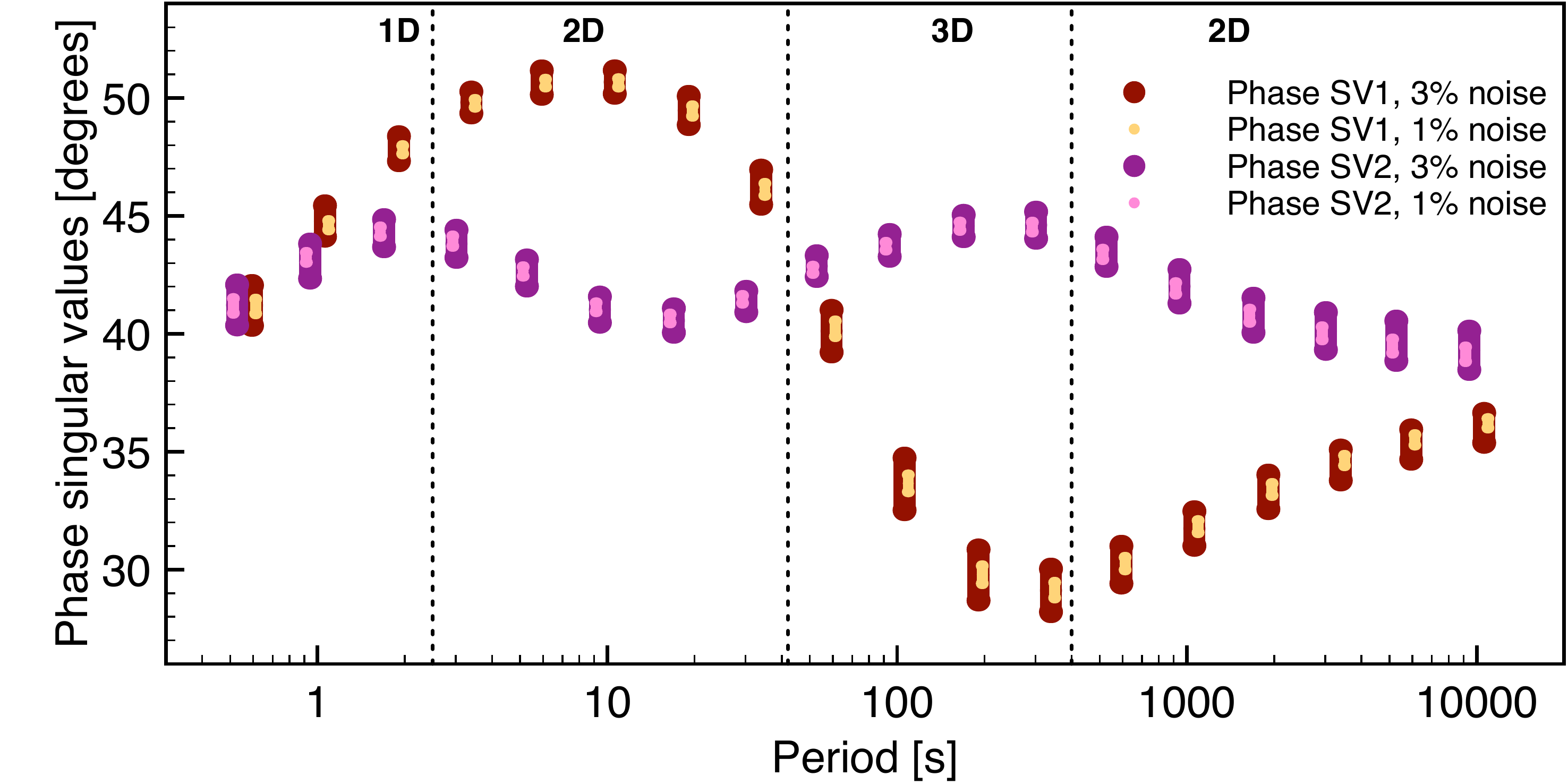}	
	\caption{Impedance data of site \emph{A09} of the DSM1 \citep{Miensopust:2013} are decomposed into Phase and Amplitude Tensor skew angles, strike angles and singular values. Note that the amplitude skew angle is normalised to $\psi_P^{\mathrm{norm}}=90^{\circ}-\psi_P$ (refer to text for details). To illustrate the statistical stability of the decomposition, a confidence limit of $1\%$ and $3\%$ was added to the impedance based on the impedance determinants square root. Periods of different data are slightly shifted for comparison's sake.}
	\label{fig:exampleAmpPhaseParam}
\end{figure}
In the following, we analyse the Amplitude and Phase Tensor representations for different dimensionality cases, considering an isotropic subsurface in 1D, 2D, 3D and in the special case of a reduction of dimensionality with increasing depth, i.e.~a 2D structure beneath a 3D anomaly. To illustrate our analysis we use a data set from the Dublin 3D Modelling and Inversion Workshops \citep{Miensopust:2013}. These workshops provided a number of freely available, large synthetic datasets, which are optimal for testing algorithms and ideas, not only restricted to modelling and inversion. We consider data of 
site $A09$ of the Dublin Secret Model One (DSM1). We chose DSM1, because the data are: (i) noise-free, (ii) distortion-free and (iii) contain clearly a 3D subsurface model with (iv) a homogenous near surface. In this data set, site $A09$ was chosen, because it is situated on top of two differently oriented 2D structures at different depths in an homogenous background, therefore this site is truly 3D but still allows to investigate the effects that correspond to these features. The model is represented in Figure \ref{fig:DSM1model}, where the position of site $A09$ is indicated. 
The 
impedance data of site $A09$ are plotted in Figure \ref{fig:exampleImpedances} and have been parameterised according to \eqref{eq:TensorParameterBooker} into the parameters: amplitude and phase skew angles, strike angles and singular values, which are plotted in Figure \ref{fig:exampleAmpPhaseParam}. We assumed confidence limits for the impedance tensor of $1\%$ and $3\%$ of the square root of its determinant values at each period to show the effect of impedance uncertainty on the tensor parameters resulting from the decomposition. In the decomposition algorithm, the impedance tensor components have been sampled assuming a gaussian distribution with the standard deviation according to the confidence intervals, so that each impedance sample decomposes into a sample for the Amplitude and Phase Tensor parameters. We use the median and the median absolute deviation to the median of the parameter distributions to provide robust estimates of mean and standard deviation of the decomposition parameters, respectively, which improves real data estimates when the exact noise distributions are unknown. 

\subsubsection{One-dimensional Earth Case}
In the 1D case, the MT impedance tensor and the Phase Tensor simplify to:
 \begin{equation}
\mathbf{Z}_{1D}=
 \begin{pmatrix}
0	&z_{1D}      \\
-z_{1D} 	&0
\end{pmatrix},\quad\mathrm{and}\quad
\mathbf{\Phi}_{1D}=
 \begin{pmatrix}
 \phi_{1D}	&0      \\
0	&\phi_{1D}
\end{pmatrix}.
\end{equation} 
Denoting the conjugate transpose of a matrix $\mathbf{M}$ by $\mathbf{M}^H$ and considering that $e(\mathbf{\Phi}_{1D})(e(\mathbf{\Phi}_{1D}))^{H}=\mathbf{\mathbb{I}}$, it follows from \eqref{eq:ampTen} that $\mathbf{Z}_{1D}\mathbf{Z}_{1D}^H=\mathbf{P}_{1D}\mathbf{P}_{1D}^T$. Then, considering the $1D$ impedance given by $z_{1D}=|z_{1D}|\exp(i\arctan\phi_{1D})$, where $\rho_{1D}=|z_{1D}|$, the Amplitude Tensor can be expressed by the apparent resistivities $\rho_{1D}^{\mathrm{app}}$:
 \begin{equation}
\mathbf{P}_{1D}=
 \begin{pmatrix}
0	&\rho_{1D}      \\
-\rho_{1D} 	&0
\end{pmatrix}=
	\begin{pmatrix}
		0&\sqrt{\omega\mu_0\rho_{1D}^{\mathrm{app}}}\\
		-\sqrt{\omega\mu_0\rho_{1D}^{\mathrm{app}}}&0
	\end{pmatrix},
\end{equation} 
with angular frequency $\omega=2\pi/T$, period $T$ and magnetic permeability of free space $\mu_0$. On the other hand, it can be seen from \eqref{psi_X} that, in the 1D case, the skew angle is $\psi_P=\frac{\pi}{2}$ at all frequencies. This corresponds to $\psi_\Phi=0$, because the 1D Amplitude Tensor is anti-diagonal like the impedance. Then, for simplicity, we define the normalised amplitude skew angle as:
 \begin{equation}
	\psi_P^{\mathrm{norm}}=\frac{\pi}{2}-\psi_P,
\end{equation} 
so that it equals zero for 1D and 2D situations in analogy to the phase skew angle. In other words, the Amplitude Tensor can be represented by a unity matrix rotated by $\psi_P=\frac{\pi}{2}$ and multiplied by the scalar $|z_{1D}|$ and therefore, the singular values
$\rho_1$ and $\rho_2$ are equal. 
In Figure \ref{fig:exampleAmpPhaseParam}, the Amplitude Tensor parameters at the three shortest periods exemplify the properties we have discussed for the 1D case. 

\subsubsection{Two-dimensional Earth Case}
In the 2D case, the MT impedance tensor and the Phase Tensor simplify to:
 \begin{equation}
\mathbf{Z}_{2D}=
 \begin{pmatrix}
0	&z_{TE}      \\
-z_{TM} 	&0
\end{pmatrix}\quad\mathrm{and}\quad
\mathbf{\Phi}_{2D}=
 \begin{pmatrix}
 \phi_{TE}	&0      \\
0	&\phi_{TM}
\end{pmatrix}.
\end{equation} 
Besides, considering $e(\mathbf{\Phi}_{2D})(e(\mathbf{\Phi}_{2D}))^{H}=\mathbf{\mathbb{I}}$, it follows from \eqref{eq:ampTen} that $\mathbf{Z}_{2D}\mathbf{Z}_{2D}^H=\mathbf{P}_{2D}\mathbf{P}_{2D}^T$. Then, using that the 2D impedances for the transverse electric (TE) and transverse magnetic (TM) modes are given by $z_{m}=|z_{m}|\exp(i\arctan\phi_{m})$, with $\rho_{m}=|z_{m}|$ and $m$ denoting either the TE or TM mode, the Amplitude Tensor can be expressed by the apparent resistivities $\rho_{m}^{\mathrm{app}}$:
 \begin{equation}
\mathbf{P}_{2D}=
 \begin{pmatrix}
0	&\rho_{TE}      \\
-\rho_{TM} 	&0
\end{pmatrix}=
	\begin{pmatrix}
		0&\sqrt{\omega\mu_0\rho_{TE}^{\mathrm{app}}}\\
		-\sqrt{\omega\mu_0\rho_{TM}^{\mathrm{app}}}&0
	\end{pmatrix}.
\end{equation} 
Again, as in the 1D case, it can be seen from \eqref{psi_X} that the skew angle is $\psi_P=\frac{\pi}{2}$. It is also easy to see that the strike angle parameters of the Amplitude and Phase Tensors must be equal in this case, since an arbitrary oriented 2D impedance, $\mathbf{Z}_{\mathrm{2D}, \mathrm{rot}}$, can always be expressed as a product of its local 2D form, $\mathbf{Z}_\mathrm{2D}$, and a coordinate rotation by the strike angle $\theta_z$: $\mathbf{Z}_{\mathrm{2D},\mathrm{rot}}=\mathbf{R}(\theta_z)\mathbf{Z}_\mathrm{2D}\mathbf{R}(-\theta_z)$. Then, equalling the diagonal, local 2D form $\mathbf{Z}_\mathrm{2D}$ with \eqref{eq:Z2APdecomposed} 
 implies that the strike angle difference between Amplitude and Phase Tensor is $\theta_P-\theta_\Phi=0$ 
for 2D subsurfaces. Further, note that the singular values of Phase and Amplitude Tensors equal the impedance phase and amplitude values, respectively, of the TE and TM modes, similarly to the 1D case. 
In Figure \ref{fig:exampleAmpPhaseParam} the Amplitude Tensor parameters at the period range from $2.5$ to $40\,\mathrm{s}$ exemplify the discussed properties for a quasi 2D case. 
Note that in this work and in analogy to the interpretation of the Phase Tensor skew angle of \cite{Booker:2014}, we consider that a subsurface is 2D, if the confidence interval of the normalised amplitude skew angle lies between $-6^\circ$ and $6^\circ$.

\subsubsection{Three-dimensional Earth Case}
The change over period of all Amplitude Tensor parameters is the signature of a 3D subsurface. In Figure \ref{fig:exampleAmpPhaseParam}, the period range from $40$ to $400\,\mathrm{s}$ exemplifies the 3D Amplitude Tensor parameters. If some parameters are constant but the Amplitude Tensor cannot be categorised as either 1D or 2D by the criteria we discussed above, it indicates a situation that is discussed on a dedicated example in the following section. 

Note that in 3D situations the singular values of the Amplitude Tensor can also be defined as apparent resistivities but interpreting them in terms of subsurface structure is not as straight forward as for 1D or 2D cases because interpretation must include the amplitude skew and strike angles.

\subsubsection{High-dimensional Over Low-dimensional Subsurface}
\label{sec:3D2Dexample}
We refer to a high-dimensional subsurface when its dimensionality can be categorized as 2D or 3D and to a low-dimensional subsurface when it is 2D, 1D or an homogenous half-space. Note that the 2D case is present in both categories. The reason is to include situations in which a shallow 2D subsurface structure with a given strike orientation changes into a deeper 2D structure with a different strike angle. 

For clarity of exposition of this dimensionality case, let us take and discuss more extensively the impedance data of site $A09$, which we have already used to illustrate the other dimensionality cases  in Figure \ref{fig:exampleAmpPhaseParam}. 
First, for periods from $0.5\,\mathrm{s}$ to $2\,\mathrm{s}$, the model appears quite 1D, because the normalised amplitude and phase skew angles are both around zero. In the next few periods ($5-25\,\mathrm{s}$), both strike angles are established at
constant $0^\circ$ and the singular values of Amplitude and Phase begin to
spread out, which indicates the presence of a 2D structure with the
strike direction along East-West (or North-South due to the inherent $90^\circ$ strike ambiguity). At longer periods ($>25\,\mathrm{s}$), the
signal penetrates deeper and covers a larger volume recognising the three
dimensionality of the model, which results in large and different
values for the skew angles of Amplitude and Phase Tensors and in a differently
varying strike angle. 
At very long periods ($>1,\!500\,\mathrm{s}$), the phase skew angle has returned to a value below $6^\circ$, indicating that the response at those periods is quasi-two-dimensional. However, the deep phase strike angle is considerable different from the shallow phase strike angle, so that the strike angle cannot be assumed constant for the entire period range (according to the Phase Tensor) and, therefore, the data set should be considered 3D when interpreting all periods. The negligible phase skew angle at these periods indicates 
that the shallow anomaly (that caused the initial strike angle) acts as a galvanic distorter and is, at the longest periods, only present in the real-valued Amplitude Tensor. For this reason, the amplitude skew angle is not negligible but constant and represents the galvanic electric field distortion.

Based on the observations from this example, we can derive some general statements for the special situation of a shallow structure that is dimensionally incompatible (i.e.~3D or 2D with different strike angle) with the deeper low-dimensional structures. We state that in this special case, the Amplitude Tensor adopts a specific form at long periods, which is associated to the dimensionality of the deep structure. 
If the deep structure is 2D then, the amplitude strike angle, $\theta_P$, and the amplitude skew angle, $\psi_P$, are constant and at least one singular value, $\rho_1$ or $\rho_2$, is variable.
On the other hand, if the deep structure can be assumed 1D then, the behaviour of the Amplitude Tensor parameters will be the same as in the 2D case with the only difference of a constant value for the ratio $\rho_1/\rho_2$. Finally, if the deep structure can be assumed an homogeneous half-space, then the only difference to the previous cases is that the Amplitude Tensor singular values remain constant at their constant offset. 
Figure \ref{fig:exampleAmpPhaseParam} exemplifies, at periods from $1,\!500$ to $10,\!000\,\mathrm{s}$, the Amplitude Tensor parameters for the situation of a seemingly 2D deep structure under a differently oriented shallow 2D structure.

\section{Insights and Applications of the Amplitude Tensor}
\label{sec:IndGalAmp}
\subsection{Galvanic and Inductive Amplitude Tensors}
\label{sec:IndGalAmpTen}
The Amplitude Tensor contains the amplitude information of both, the electromagnetic induction processes and galvanic effects, since, as we have demonstrated, it contains all amplitude information of the MT impedance tensor. Further, it has been shown that, under reasonable assumptions \citep{groom:1992,Chave:1994}, for any structure whose response has reached the galvanic limit, the galvanic distortion of the electric field can be represented by a real matrix pre-multiplied to the regional impedance tensor: $\mathbf{Z}_d=\mathbf{CZ}^{\mathrm{reg}}$. Then, by construction and as can be seen in \eqref{eq:Z2APdecomposed}, the electrically distorted Amplitude Tensor must be $\mathbf{P}_d=\mathbf{CP}^{\mathrm{reg}}$. 
Given that the Amplitude Tensor contains galvanic and inductive information, we hypothesise that it can be separated generally in a product between two tensors representing the galvanic, $\mathbf{P}^{\mathrm{gal}}$, and the inductive, $\mathbf{P}^{\mathrm{ind}}$, amplitude contributions:
 \begin{equation}
\mathbf{P}=\mathbf{P}^{\mathrm{gal}}\mathbf{P}^{\mathrm{ind}}.
\label{eq:AmpGalInd}
\end{equation} 
The reason for the conjecture of this hypothesis is that the removal of the galvanic distortion $\mathbf{C}$ from the distorted impedance $\mathbf{Z}_d$ leaves the regional inductive response $\mathbf{Z}_{\mathrm{reg}}$ of the subsurface at any given period. Applying this logic to the Amplitude Tensor and period-wise associating $\mathbf{C}$ with $\mathbf{P}^{\mathrm{gal}}$ and $\mathbf{P}^{\mathrm{reg}}$ with $\mathbf{P}^{\mathrm{ind}}$ justifies the construction of \eqref{eq:AmpGalInd} and our proposed hypothesis. Note that $\mathbf{C}$ and $\mathbf{P}^{\mathrm{gal}}$ as well as $\mathbf{P}^{\mathrm{reg}}$ and $\mathbf{P}^{\mathrm{ind}}$ can only be associated period-wise and are not generally equal when considering an extended period range at once, because $\mathbf{P}^{\mathrm{gal}}$ is period dependent. 
For example, if longer period data is used to determine $\mathbf{C}$ and to correct shorter period impedance data, an overcorrection can occur that introduces systematic errors in the subsequent data interpretation.

\subsection{Approximating the Inductive Amplitude Tensor with the Phase Tensor}
\label{sec:ApproxIndAmp}
The inductive Amplitude Tensor is, by definition, devoid of electric galvanic information and represents the regional inductive amplitude response of the subsurface just like the Phase Tensor represents the phase response. 
This means in particular that both, the inductive Amplitude and Phase Tensors, respond to the same subsurface volume that is sensed by induction at each given period and therefore, their observed geometric parameters, strike angle, skew angle and macroscopic anisotropy, must represent the same subsurface conductivity distribution and geometry. The angular tensor parameters, strike and skew, are directly provided by the parameterisation \eqref{eq:TensorParameterBooker} and in appendix \ref{sec:logAniso} we motivate and define the measure of macroscopic anisotropy for the Amplitude and Phase Tensors as function of their respective singular values. 

Based on the idea that inductive Amplitude and Phase Tensors represent the inductive response of the same subsurface volume, we propose to approximate the parameters, strike angle, $\theta^{\mathrm{ind}}_P$, skew angle, $\psi^{\mathrm{ind}}_P$ and macroscopic anisotropy, $\rho^{\mathrm{ind}}_a$ (see appendix \ref{sec:logAniso}), of the inductive Amplitude Tensor with the respective values obtained from the Phase Tensor:
 \begin{equation}
  \widetilde{\mathbf{P}}^{\mathrm{ind}}\approx \mathbf{R}(-\theta_\Phi)\begin{pmatrix}\exp{(\ln\rho+\phi_a)}&0\\0&\exp{(\ln\rho-\phi_a)}\end{pmatrix}
  \mathbf{R}(\pi/2-\psi_\Phi)\mathbf{R}(\theta_\Phi),
  \label{eq:AmpInd}
\end{equation} 
with $\rho=\sqrt{\rho_1\rho_2}$. This proposition allows an approximate solution for the galvanic Amplitude Tensor with the decomposition given by \eqref{eq:AmpGalInd}:
 \begin{equation}
  \widetilde{\mathbf{P}}^{\mathrm{gal}}\approx \mathbf{P} \left[\widetilde{\mathbf{P}}^{\mathrm{ind}}\right]^{-1}.
  \label{eq:AmpGal}
\end{equation} 
The motivation for this proposition is that the galvanic Amplitude Tensor equals the galvanic electric distortion at the shortest measured periods. Therefore, the estimation of the galvanic Amplitude Tensor allows us to recover distortion parameters without the traditional assumptions on the regional impedance dimensionality in contrast to the work by \cite{Larsen:1977} and  \cite{Chave:1994}. This is particularly advantageous for the recovery of the anisotropic distortion parameter for which, until now, it was necessary to assume a regional 1D subsurface. 

In the following, we assess the performance of the proposition \eqref{eq:AmpInd} on an explicit example in which we estimate the electric galvanic anisotropic distortion assuming a regional 2D impedance. 
This 2D structural assumption of the impedance is necessary to derive an explicit formula for the anisotropic distortion parameter in appendix \ref{sec:logAniso} but is not inherent in \eqref{eq:AmpInd} and \eqref{eq:AmpGal}. We refer to a companion paper \citep{Neukirch:2017b} for an algorithm that recovers quantitatively, using \eqref{eq:AmpInd} and \eqref{eq:AmpGal}, the galvanic electric distortion parameters, including the anisotropic parameter and without any assumption on the dimensionality of the subsurface structure, which is, however, beyond the scope of this work due to the amount of required technical detail.

\subsubsection{Example: Recovering Anisotropic Distortion for 2D Impedances}
\label{sec:RecAniDis2D}
\begin{figure}[p]
		\centering
		a) Recovered anisotropy from the galvanic amplitude estimate with period-wise tensor singular values\\
		\includegraphics[width=.61\textwidth]{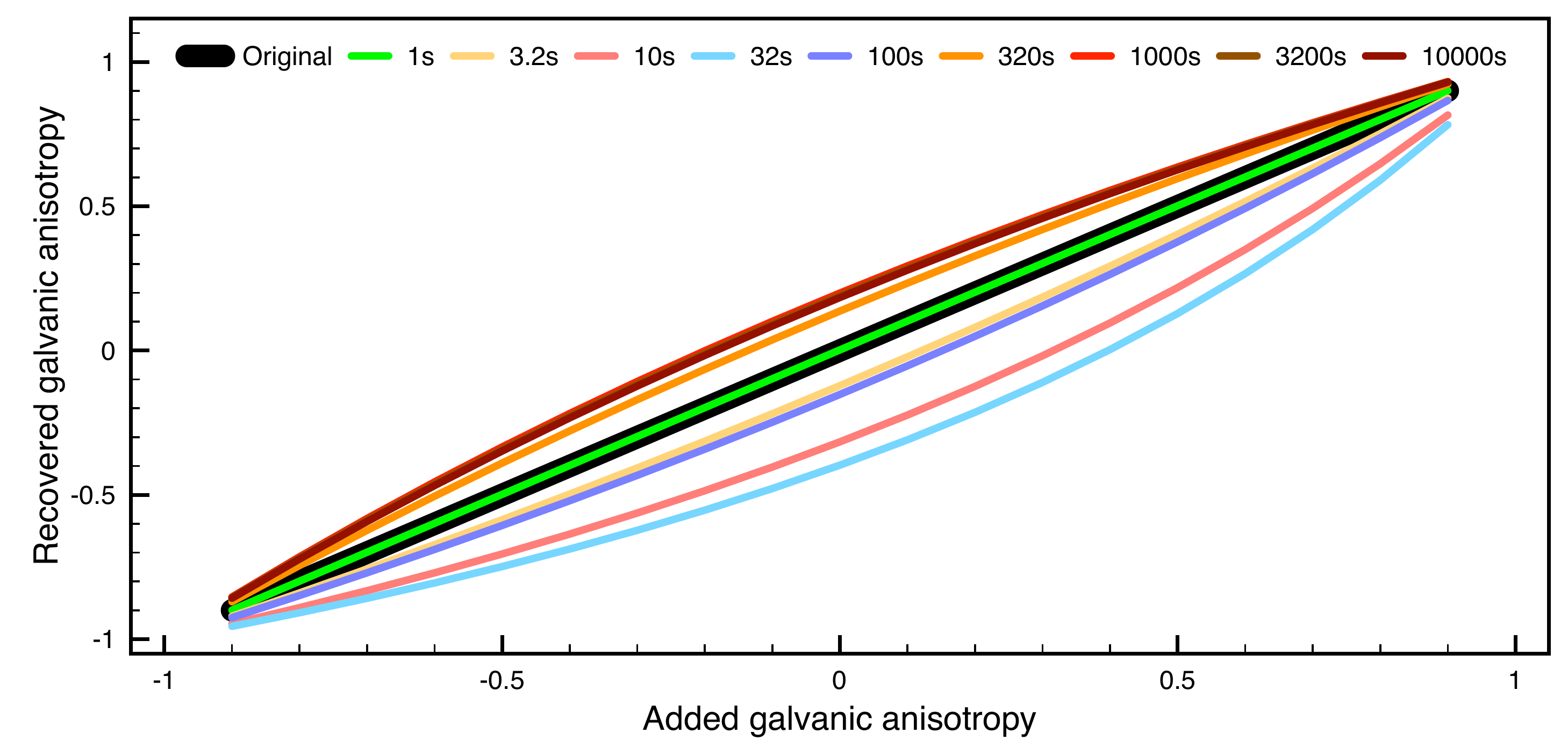}\\
		b) Recovered anisotropy from the galvanic amplitude estimate with two decade window average\\
		\includegraphics[width=.61\textwidth]{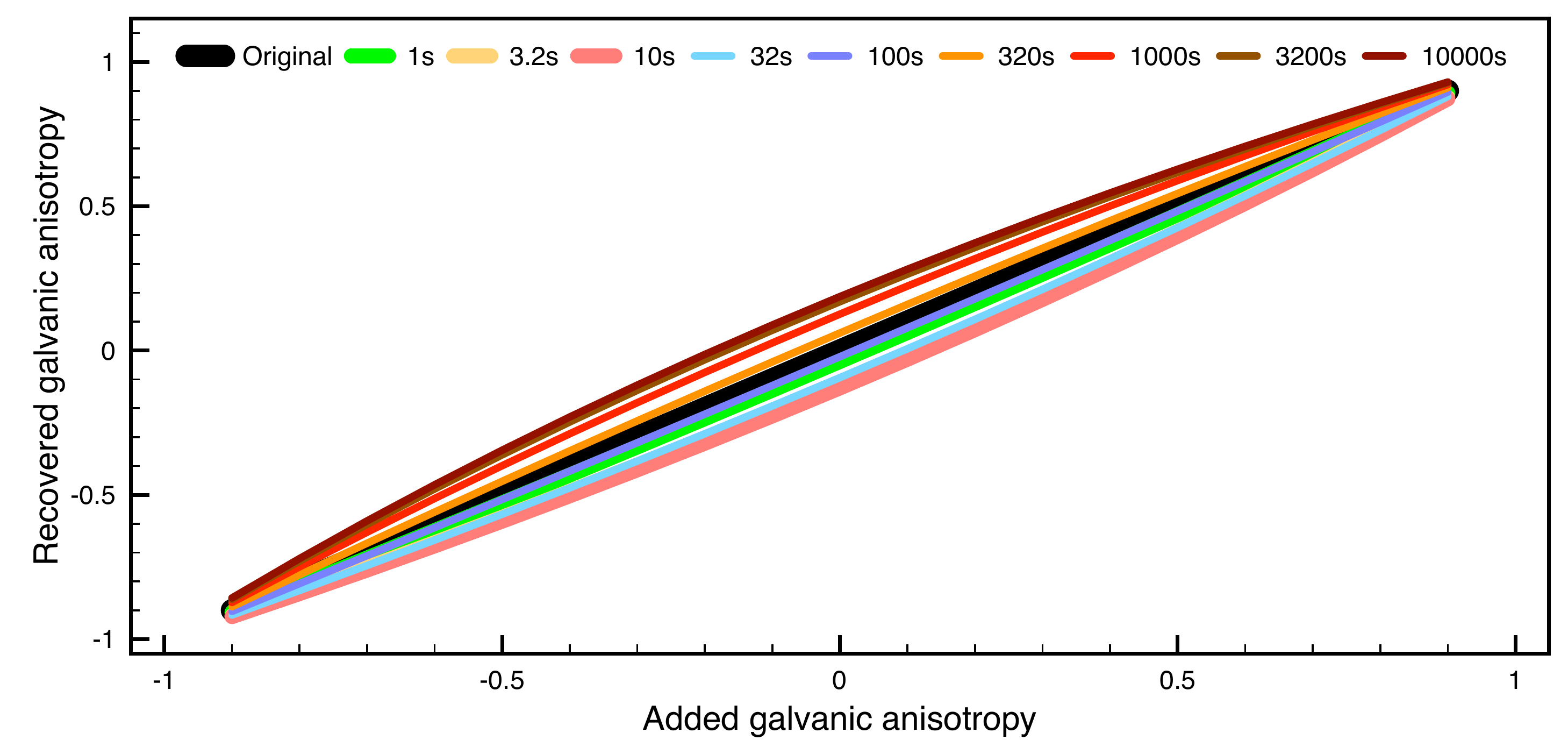}\\
		\centering
		c) Difference to added anisotropy (period-wise)\\
		\includegraphics[width=.61\textwidth]{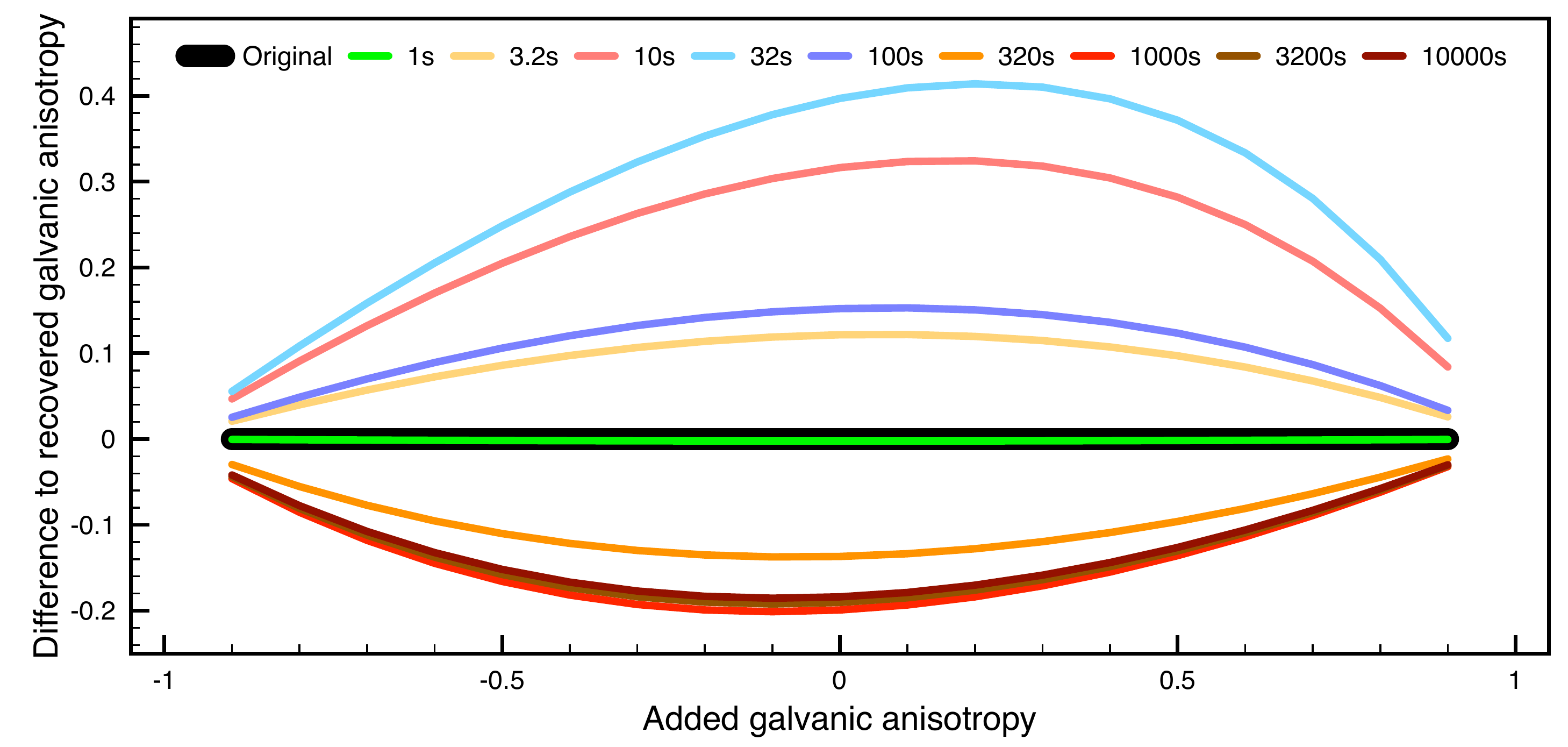}\\
		d) Difference to added anisotropy (two decade window average)\\
		\includegraphics[width=.61\textwidth]{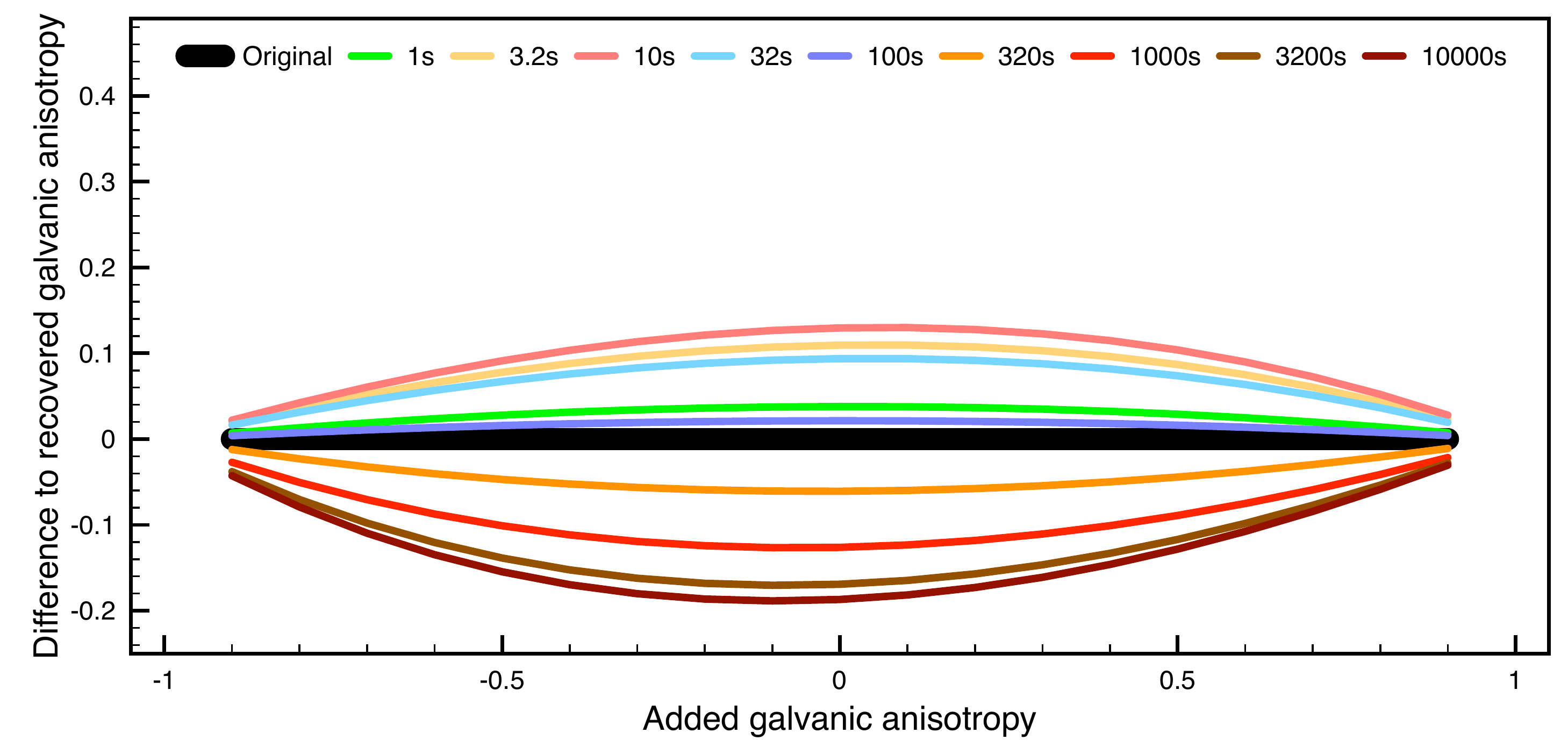}
	\caption{Recovery of the anisotropy parameter for different periods (colours) using two methods: period-wise anisotropy (a,c) and the moving average anisotropy with a two-decade gaussian window (b,d). Greenish (reddish/blueish) colour denotes a 1D (2D/3D) regime of the \emph{A09} data. Spectral averaging with the gaussian moving window of the tensor anisotropy parameters improves the galvanic amplitude estimate (b) and therewith reduces the bias of the recovered galvanic amplitude parameter (d).}
	\label{fig:galvanicAnisotropy}
\end{figure}
In state-of-the-art distortion analysis methods \citep{Groom:1989,Chave:1994,McNeice:2001,Bibby:2005}, anisotropy distortion is considered indeterminable for data that is not 1D isotropic on a sufficiently long period range. In this example, we want to show that, with the macroscopic anisotropy parameters we propose in appendix \ref{sec:logAniso}, it is possible to estimate anisotropic distortion as a first approach to overcome present difficulties in 2D impedance data interpretation. Estimating the anisotropic distortion with this methodology, we also implicitly show the validity of the approximation of $\mathbf{P}^\mathrm{ind}$ and $\mathbf{P}^\mathrm{gal}$ in \eqref{eq:AmpInd} and \eqref{eq:AmpGal}. 

We generate synthetic distorted impedance data, $\mathbf{Z}_d$, premultiplying an anisotropic distortion matrix $\mathbf{A}\in\mathbb{R}^{2\times 2}$ \citep{Groom:1989}, with an anisotropy parameter $a$ ranging from $-0.9$ to $0.9$, to the regional inductive impedance data, $\mathbf{Z}$, of the previous example (site $A09$, Figures \ref{fig:exampleImpedances} and \ref{fig:exampleAmpPhaseParam}):
 \begin{equation}
	\mathbf{Z}_d=\mathbf{AZ}=
	\begin{pmatrix}
		\frac{1-a}{\sqrt{1-a^2}}&0\\
		0&\frac{1+a}{\sqrt{1-a^2}}
	\end{pmatrix}
	\begin{pmatrix}
		Z_{xx}&Z_{xy}\\
		Z_{yx}&Z_{yy}
	\end{pmatrix}.
	\label{eq:DistImp}
\end{equation} 
We consider the impedance tensor data at each period as separate data in order to test 1D, 2D and 3D data (cf.~dimensionality of period sections with Figures \ref{fig:exampleImpedances} and \ref{fig:exampleAmpPhaseParam}).  
Besides, from the Amplitude-Phase decomposition of the distorted impedance, $\mathbf{Z}_d(\omega)$, into Amplitude and Phase Tensors and \eqref{eq:AmpGalInd}: 
 \begin{equation}
\mathbf{Z}_d(\omega)=\mathbf{P}_d(\omega)e(\mathbf{\Phi(\omega)})=\mathbf{P}^\mathrm{gal}(\omega)\mathbf{P}^\mathrm{ind}(\omega)e(\mathbf{\Phi}(\omega)).
	\label{eq:DistImpDeco}
\end{equation} 
Thus, from the comparison of \eqref{eq:DistImp} and \eqref{eq:DistImpDeco}, we need to compute $\mathbf{P}^\mathrm{gal}$ of the distorted data in order to show that we can recover the anisotropy parameter $a$, since $\mathbf{Z}$ denotes the regional inductive impedance. Then, we decompose the data $\mathbf{Z}_d(\omega)$ and we parameterise the resulting $\mathbf{\Phi}(\omega)$ and $\mathbf{P}_d(\omega)$ following \eqref{eq:TensorParameterBooker}. With the singular values of $\mathbf{\Phi}(\omega)$ and $\mathbf{P}_d(\omega)$ we compute the galvanic amplitude singular values, $\rho_1^{\mathrm{gal}}(\omega)$ and $\rho_2^{\mathrm{gal}}(\omega)$, using (see Appendix B):
 \begin{equation}
  \rho_1^{\mathrm{gal}}(\omega)=\exp\left(\rho_a(\omega)-\phi_a(\omega)\right)
  \qquad\mathrm{and}\qquad
  \rho_2^{\mathrm{gal}}(\omega)=\exp\left(-\left(\rho_a(\omega)-\phi_a(\omega)\right)\right),
  \label{eq:galvanicAmplitude}
\end{equation} 
where $\rho_a$ and $\phi_a$ are the macroscopic anisotropy parameters as defined by \eqref{eq:rPArAA} in appendix \ref{sec:logAniso}, and where we assumed a 2D impedance as detailed in appendix \ref{sec:AmpSV}. 
Then, identifying each element of $\mathbf{A}$ with each element of $\mathbf{P}^{\mathrm{gal}}(\omega)$ and using that our impedance data are at a single period, the recovered anisotropic distortion parameter, $a_r(\omega)$, for 2D impedances is given by:
 \begin{equation}
a_r(\omega)=-\frac{\left(\rho_1^{\mathrm{gal}}(\omega)\right)^2-1}{\left(\rho_1^{\mathrm{gal}}(\omega)\right)^2+1}=\frac{\left(\rho_2^{\mathrm{gal}}(\omega)\right)^2-1}{\left(\rho_2^{\mathrm{gal}}(\omega)\right)^2+1}.
\end{equation} 
The comparison of $a$ with $a_r(\omega)$ is illustrated in Figure \ref{fig:galvanicAnisotropy} at different periods and demonstrates that the estimation is qualitatively correct with a maximum bias of $0.2$ for small anisotropy values and when considering averaged anisotropy values (see appendix \ref{sec:logAniso} for more details). 
Moreover, the bias between $a_r$ and $a$ is similar for different periods, and in particular, the periods corresponding to 3D data, $32\,\mathrm{s}$ and $100\,\mathrm{s}$, show slightly smaller bias than the period of $1\,\mathrm{s}$ associated to 1D data. Thus, surprisingly, it seems that the bias is unrelated to the data dimensionality, although we assumed a 2D impedance to derive the explicit expressions for the galvanic amplitude singular values in \eqref{eq:galvanicAmplitude}.

We evaluate if the obtained values of bias are significant compared to typical measurement errors. First, we calculate the galvanic amplitude singular values of the original distortion matrix $\mathbf{A}$, which are directly the diagonal elements of $\mathbf{A}$. Then, we divide them by the corresponding estimated galvanic amplitude singular values of 
$\mathbf{P}^{\mathrm{gal}}$, $\rho_1^{\mathrm{gal}}$ and $\rho_2^{\mathrm{gal}}$, to obtain the systematic error, $\Delta_{\rho_1^{\mathrm{gal}}}$ and $\Delta_{\rho_2^{\mathrm{gal}}}$, introduced by approximating $\mathbf{P}^{\mathrm{ind}}$ with the Phase Tensor parameters as in \eqref{eq:AmpInd}.
This error measure indicates the magnitude of anisotropic galvanic distortion that remains in the data and is represented by a relative error of the amplitude. 
Instead of investigating this relative error, which is asymmetric due to the logarithmic nature of the amplitude, we use this result to calculate the analogue phase error, $\Delta_{\phi_{1/2}}$, by assuming, first, that the phase error is equal to the amplitude error, which is a very common assumption for MT \citep{gamble:1978,wight:1987}, and, second, by noting that the phase relates to the logarithmic amplitude (see appendix \ref{sec:logAniso}), thus:
 \begin{equation}
\Delta_{\phi_{1/2}}=\left|\log\left(\Delta_{\rho^\mathrm{gal}_{1/2}}\right)\right|\frac{180^\circ}{\pi}.
\end{equation} 
Performing these calculations, for example, when the bias of $a_r$ is $\Delta_a=0.2$ for $a=0$, we find an analogue phase error of $5^\circ$. Further, when the bias of $a_r$ is $\Delta_a=0.15$ for $a=0.5$, the analogue phase error is $5.6^\circ$. This value continues to increase to $6.6^\circ$ for $\Delta_a=0.04$ at $a=0.9$. 

From this analysis we find two conclusions. First, the additional error introduced by approximating $\mathbf{P}^{\mathrm{ind}}$ with the Phase Tensor is fairly constant, although a slight increment can be detected for larger values of anisotropic distortion $a$. And secondly, the maximal error has the same order of magnitude as the typical measurement errors and therefore it is likely within error confidence of typical measurements. Let us stress that, although in this example we have only provided an estimate of the galvanic anisotropic distortion, it is still a remarkable result since until now it was generally unresolvable for 2D MT data at a single site if no other complementary data was available \citep{Chave:1994,McNeice:2001,Bibby:2005}.

\subsection{Amplitude Tensor Parameter Maps}
The mapping of the Amplitude Tensor parameters together with the Phase Tensor parameters can be used to directly extract information of the subsurface model for, e.g., appreciation of starting models for computationally expensive 3D inversions. To show this we study four different synthetic models taken from the literature and a real data set. The parameter map plots for each of these models are given in Figures \ref{fig:cbb2004b} to \ref{fig:em3d3c} and \ref{samtex}. 
Each plot contains all four parameters (scale, relative anisotropy,
strike and skew angles) of the corresponding tensor. The scales $\ln\rho$ and $\phi$, as
defined by \eqref{eq:detRhoPhi}, are colour-coded in degrees (Phase Tensor) and
logarithmic apparent resistivity (Amplitude Tensor). The relative anisotropy
parameters, $\rho_a$ and $\phi_a$), as defined in \eqref{eq:rPArAA}, are plotted as an ellipse
with major axis $1+\rho_a$ and $1+\phi_a$ and minor axis $1-\rho_a$ and $1-\phi_a$, respectively. The unit circles of $\rho_a=2/3$, $\rho_a=0$ and $\rho_a=-2/3$, and similarly for $\phi_a$, are provided for comparison at the bottom right corner of the plots. The strike angle is typically aligned
to the major or minor axis of the ellipse but additionally highlighted by a
black line. Finally, the skew angle is displayed qualitatively as white
($0^\circ$), grey ($\le 6^\circ$) and black ($> 6^\circ$) filling out the
smallest ellipse unit circle to illustrate regimes of 2D, quasi 2D and 3D,
respectively.

\subsubsection{Model of \cite{Caldwell:2004}}
\begin{sidewaysfigure}[p]
	\centering
	a) Phase Tensor \\
		\includegraphics[trim=0 50 190 0,  clip,height=150bp]{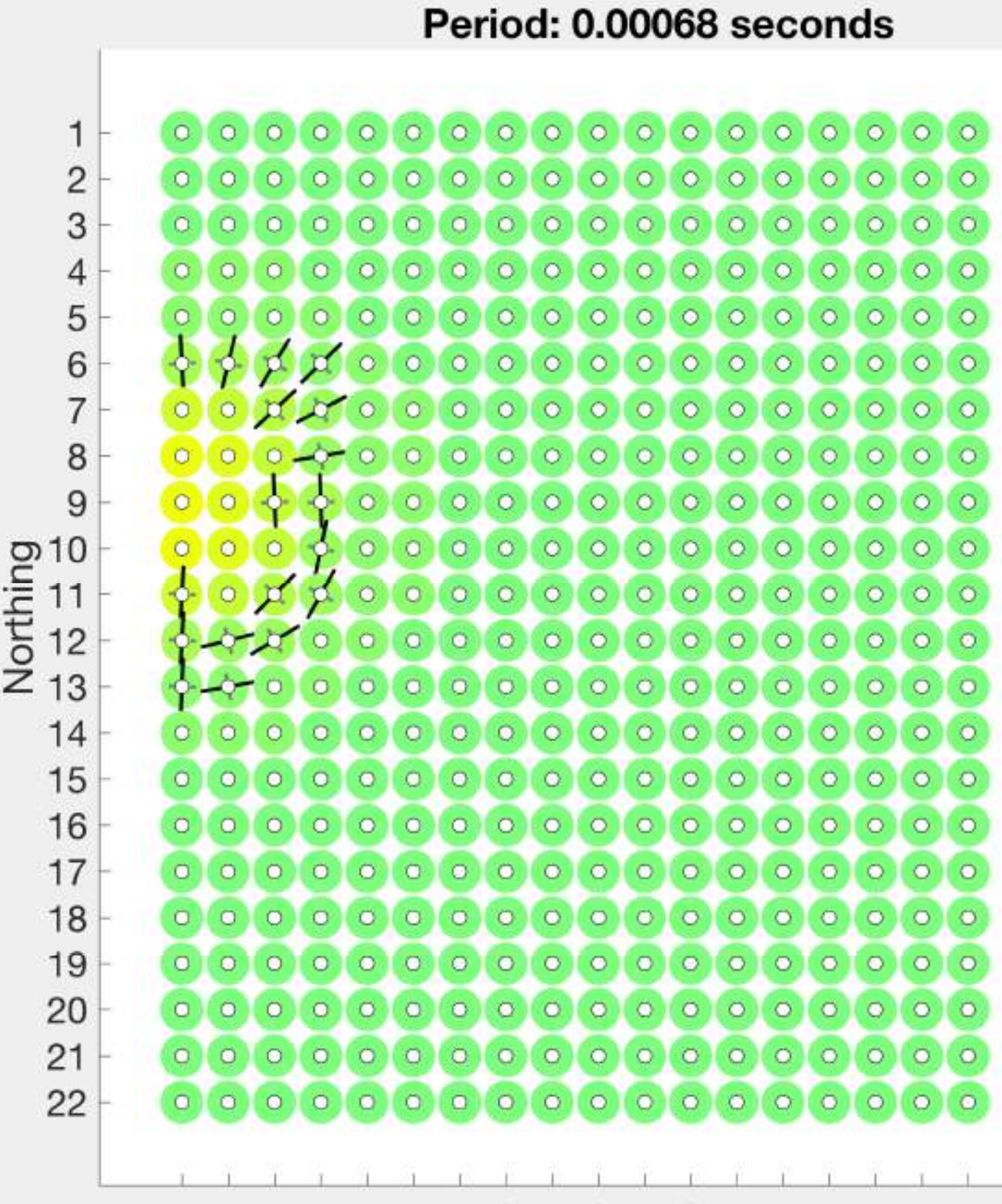}
		\includegraphics[trim=50 50 190 0,clip,height=150bp]{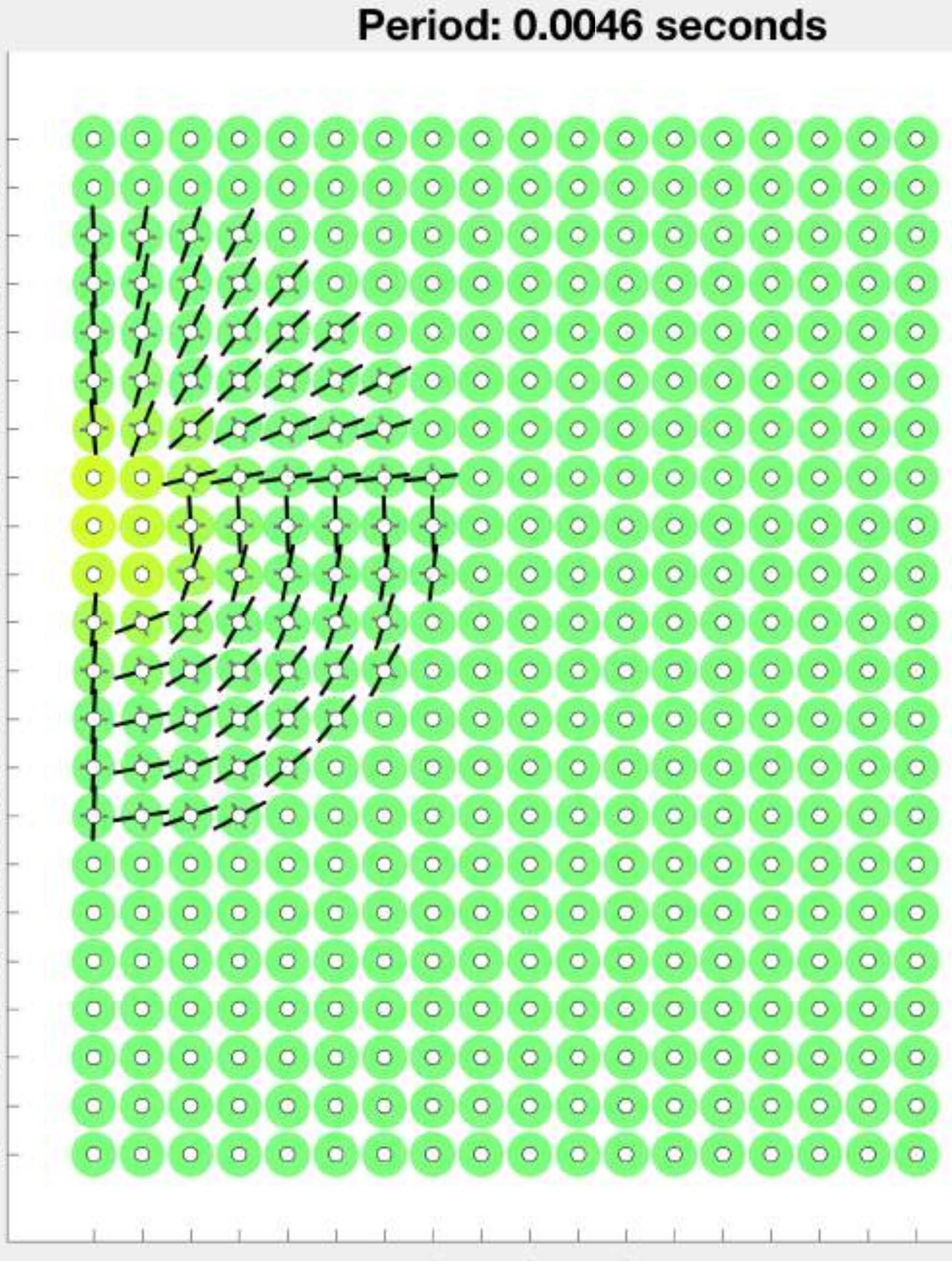}
		\includegraphics[trim=50 50 190 0,clip,height=150bp]{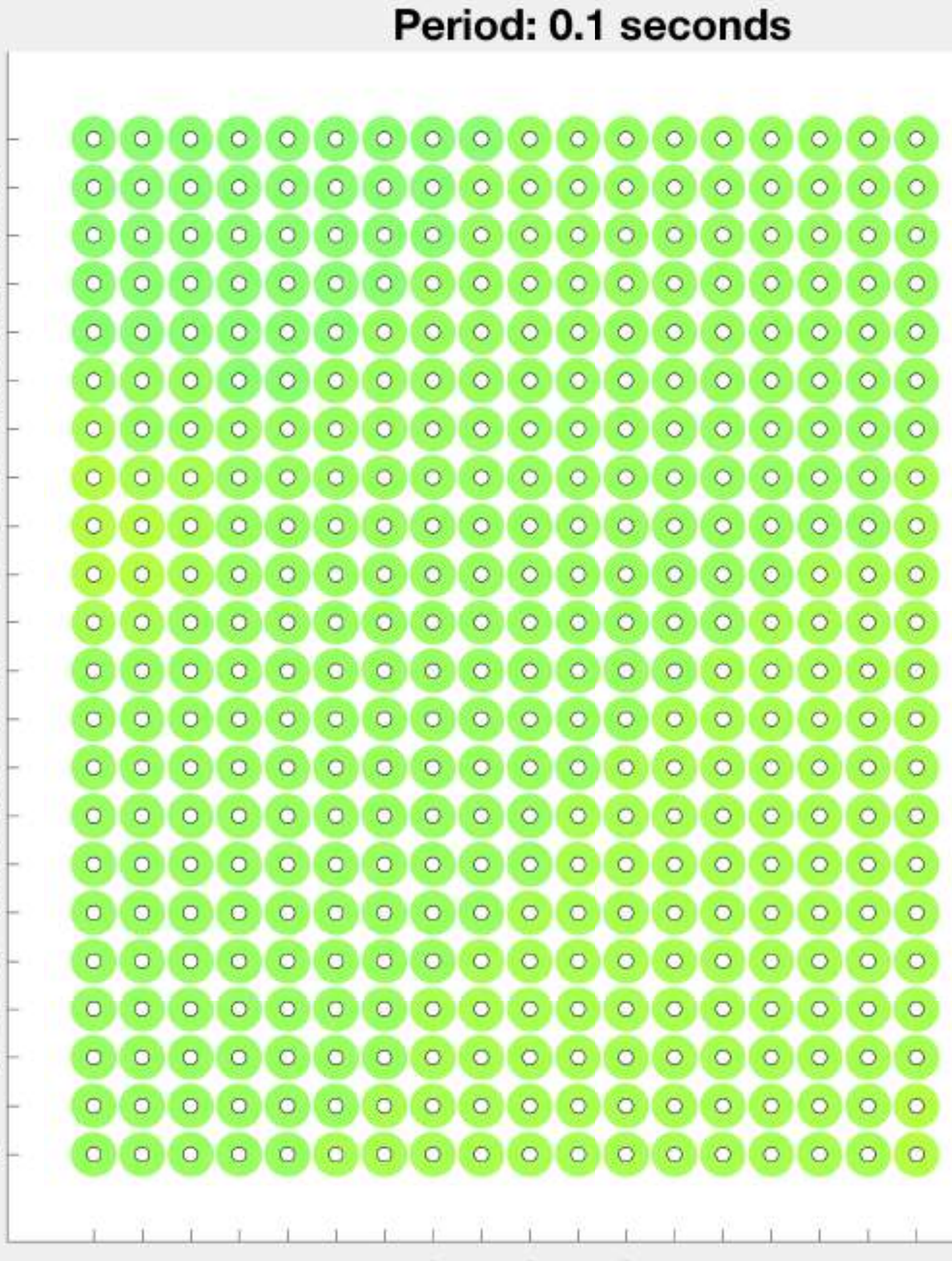}
		\includegraphics[trim=50 50 190 0,clip,height=150bp]{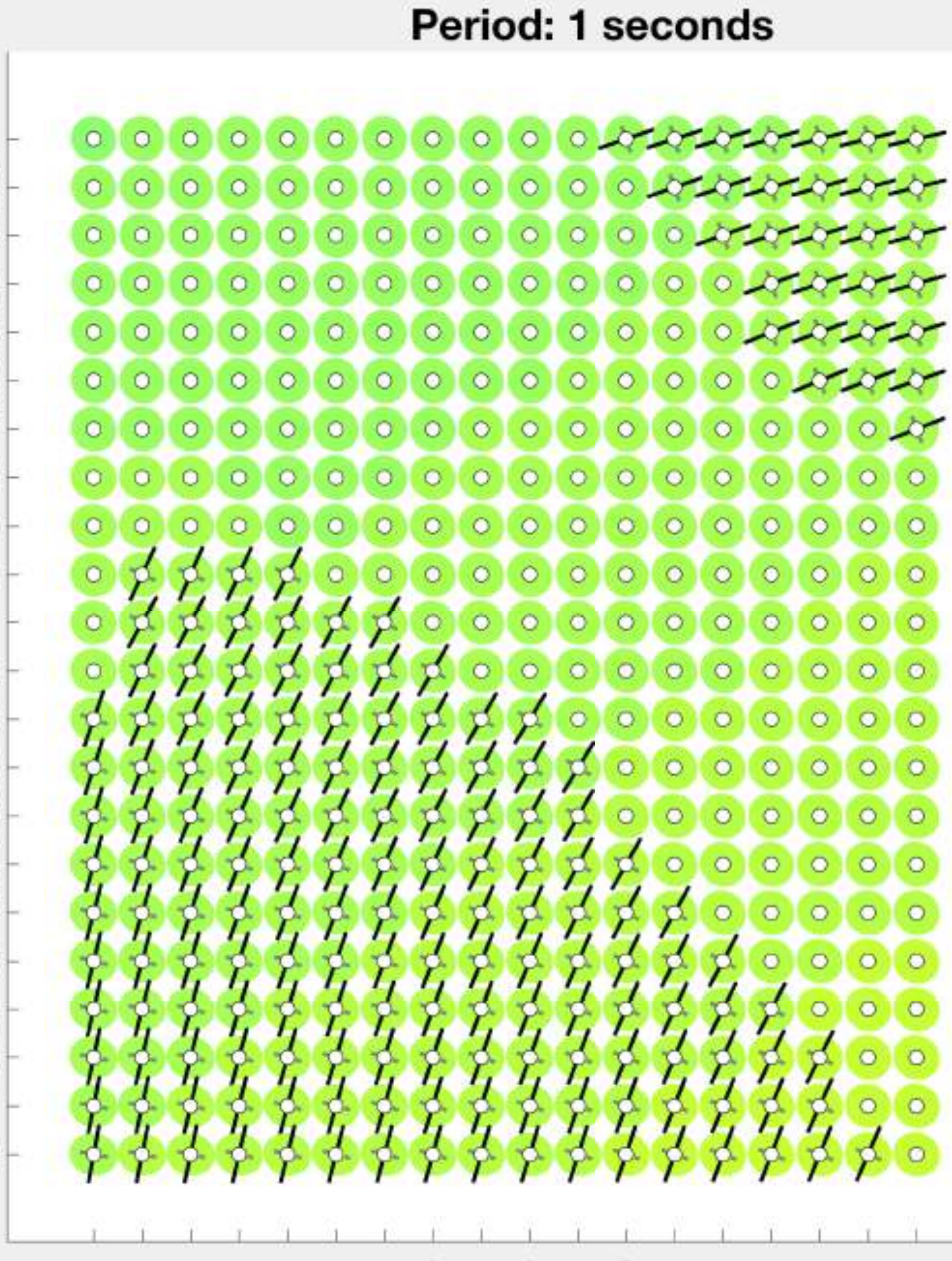}
		\includegraphics[trim=50 50 0 0,    clip,height=150bp]{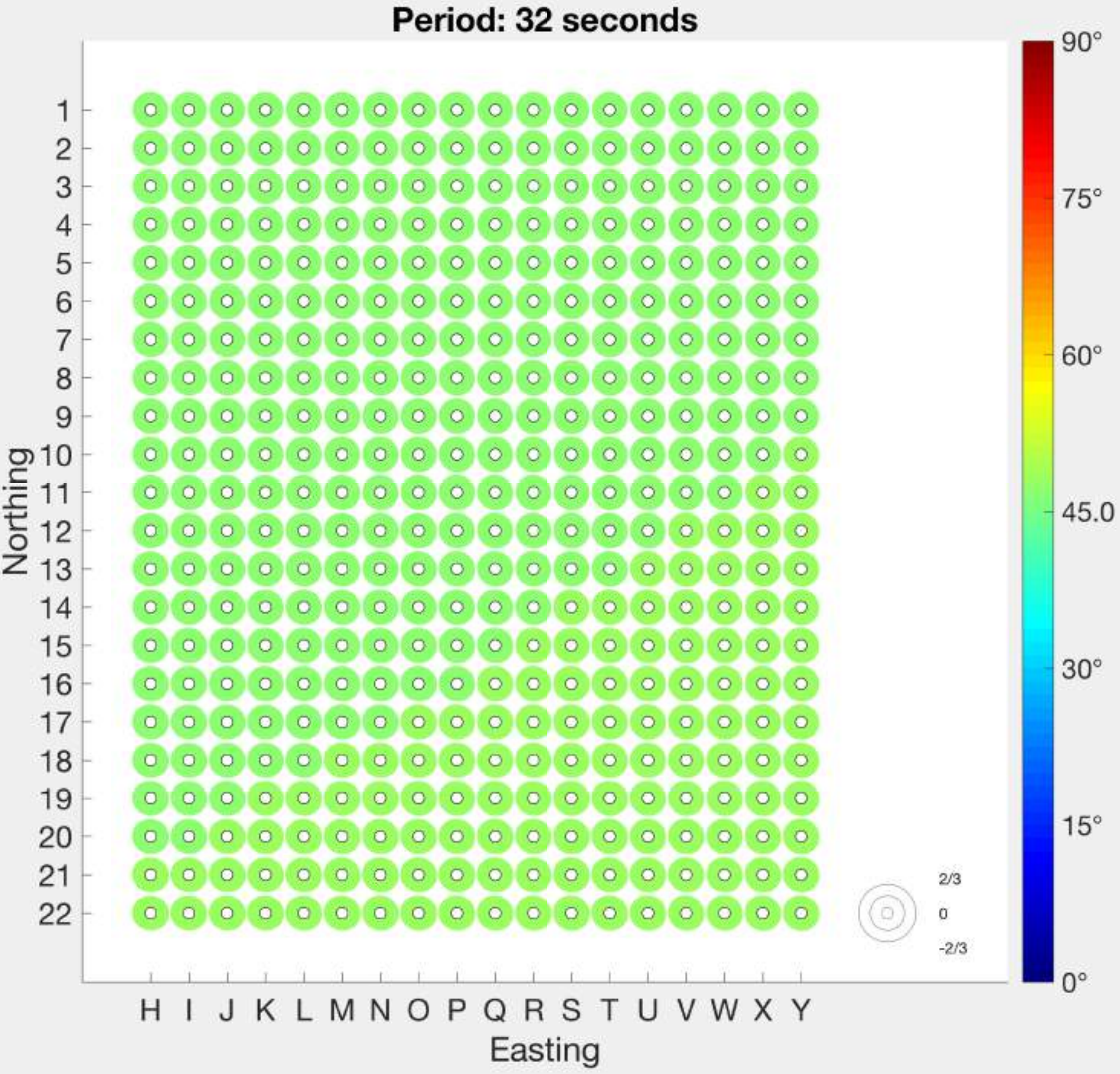}\\
	b) Amplitude Tensor\\
		\includegraphics[trim=0 0 190 50,  clip,height=150bp]{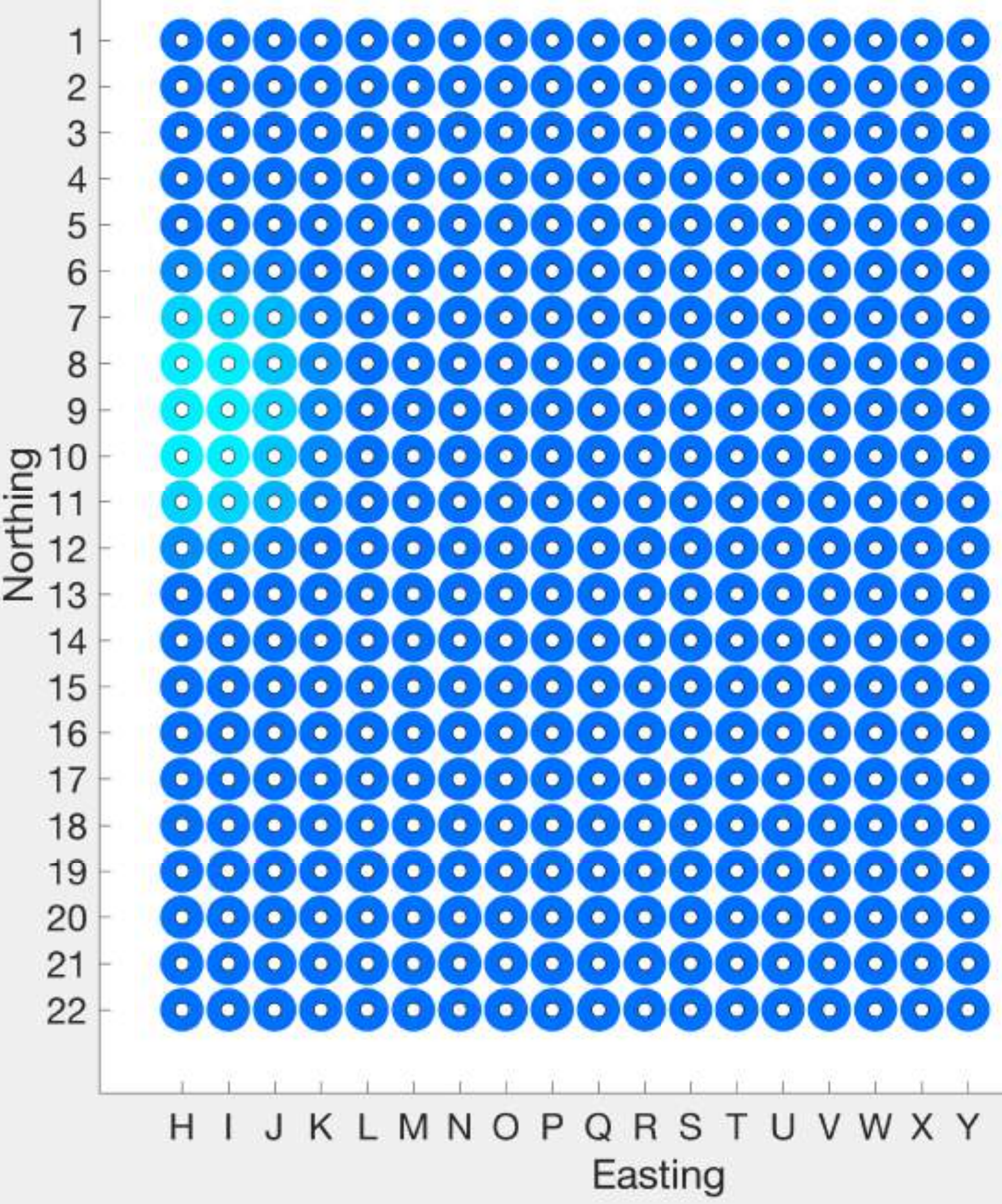}
		\includegraphics[trim=50 0 190 50,clip,height=150bp]{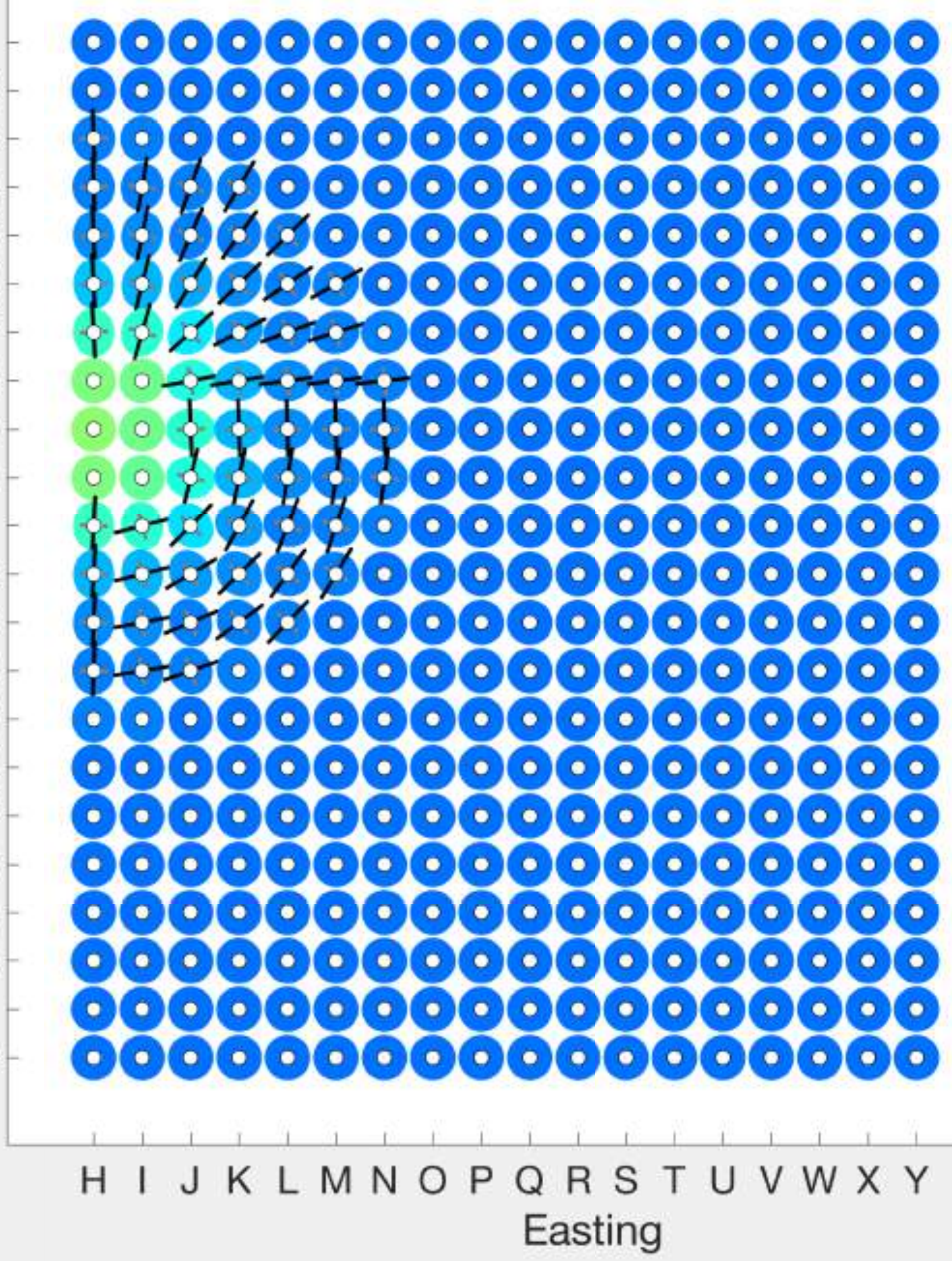}
		\includegraphics[trim=50 0 190 50,clip,height=150bp]{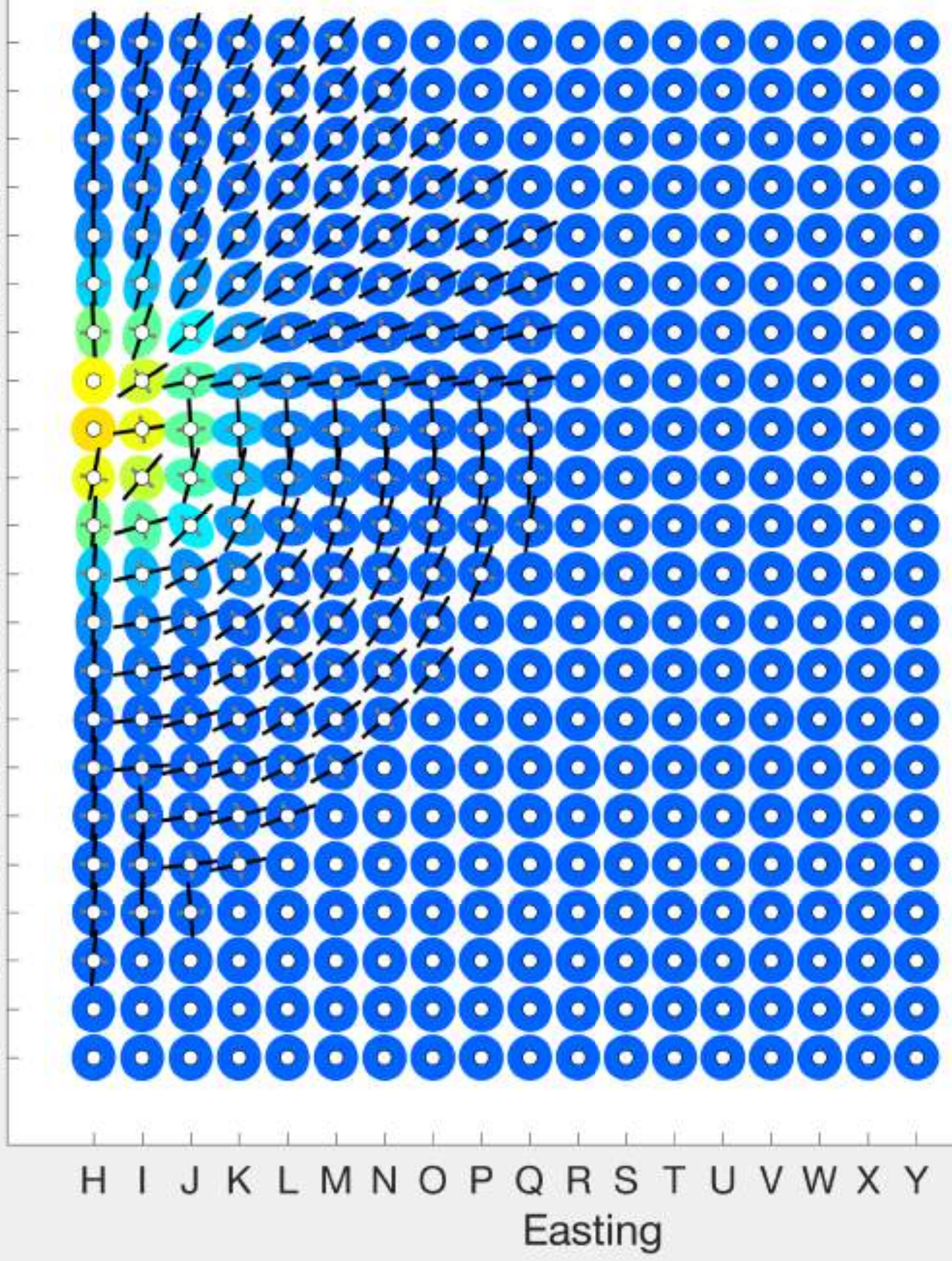}
		\includegraphics[trim=50 0 190 50,clip,height=150bp]{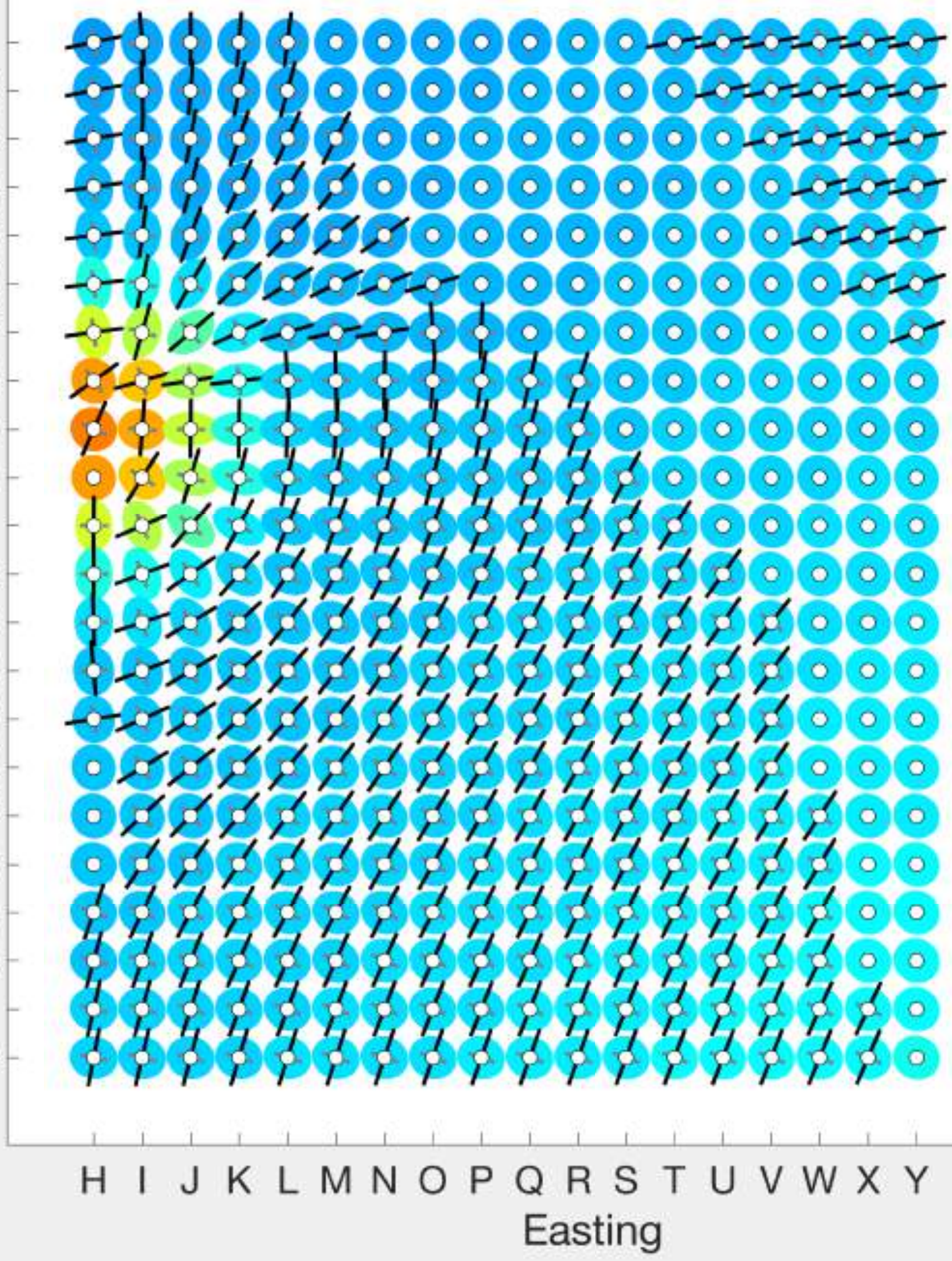}
		\includegraphics[trim=50 0 0 50,    clip,height=150bp]{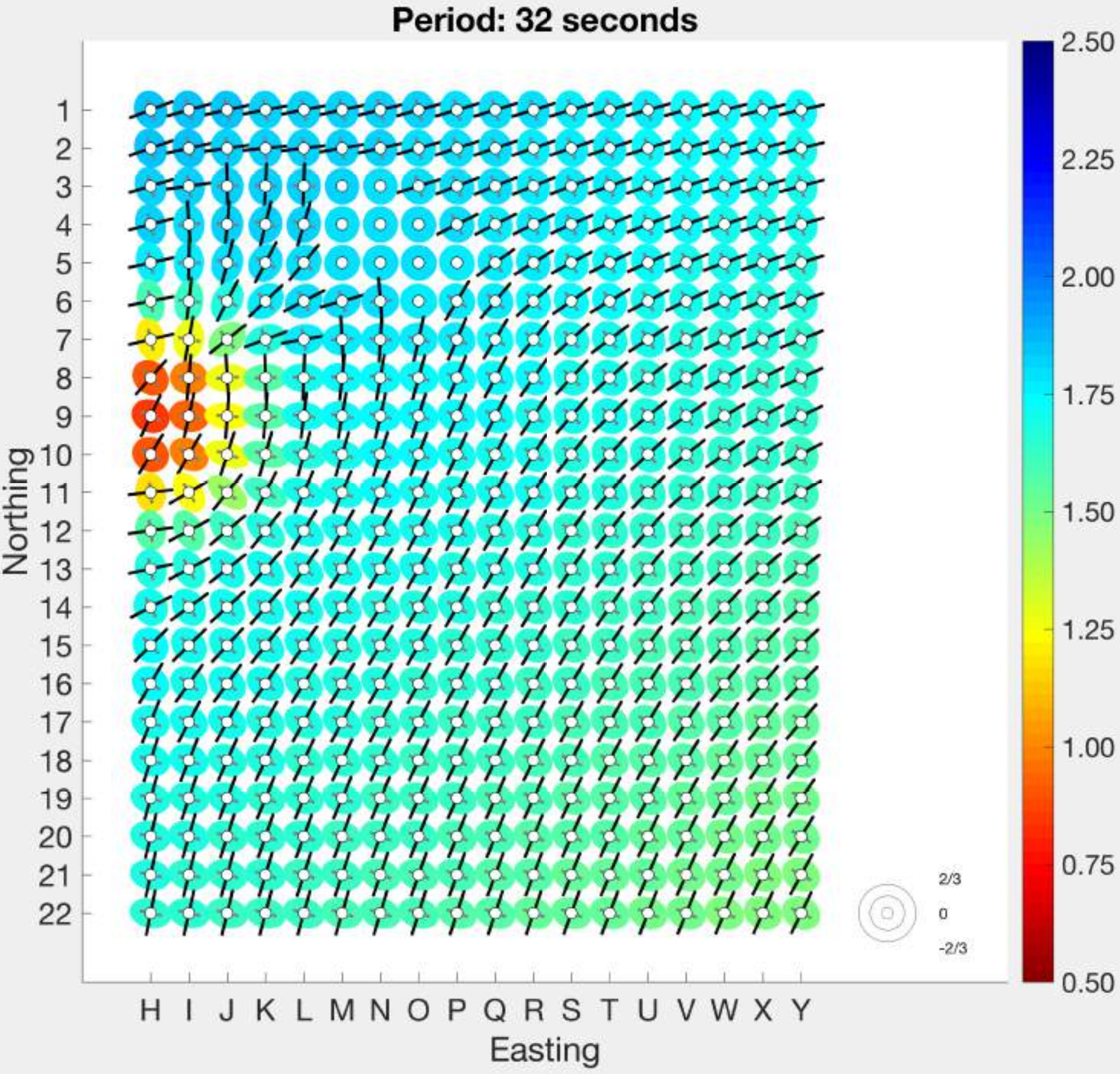}
	\caption{Phase Tensor (a) and Amplitude Tensor (b) parameters at different periods (columns) obtained using the model from \cite{Caldwell:2004}. Tensor parameters are illustrated by scale (degrees and logarithmic apparent resistivity in colour), anisotropy and strike angle (ellipse).}
	\label{fig:cbb2004b}
\end{sidewaysfigure}
The first model we study was proposed by \cite{Caldwell:2004} to introduce the Phase Tensor and demonstrates the significance of spatial strike angle variation towards three-dimensionality of data. 
We study this model to show which additional information can be extracted from the Amplitude Tensor in relation to the Phase Tensor. 
The model consists in an homogeneous background of $100\,\mathrm{\Omega m}$ with a small and shallow ($50\,\mathrm{m}$ top depth) $10\,\mathrm{\Omega m}$ body to the West (only half covered by the map) and a large and deep ($1.5\,\mathrm{km}$ top depth) $1\,\mathrm{\Omega m}$ conductive body to the South-East (only a small portion covered by the map). Figure \ref{fig:cbb2004b} shows map views of the Phase and Amplitude Tensors at various periods for this model.

The Phase Tensor begins to sense the small and shallow Western body at around
$10^{-4}\,\mathrm{s}$ (not shown here, see section \ref{sec:supp}) and the Amplitude Tensor at around $6.8 \cdot 10^{-4}\,\mathrm{s}$ as can be seen by the decreasing apparent resistivity values. 
At around a period of $4.6 \cdot 10^{-3}\,\mathrm{s}$ we can observe that the Phase Tensor has the widest extend of sensitivity to the small Western body indicated by the area with significant strike angle. The scale parameter of $\phi>45^\circ$ clearly indicates the position and the conductive nature of the body. 
Outside the position of the conductive body, the anisotropy ellipses and strike angles are both elongated and measurable, but they are not in the area below the conductive body and therewith this difference specifies also the position and form of the body.
With increasing period the Phase Tensor parameter values approach their homogenous half-space values ($45^\circ$ phase scale, zero anisotropy, zero skew angle and indeterminable strike angle) until the Western body is indeterminable at a period of $0.1\,\mathrm{s}$. Then, at a period around $0.22\,\mathrm{s}$ (not shown here, see section \ref{sec:supp}), the increase of the anisotropy values and the appearance of a measurable strike angle indicate that the Phase Tensor begins to sense the deep and large South-Eastern anomaly.
At a period around $1\,\mathrm{s}$, the area in which the large anomaly is detected by the Phase Tensor parameters is largest and thereafter the area of influence decreases again until the response is below measurable at $32\,\mathrm{s}$. From these observations, we can conclude that the Phase Tensor is sensitive only to structures of specific sizes and depth at a given period, i.e.~only for short periods small, near-surface bodies affect the Phase Tensor parameters and at longer periods deeper and larger structures are sensed. 

On the other hand, the Amplitude Tensor behaves quite differently to the Phase Tensor. 
At low periods, i.e.~at $6 \cdot 10^{-4}\,\mathrm{s}$ and $4.6 \cdot 10^{-3}\,\mathrm{s}$, the Amplitude Tensor scale parameter values at the West of the model are below the background resistivity of $100\,\Omega m$ indicating the position and conductive nature of the western anomaly. Additionally, at these periods, the amplitude strike and anisotropy values are specifying the Western body's form and position just as we observed for the Phase Tensor. 
Observing again the Amplitude Tensor parameters at a larger period, i.e.~at $0.1\,\mathrm{s}$ and $1\,\mathrm{s}$, the parameters show the same behaviour as for lower period values, indicating that the Amplitude Tensor is still affected by the near surface anomaly in contrast to what we observed for the Phase Tensor. 
At these larger periods, we also observe that the apparent resistivity, represented by the amplitude scale parameter, decreases over the shallower body's position, to values around $10\,\Omega m$ and that the anisotropy and strike angle parameters begin to be affected by the large South-Eastern anomaly.
At the longest period of $32\,\mathrm{s}$, the amplitude scale parameter decreases towards the South-East indicating the presence of the large and deep conductive body, and the anisotropy and strike angle parameters (considering the $90^\circ$ strike angle ambiguity) point towards its position. At this period, the small and shallow Western and the large and deep South-Eastern bodies can still be observed in the amplitude parameters but not in the phase parameters. The absence of the bodies' influence in the phase parameters and its presence in the amplitude parameters indicate that the anomalies produced by the bodies are not governed by inductive effects at this period but by purely galvanic effects. 
Therefore, from this example, we can draw two conclusions: (i) from the Amplitude Tensor together with the Phase Tensor parameters map plots, we can observe the position, size and conductive character of isolated subsurface bodies and therewith we would be able to deduce an educated starting model for a subsequent 3D inversion; (ii) the Amplitude Tensor contains the galvanic information of all subsurface structures and supports our interpretation of the galvanic amplitude, i.~e.~in section \ref{sec:IndGalAmpTen}.

\subsubsection{3D MT Modelling and Inversion Workshops}

\begin{sidewaysfigure}[p]
	\centering
	a) Phase Tensor\\
		\includegraphics[trim=0 50 275 0,  clip,height=150bp]{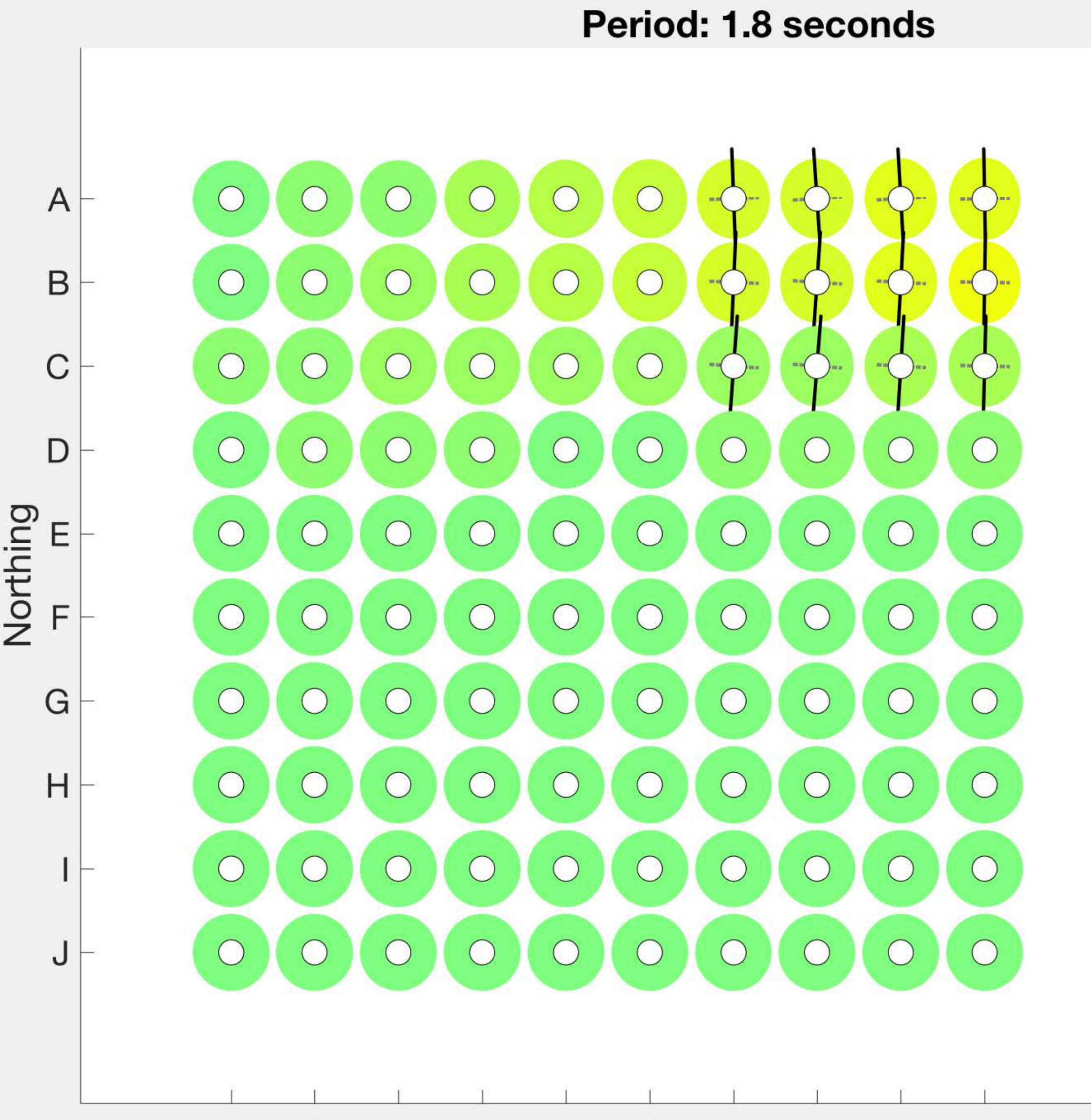}
		\includegraphics[trim=50 50 275 0,clip,height=150bp]{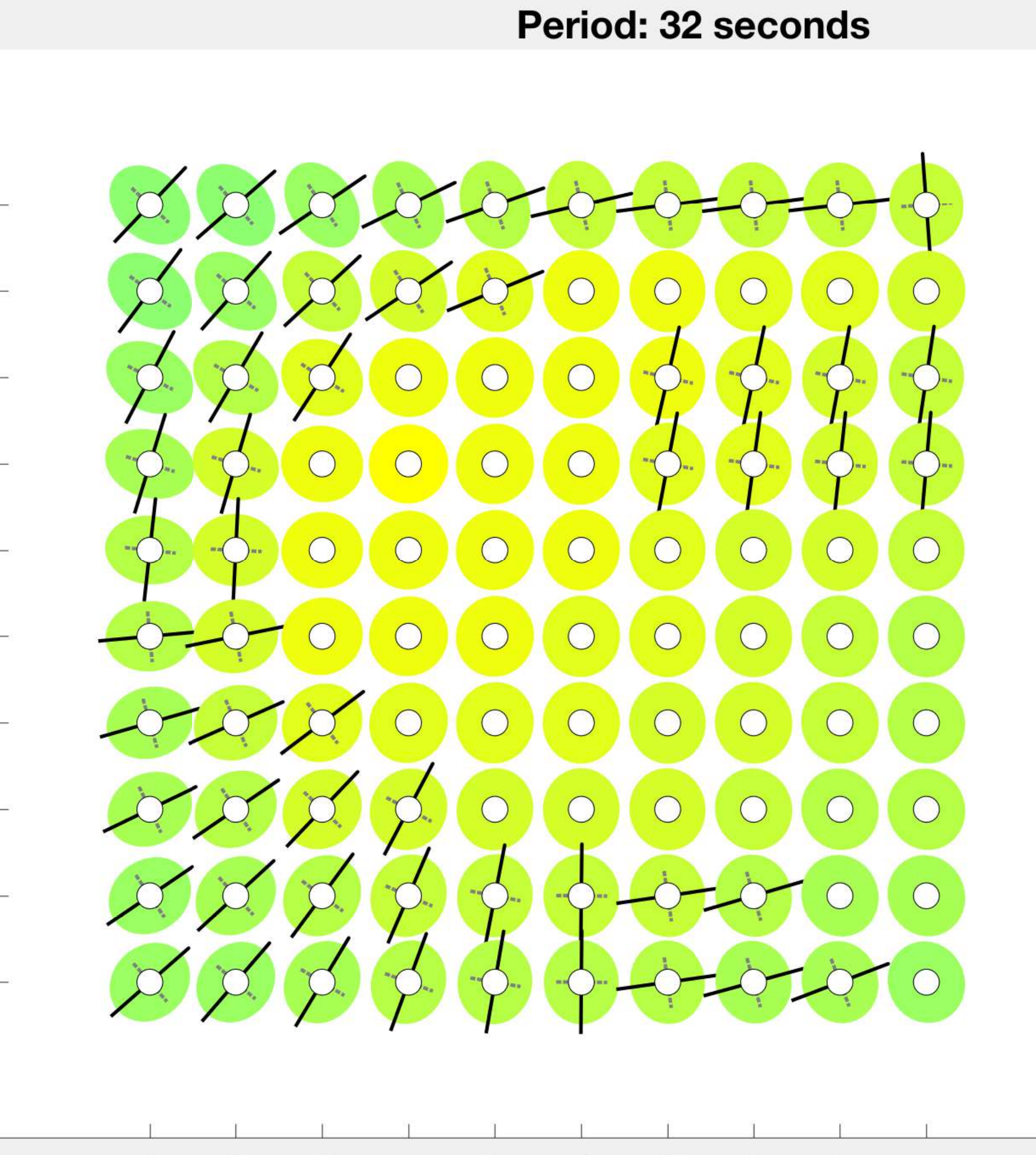}
		\includegraphics[trim=50 50 275 0,clip,height=150bp]{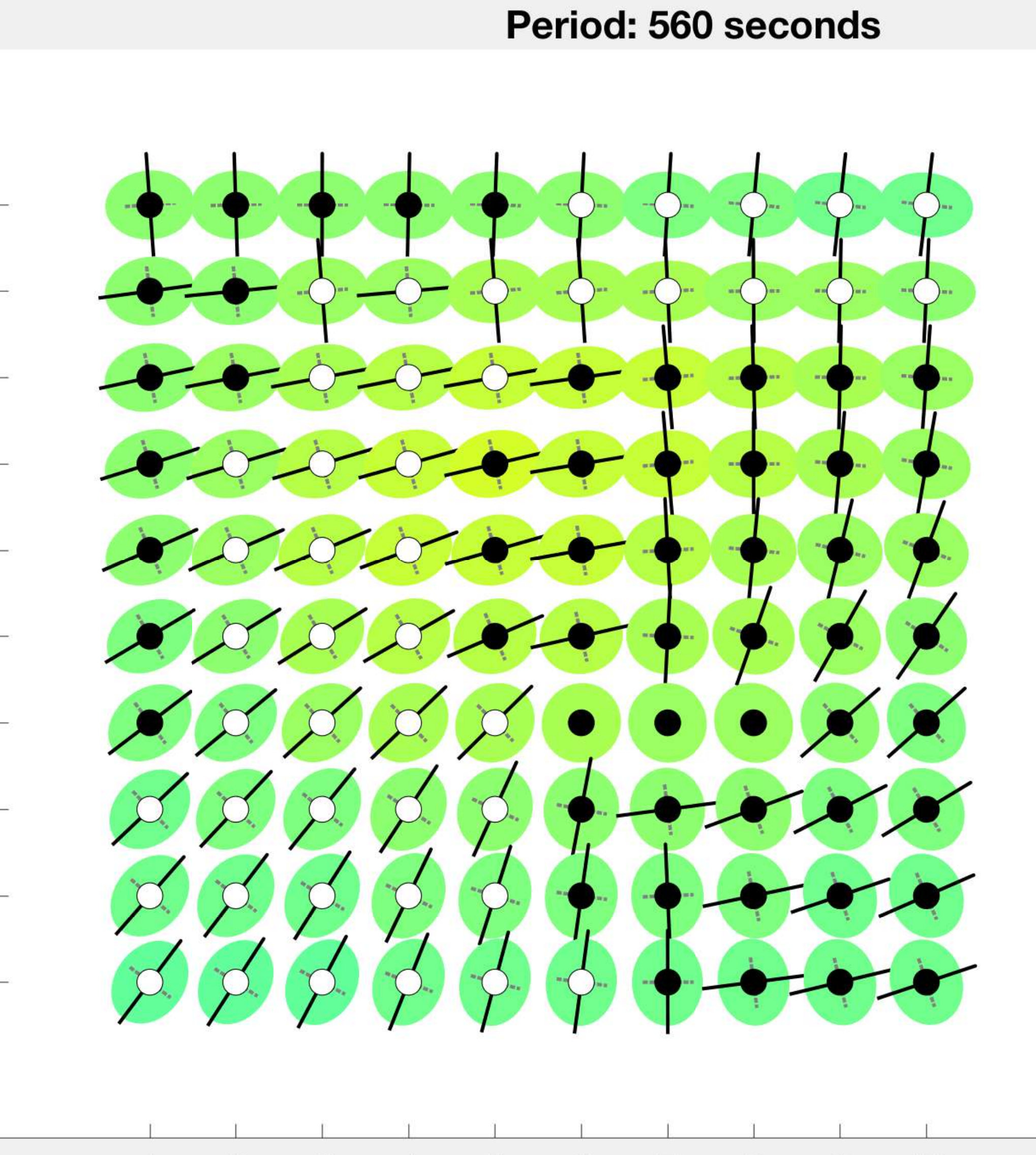}
		\includegraphics[trim=50 50 0 0,    clip,height=150bp]{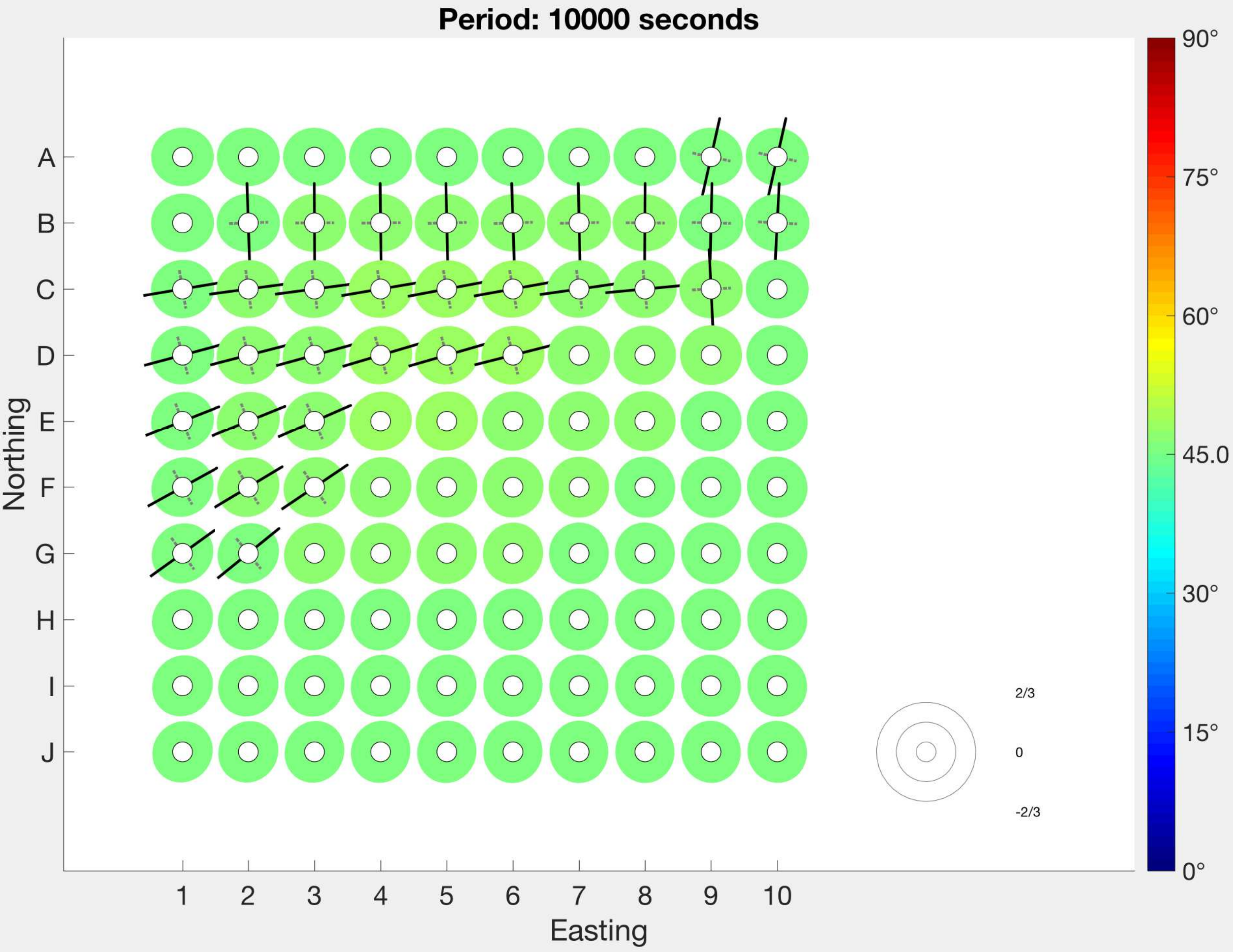}\\
	b) Amplitude Tensor\\
		\includegraphics[trim=0 0 275 50,  clip,height=150bp]{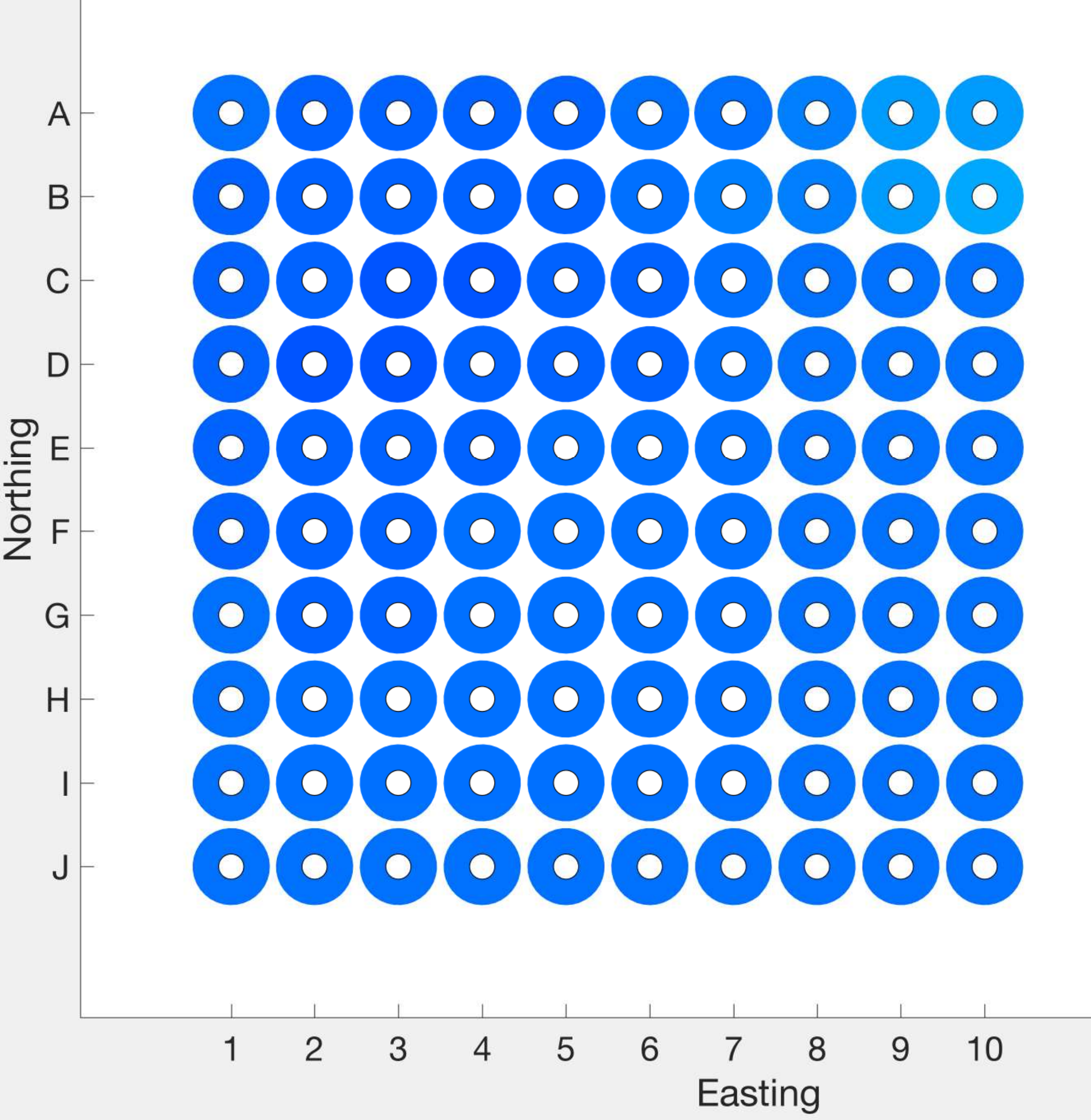}
		\includegraphics[trim=50 0 275 50,clip,height=150bp]{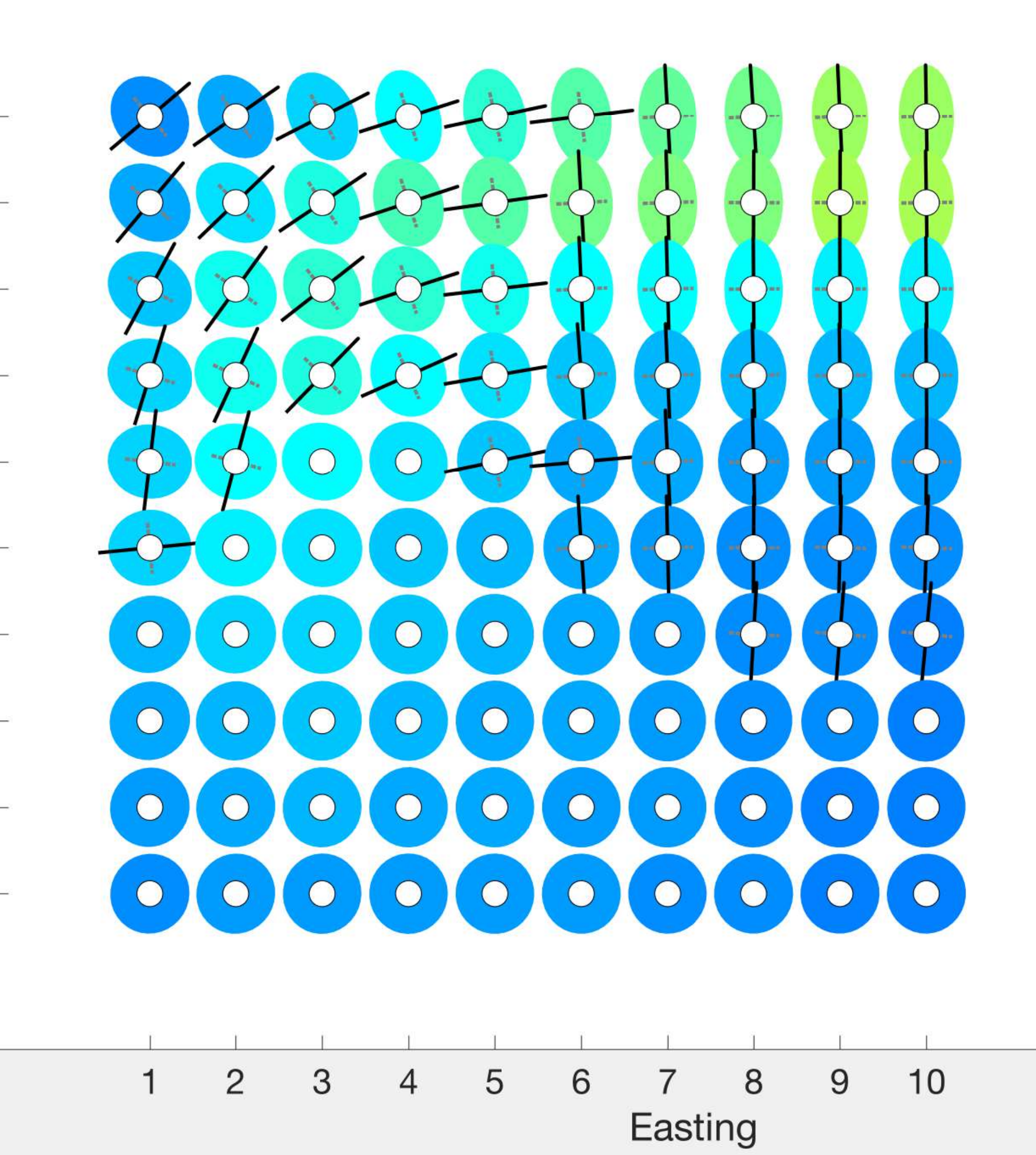}
		\includegraphics[trim=50 0 275 50,clip,height=150bp]{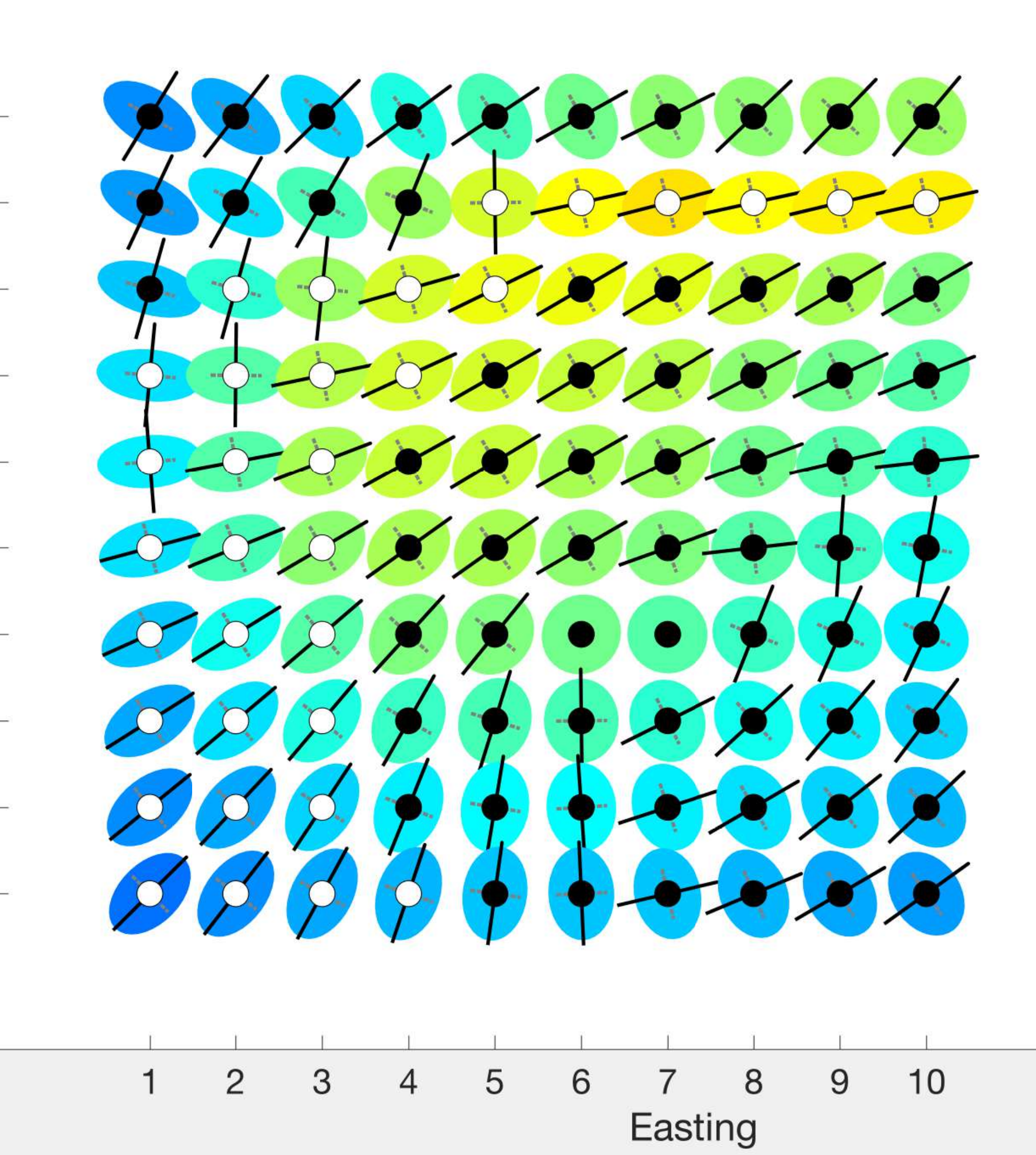}
		\includegraphics[trim=50 0 0 50,    clip,height=150bp]{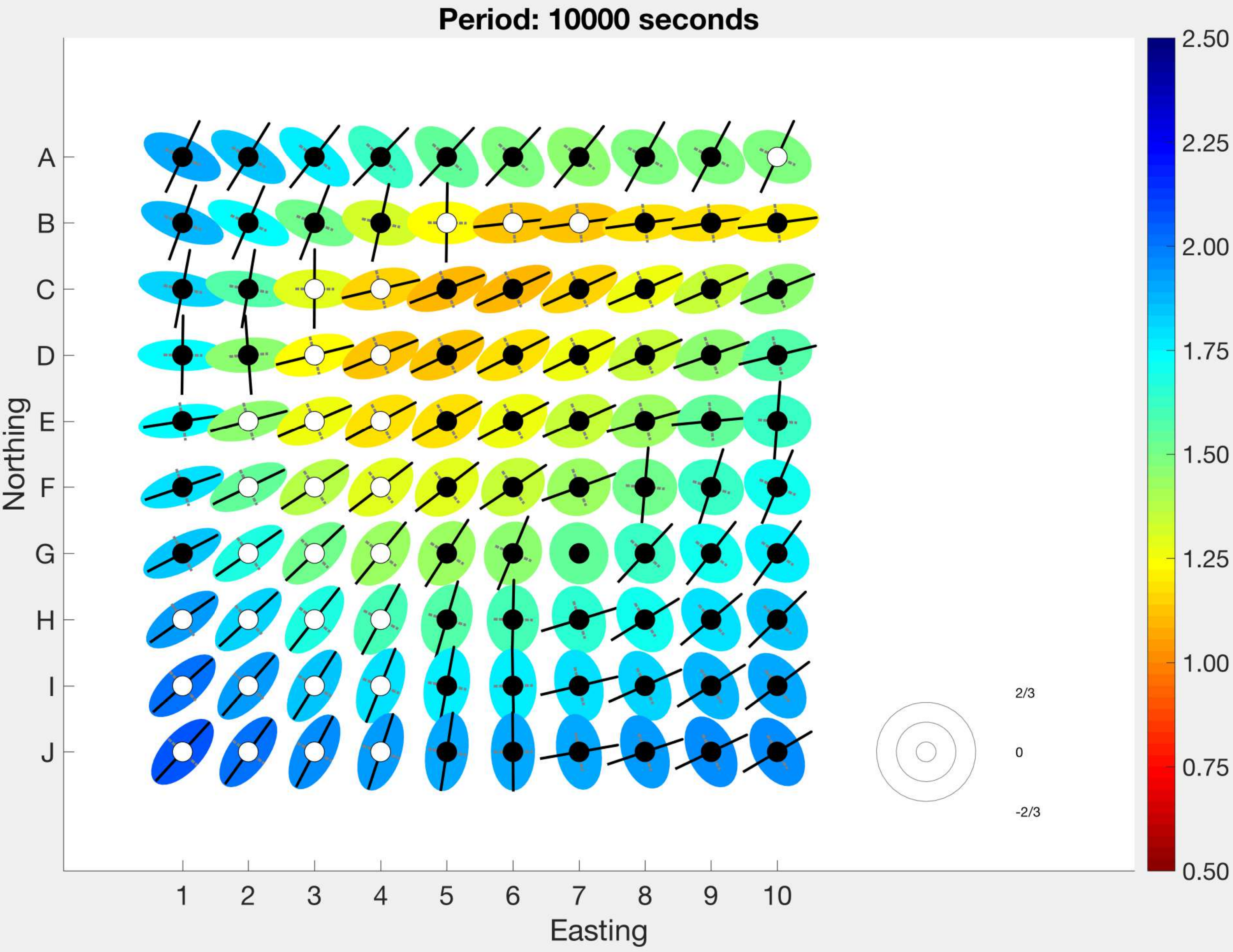}
	\caption{Phase Tensor (a) and Amplitude Tensor (b) parameters at different periods (columns) obtained using the Dublin 3D Modelling and Inversion Workshop (2008) Secret Model 1. Tensor parameters are illustrated by scale (degrees and logarithmic apparent resistivity in colour), skew (circle fill - black: 3D, white: 2D), anisotropy and strike angle (ellipse).}
	\label{fig:em3d1}
\end{sidewaysfigure}
The model used in the previous example while insightful, it lacks illustration of the skew angle parameter. Thus, to show the information given by this parameter we analyse two additional synthetic models that were particularly designed to contain 3D data.
We take them from two different 3D MT Modelling and Inversion Workshops celebrated in 2008 and 2016.

\paragraph{Secret Model $\#1$ (2008)}
One of the models is the secret Model $\#1$ from the 3D MT Modelling and Inversion Workshop held in 2008 that we already used to illustrate the dimensionality analysis of the Amplitude and Phase Tensors with site $A09$. It consists of a $100\,\mathrm{\Omega m}$ homogenous background with a conductive body ($1\,\mathrm{\Omega m}$) which, starting from the North-East, spirals down towards, first West, then South, East and finally North as displayed in Figure \ref{fig:DSM1model}. More details of the model are given by \cite{Miensopust:2013}. Figure \ref{fig:em3d1} illustrates the Phase and Amplitude Tensor parameters for the impedance response of this model at selected periods. At the shortest period, $1.8\,\mathrm{s}$, the phase scale, anisotropy ellipses and strike angle parameters are lower than $45^\circ$, elongated in the North-South direction and measurable, respectively, in the North-Eastern corner. This indicates, together with the skew values lower than $3^\circ$, that the shallowest part of the anomaly can be considered a 2D feature at this period. 
At the same period and at the same location, the amplitude scale parameter 
is slightly smaller compared to its value in other regions of the model, indicating the presence of the conductive body.
At the next period, $32\,\mathrm{s}$, the phase strike angles and the anisotropy ellipses indicate the orientation of the downwards spiralling body, and the phase scale values are larger than $45^\circ$ in the centre of the map showing the body's conductive nature. 
At the same period, the amplitude parameters show the same behaviour as the phase parameters but do not reach the same depth, which we can deduce from the fact that the strike angles and anisotropy ellipses are not yet measurable for the deeper structure that is oriented in the East-West direction. Also at this period, the Phase Tensor is less sensitive to the near-surface part of the body as shown by the phase strike angle and the anisotropy ellipse, which are immeasurable and zero, respectively.

The Phase Tensor skew angle parameter, which indicates a 3D structure when its value is non-zero, shows a maximum at 
around $560\,\mathrm{s}$ revealing the complexity of the model. At this period, $560\,\mathrm{s}$, the phase scale parameter $\phi$ approaches $45^\circ$ at most sites, indicating that the Phase Tensor sensitivity to the overall conductive body is diminishing.
At the same period, the Amplitude Tensor recognises that the shallow structure is a dominant part of the model and the parameters only diverge slightly from the values at earlier periods. This is particularly pronounced for the strike angle and anisotropy ellipses, which, in the Northern part of the map, indicate the strike direction of the shallow part of the body that is oriented East-West in contrast to the respective phase parameters, which corresponds to the deep part of the body that is oriented North-South.

The MT response is based on a diffusion process and, therefore, deeper structures must be much larger than shallower ones in order to have the same impact on the impedance. 
This explains what we are observing in the amplitude strike angle in the example discussed above at $560\,\mathrm{s}$ and what we also observe at $10,\!000\,\mathrm{s}$. The amplitude strike angle is much more affected by the superficial part of the model than by the deep one because although the structure widens and thickens with depth, this variation is not enough pronounced so that the deeper part has the same impact as the shallower one.
At the longest period, $10,\!000\,\mathrm{s}$, the Phase Tensor parameters are almost not affected by the conductive anomaly of the model, whereas the Amplitude Tensor contains the full integrated information of it according to the relative dominance of the structural scales. We interpret this phenomenon, which we also observed in the previous example, as the Amplitude Tensor is a purely galvanic response at this period.

\paragraph{Secret Model $\#3$ (2016)}
\begin{sidewaysfigure}[p]
	\centering
	a) Phase Tensor\\
		\includegraphics[trim=0 50 200 0,  clip,height=150bp]{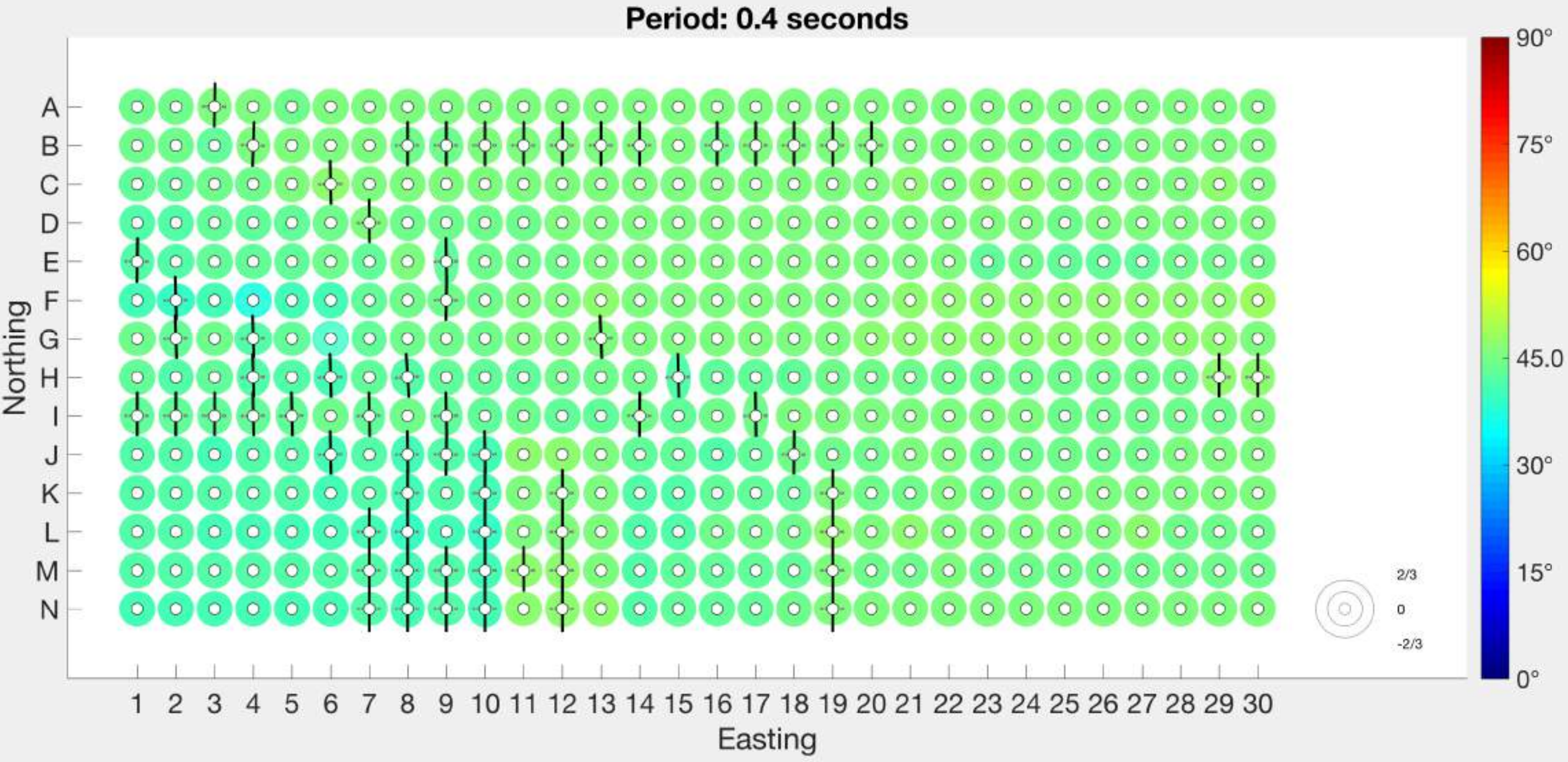}
		\includegraphics[trim=50 50 0 0,    clip,height=150bp]{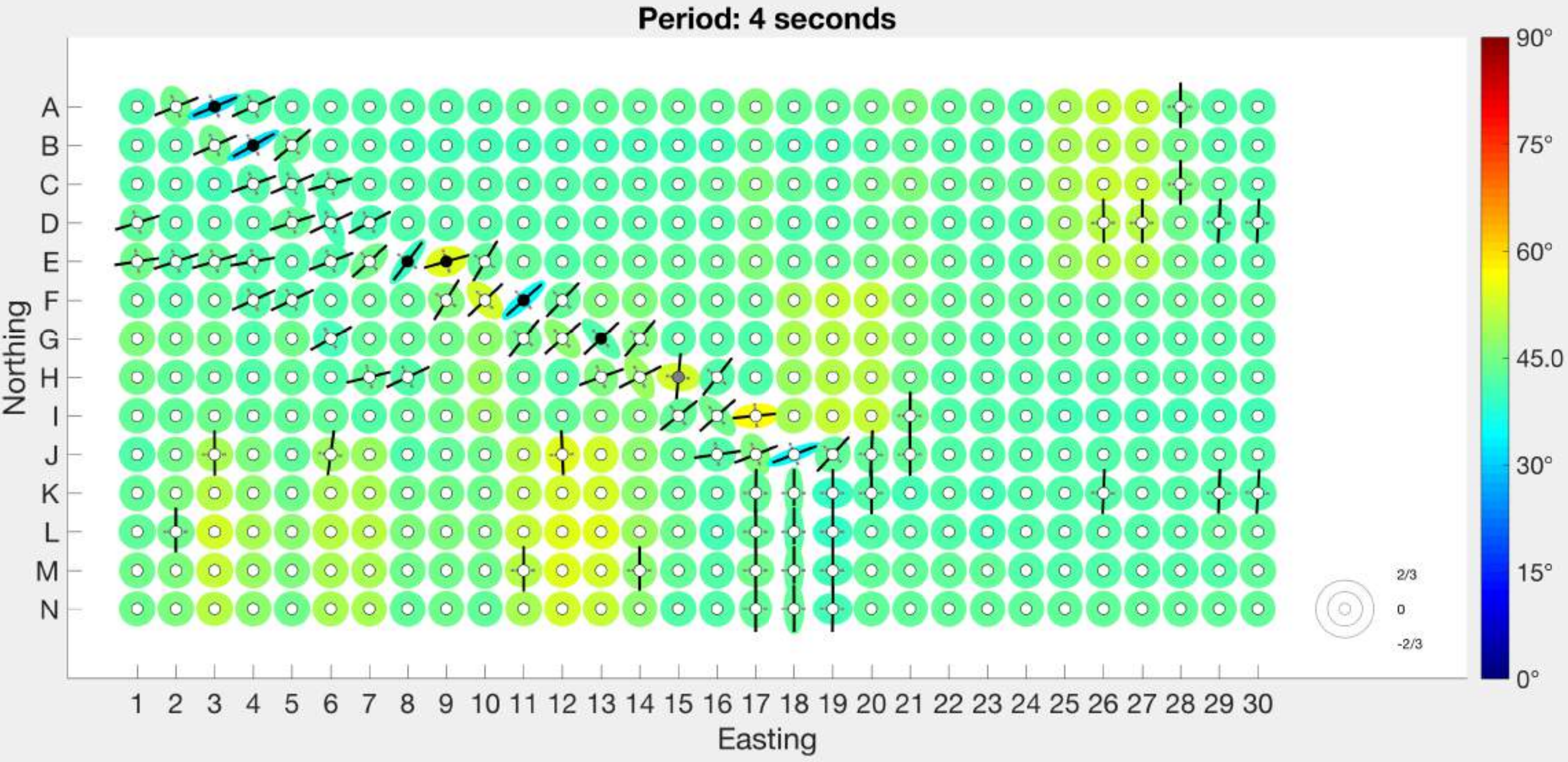}\\
	b ) Amplitude Tensor\\
		\includegraphics[trim=0 0 200 50,  clip,height=150bp]{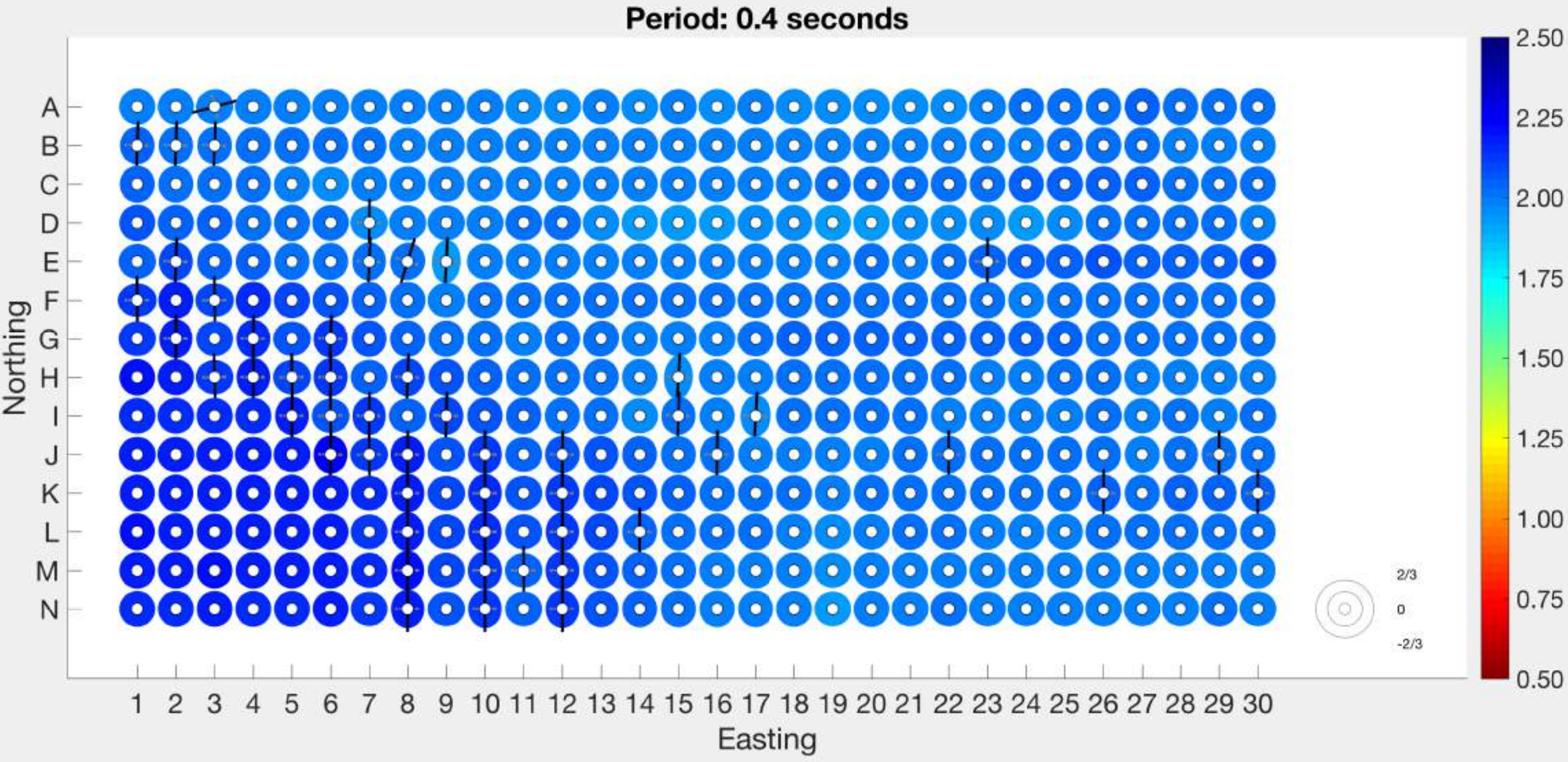}
		\includegraphics[trim=50 0 0 50,    clip,height=150bp]{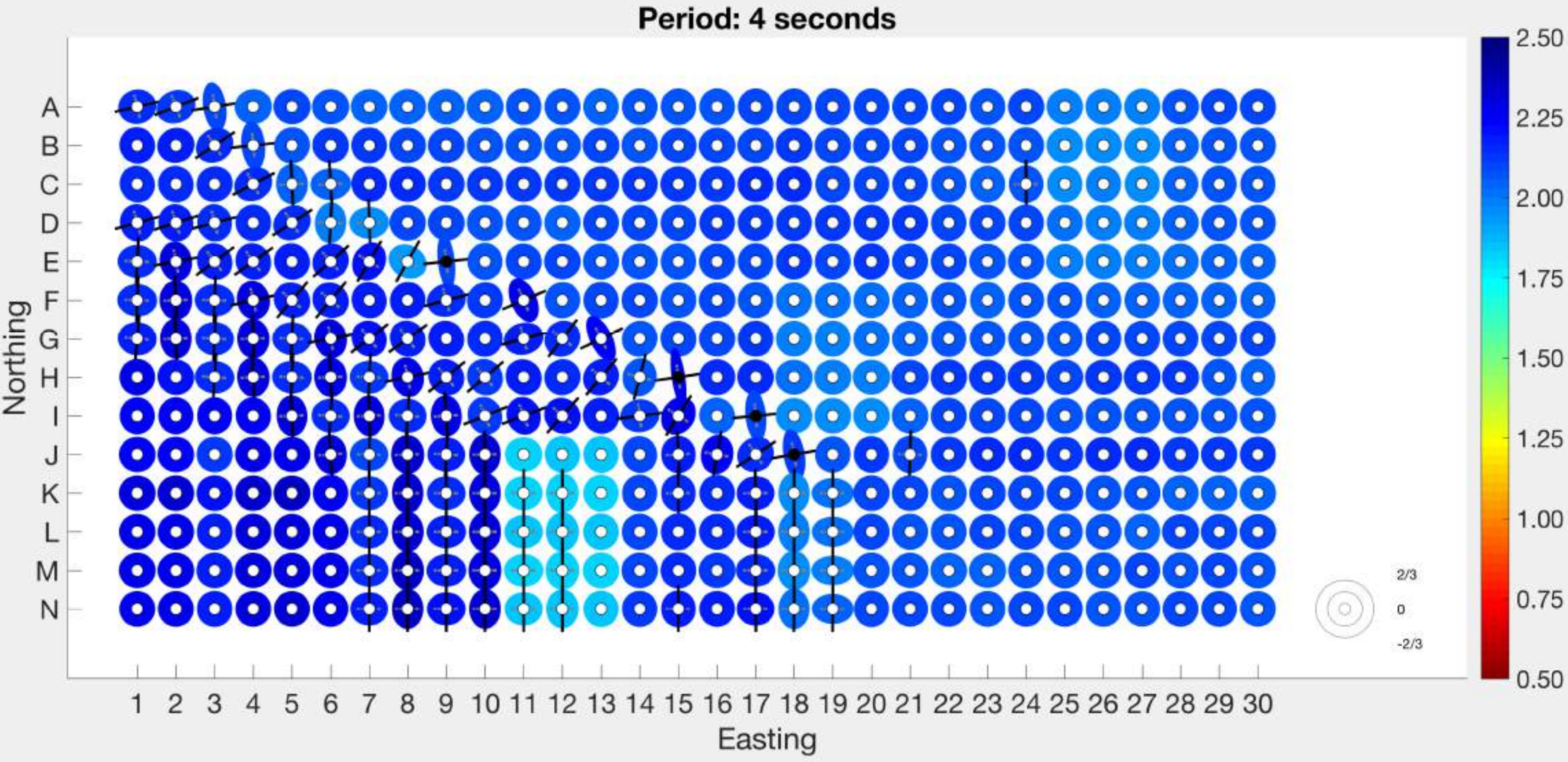}
	\caption{Phase Tensor (a) and Amplitude Tensor (b) parameters at periods of $0.4\,\mathrm{s}$ (left) and $4\,\mathrm{s}$ (right) obtained using the 3D Modelling and Inversion Workshop (2016) Secret Model 3. Tensor parameters are illustrated by scale (degrees and logarithmic apparent resistivity in colour), skew (circle fill - black: 3D, grey: quasi 2D, white: 2D), anisotropy and strike angle (ellipse).}
	\label{fig:em3d3a}
\end{sidewaysfigure}
\begin{sidewaysfigure}[p]
	\centering
	a) Phase Tensor\\
		\includegraphics[trim=0 50 200 0,  clip,height=150bp]{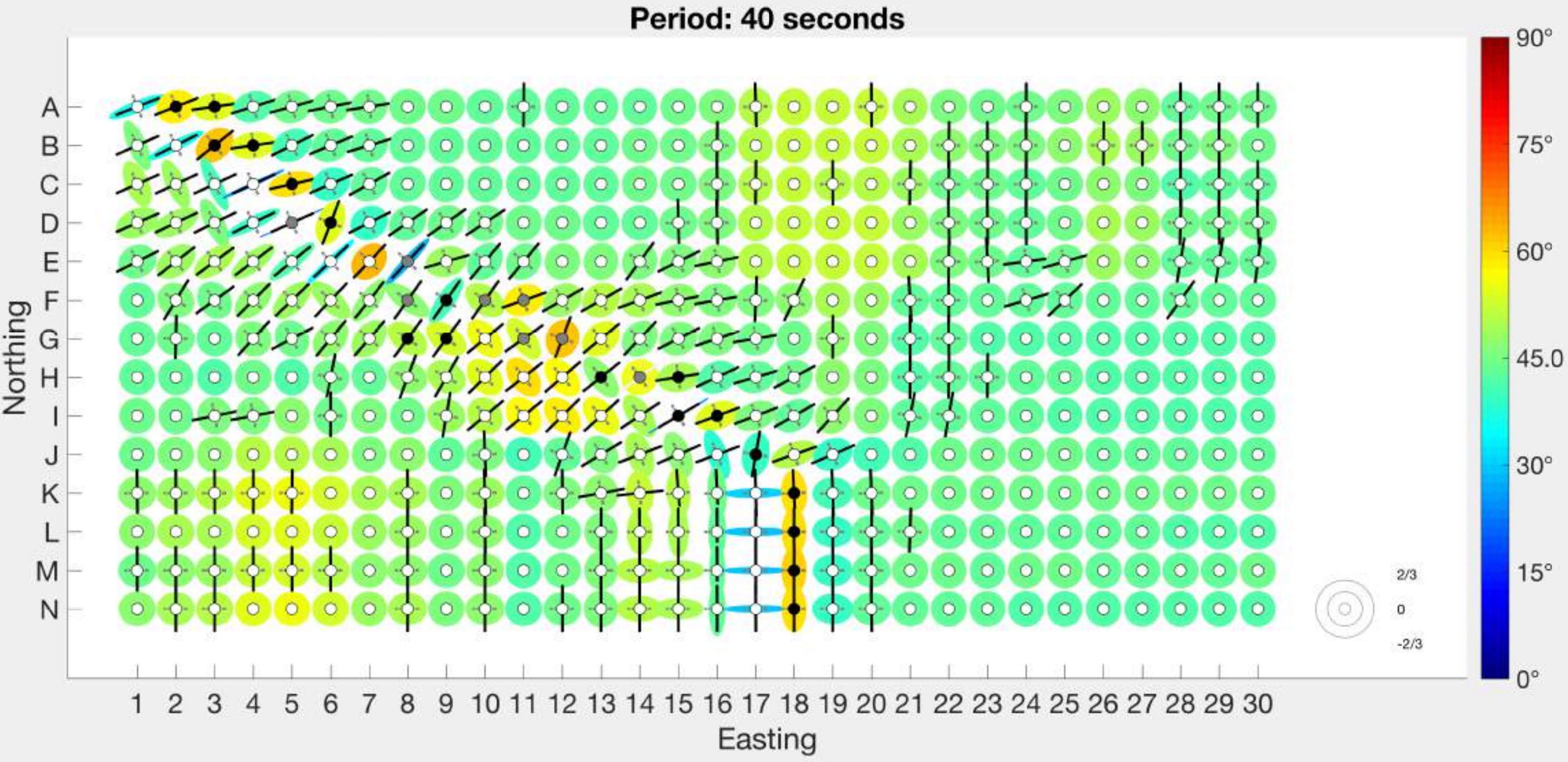}
		\includegraphics[trim=50 50 0 0,    clip,height=150bp]{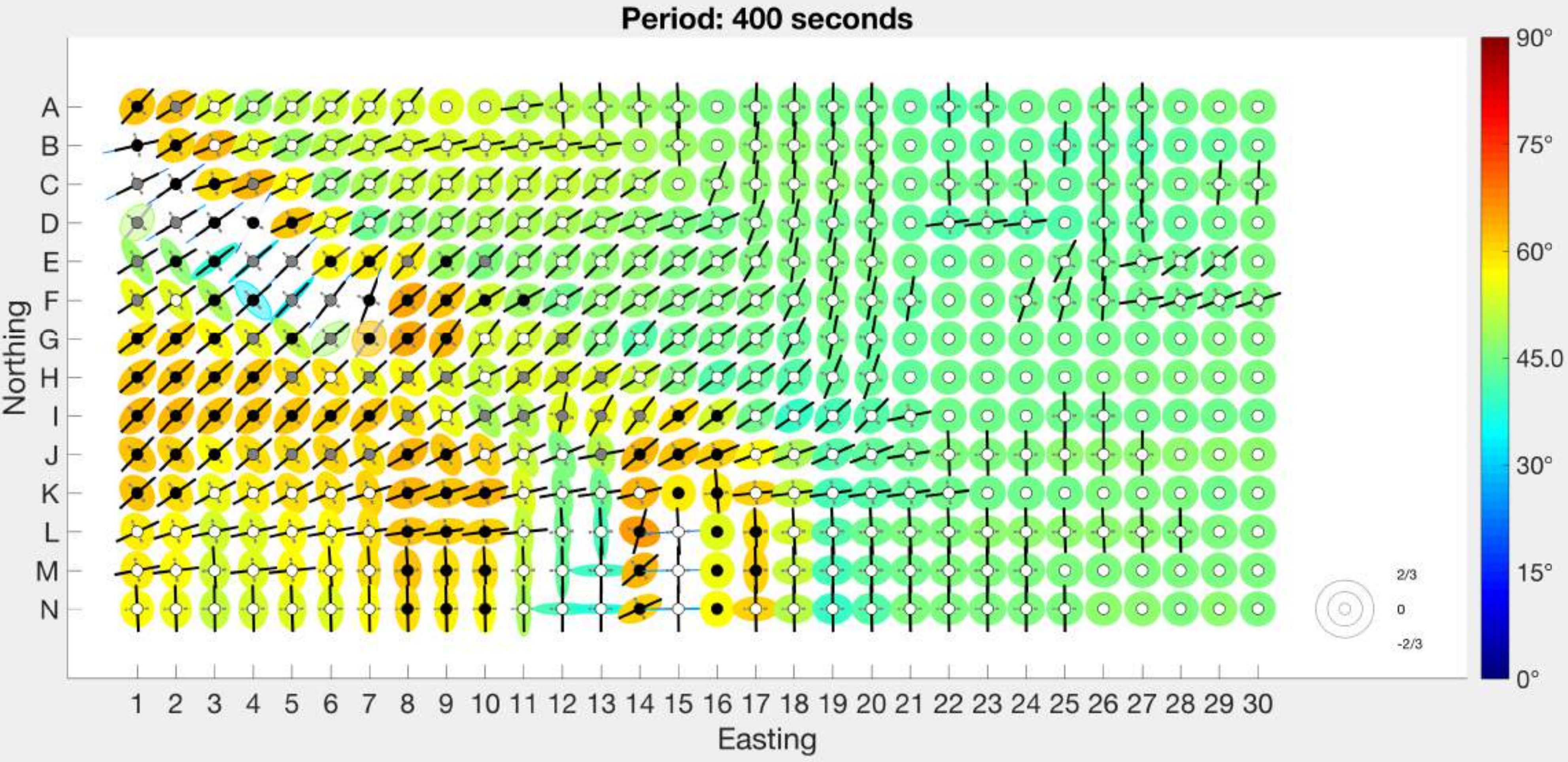}\\
	b) Amplitude Tensor\\
		\includegraphics[trim=0 0 200 50,  clip,height=150bp]{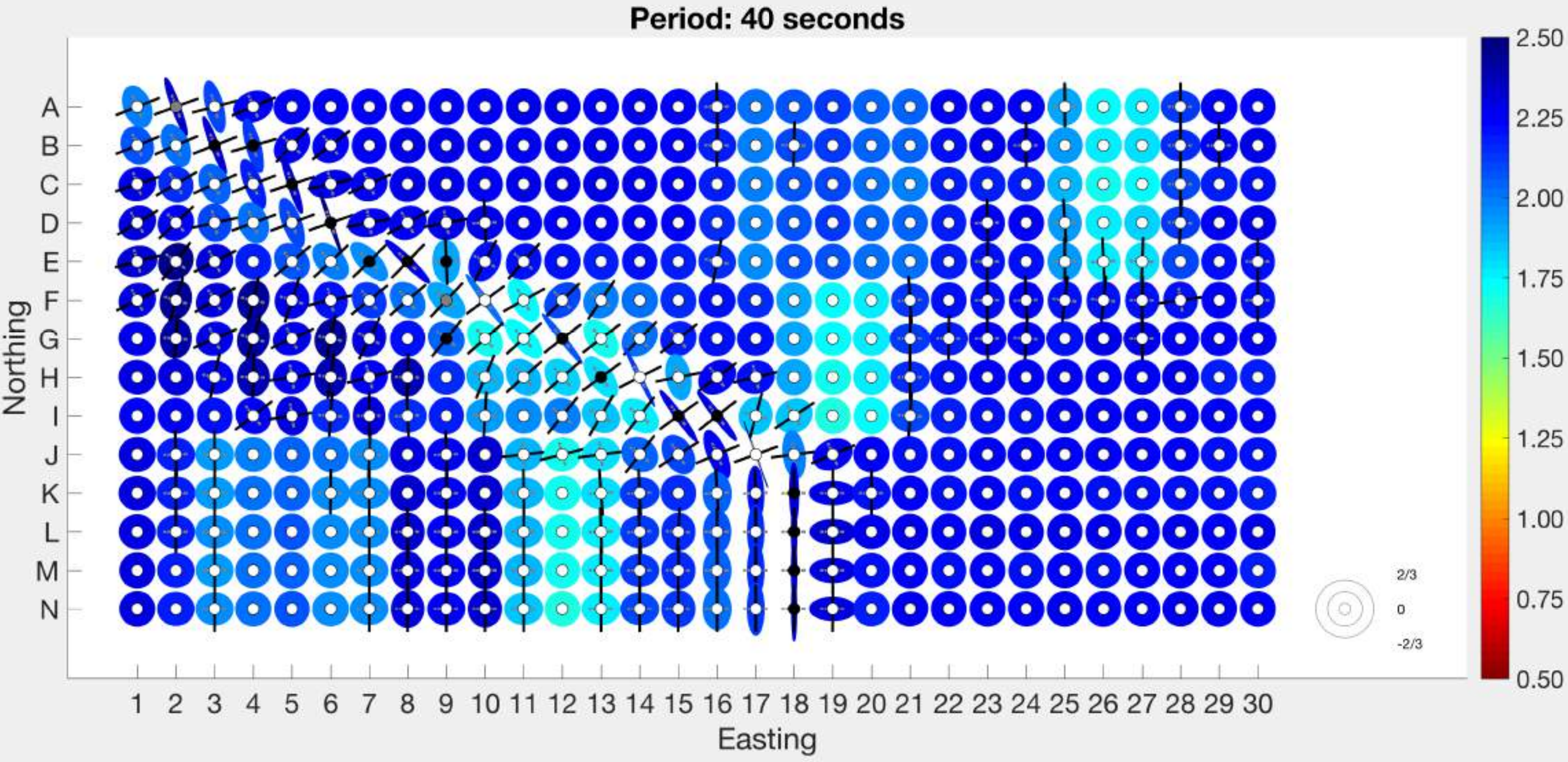}
		\includegraphics[trim=50 0 0 50,    clip,height=150bp]{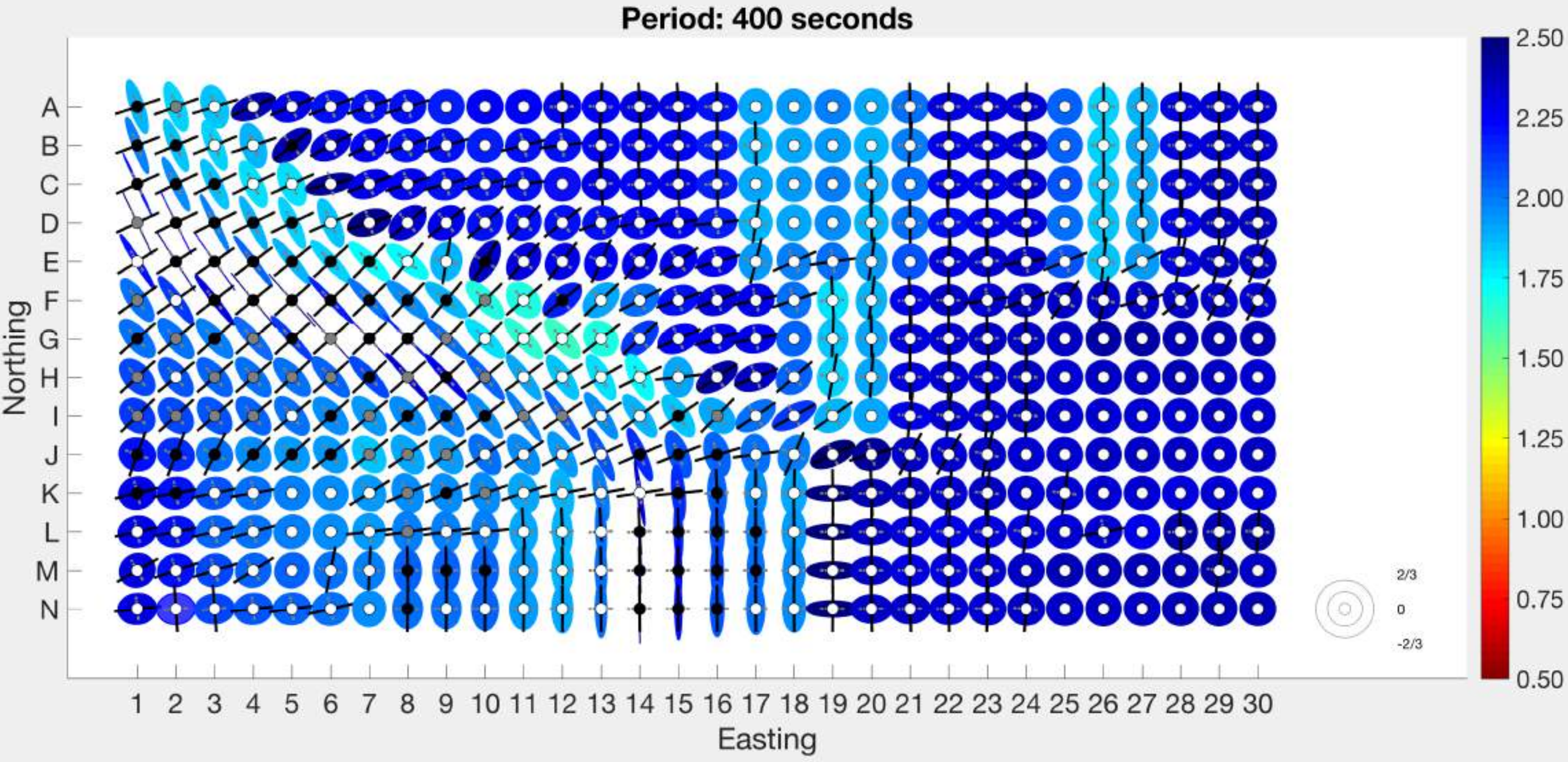}
	\caption{Phase Tensor (a) and Amplitude Tensor (b) parameters at periods of $40\,\mathrm{s}$ (left) and $400\,\mathrm{s}$ (right) obtained using the 3D Modelling and Inversion Workshop (2016) Secret Model 3. Tensor parameters are illustrated by scale (degrees and logarithmic apparent resistivity in colour), skew (circle fill - black: 3D, grey: quasi 2D, white: 2D), anisotropy and strike angle (ellipse).}
	\label{fig:em3d3b}
\end{sidewaysfigure}
\begin{sidewaysfigure}[p]
	\centering
	a) Phase Tensor\\
		\includegraphics[trim=0 50 200 0,  clip,height=150bp]{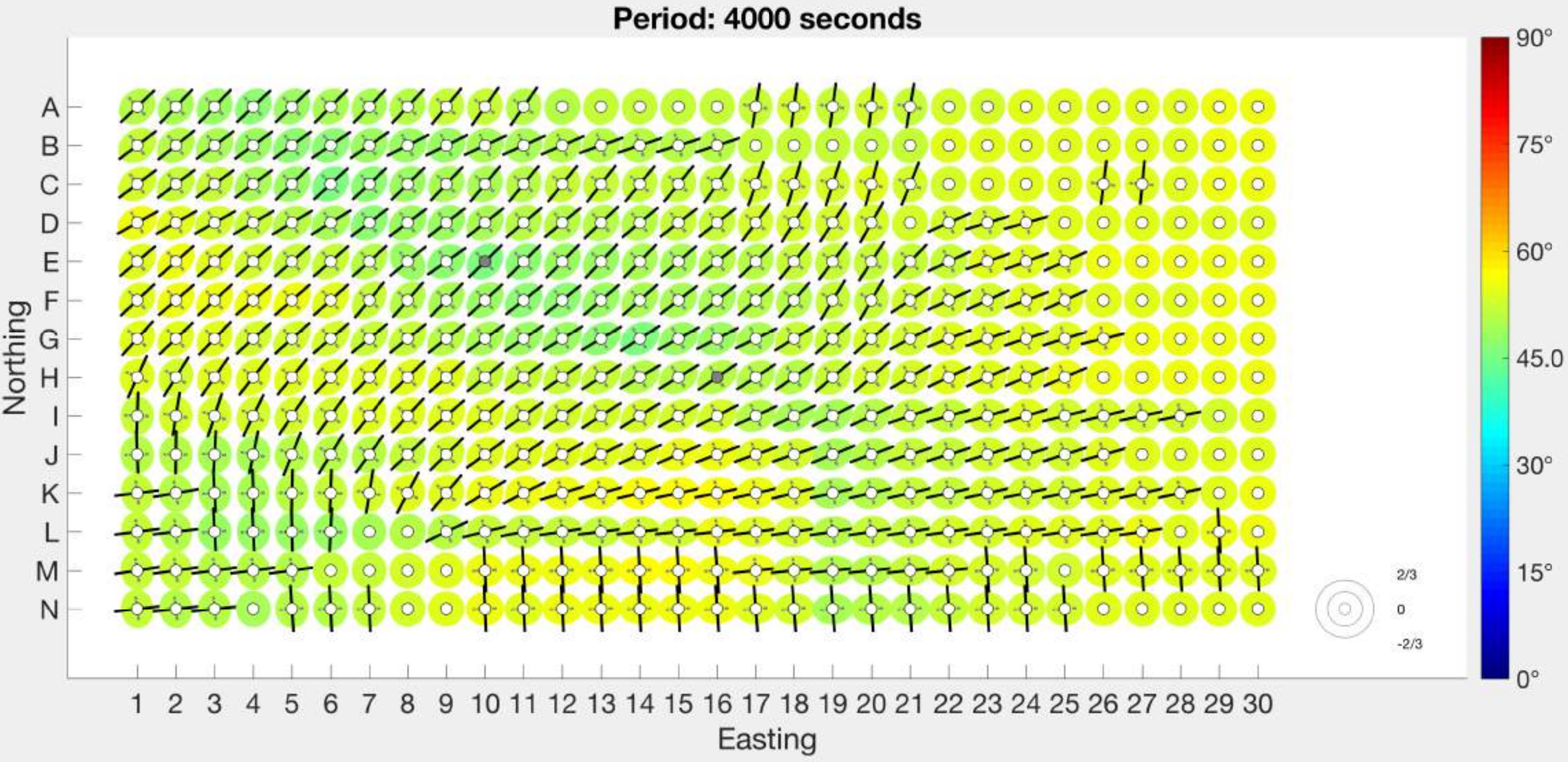}
		\includegraphics[trim=50 50 0 0,    clip,height=150bp]{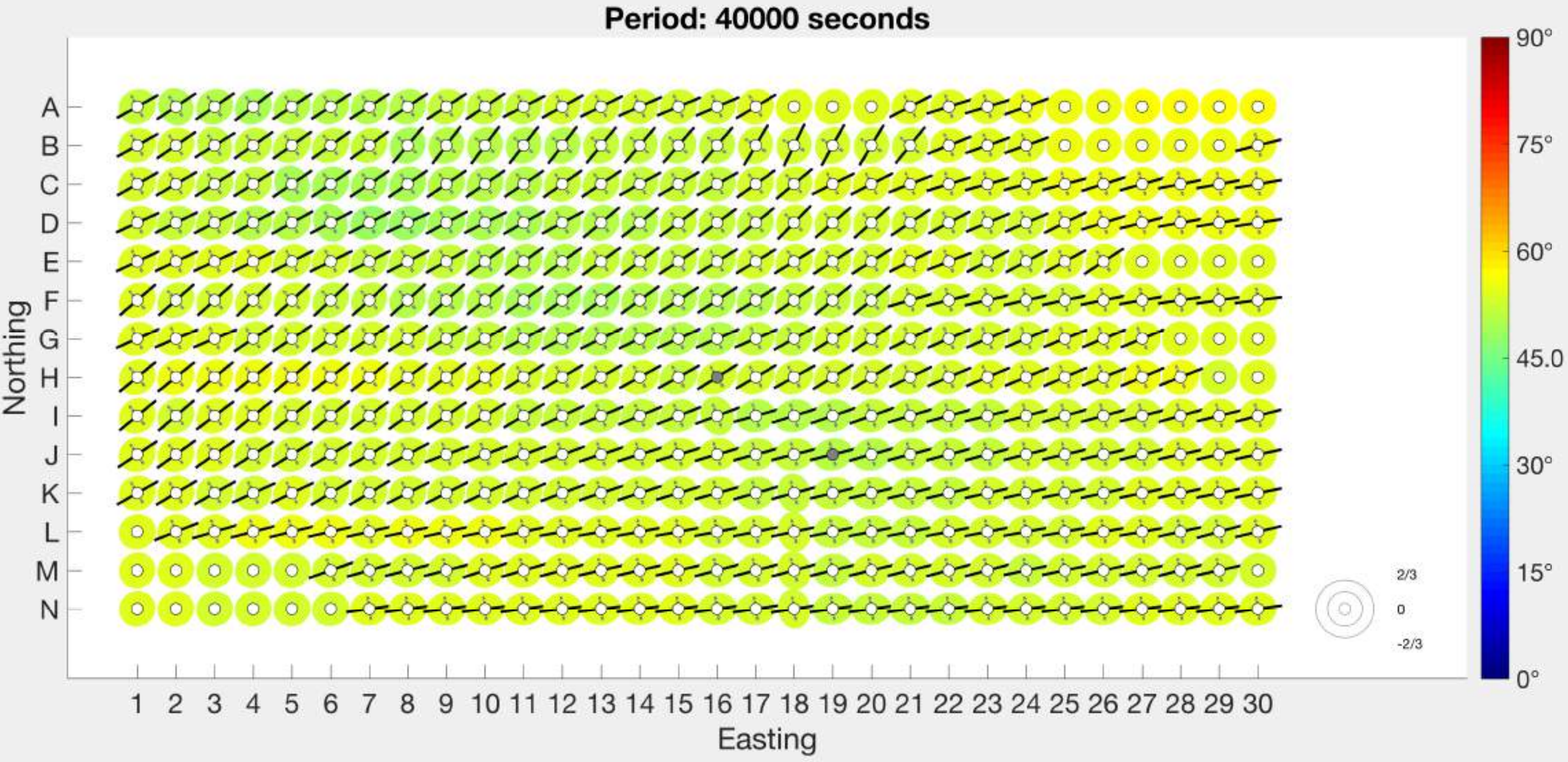}\\
	b) Amplitude Tensor\\
		\includegraphics[trim=0 0 200 50,  clip,height=150bp]{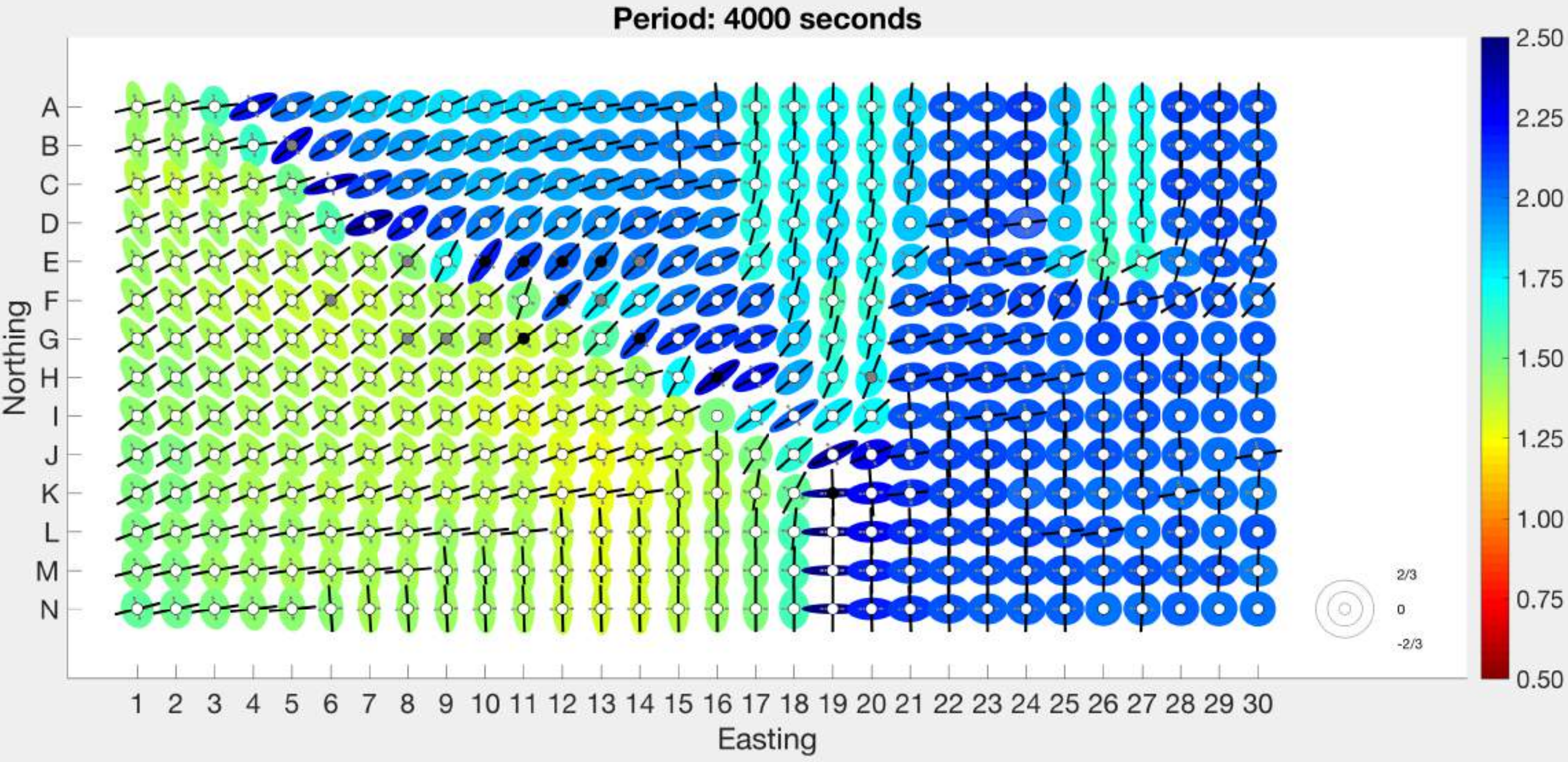}
		\includegraphics[trim=50 0 0 50,    clip,height=150bp]{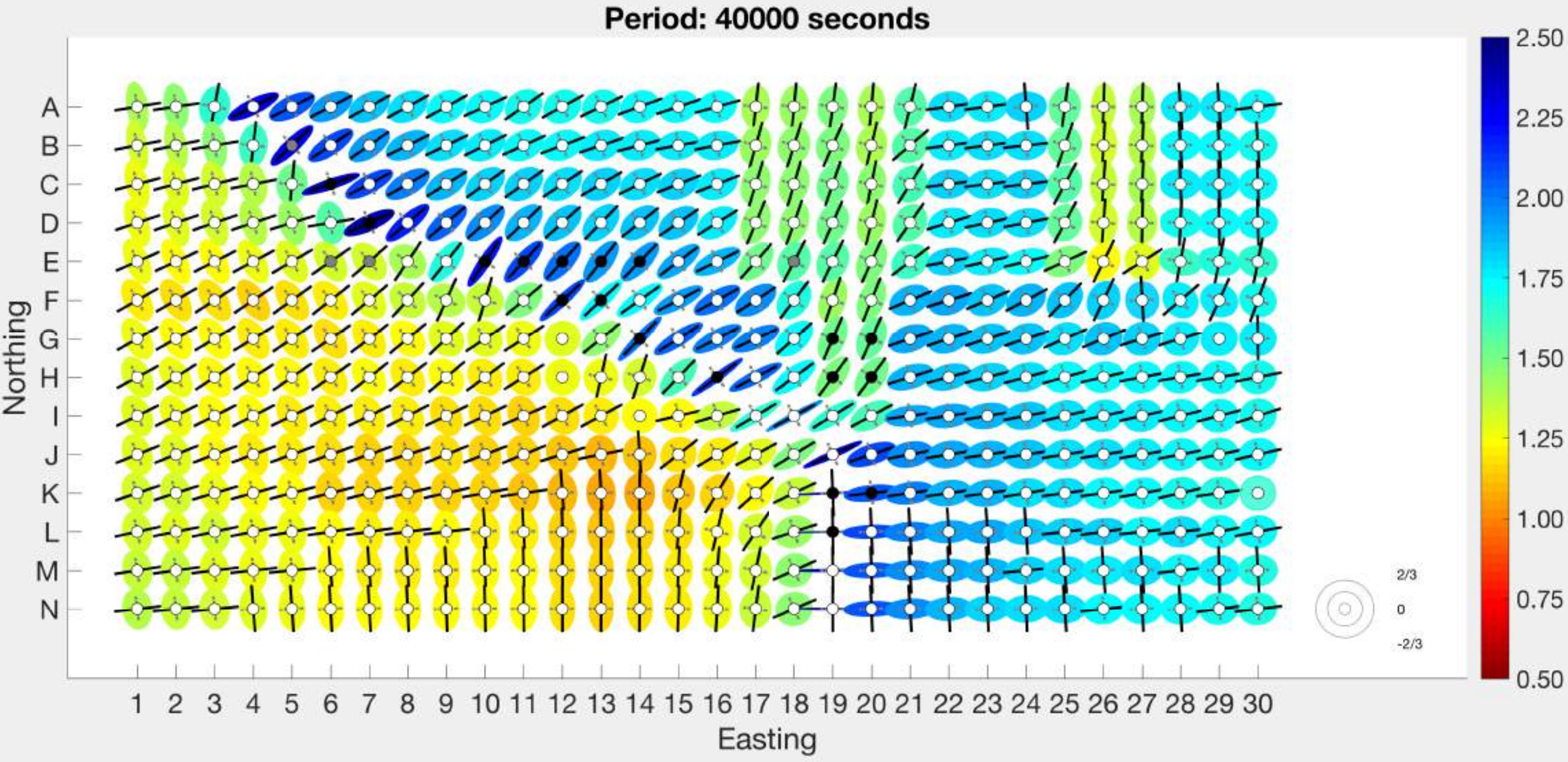}
	\caption{Phase Tensor (a) and Amplitude Tensor (b) parameters at periods of $4,\!000\,\mathrm{s}$ (left) and $40,\!000\,\mathrm{s}$ (right) obtained using the 3D Modelling and Inversion Workshop (2016) Secret Model 3. Tensor parameters are illustrated by scale (degrees and logarithmic apparent resistivity in colour), skew (circle fill - black: 3D, grey: quasi 2D, white: 2D), anisotropy and strike angle (ellipse).}
	\label{fig:em3d3c}
\end{sidewaysfigure}
We also interpret the secret model $\#3$, a large scale model that was discussed at the MT3DINV3 workshop held in 2016. The model represents an ocean-continent setting inspired by the US West coast, which includes bathymetry and topography and is noted as showing strong ocean effects. Therewith it contains more complex data that resembles more closely realistic data than the previous examples. 
Figures \ref{fig:em3d3a}, \ref{fig:em3d3b} and \ref{fig:em3d3c} display parameter map plots at a selection of periods from this model  
specifics of the data can be found on the website (\url{http://www.dias.ie/mt3dinv3/Home.html}) but model specifics are not yet published.

In Figure \ref{fig:em3d3a}, we plot the Phase and Amplitude Tensor parameters at the period of $0.4\,\mathrm{s}$. It can be observed that, at this period and at each station, the parameters of both tensors are almost all identical, with phase and amplitude scales of $45^\circ$ and $100\,\Omega m$, respectively, with indeterminable strike angle, zero anisotropy and zero skew angle, indicating that they both sense the local (homogeneous) subsurface.

At a period of $4\,\mathrm{s}$, still in Figure \ref{fig:em3d3a}, the phase scale parameter shows different shallow conductive anomalies (J3 to N3, J6-7 to N6-7, J11-13 to N11-13, F18-20 to I18-20 and A25-27 to E25-27), which are not yet obvious in the amplitude scale parameter, with a constant value of $100\,\Omega m$. Only on the region J11-13 to N11-13, the amplitude scale parameter is around $50\,\Omega m$, therefore, the conductivity anomaly in this region is likely shallower than the others. At the same period, we observe very large phase anisotropy values and large phase strike angle deviations which must be due to extreme coastal effects. However, it is not yet possible to clearly appreciate these effects in the amplitude parameters.

In Figure \ref{fig:em3d3b}, at a period of $40\,\mathrm{s}$, the Amplitude Tensor shows unmistakably the severe coastal effect represented by large anisotropy values and the different shallow anomalies 
represented by large phase anisotropy values, which could already be appreciated in the Phase Tensor parameters at $4\,\mathrm{s}$. 
At periods between $40\,\mathrm{s}$ and $400\,\mathrm{s}$, the large values of the phase and amplitude anisotropy parameters that we associate to the coastal effect continue to increase and additionally, in the region next to the aforementioned anomaly located between F18-20 to I18-20, the amplitude (phase) scale parameters decrease (increase), showing the presence of a new, larger and deeper conductive body to the North (A17-21 to E17-21).
At $400\,\mathrm{s}$, it is not possible to appreciate anymore the described responses of the phase scale parameter to the shallow conductive structures, but the effects are still present in the amplitude scale. However, at this period, we can observe another phenomenon on the region between J11-13 and N11-13, where the phase scale values have decreased below $45^\circ$ indicating the presence of a resistive body just in a position beneath a previously detected conductive body. 

In Figure \ref{fig:em3d3c}, at $4,\!000\,\mathrm{s}$, the large values of the anisotropy parameter associated to the coastal effect diminish for the Phase (reduced sensitivity) and reduces for the Amplitude (less dominance) Tensors. Moreover, at this period and also at $40,\!000\,\mathrm{s}$, the Phase Tensor strike angle is aligned to the coastline (note the $90^{\circ}$ ambiguity) which indicates that the underlying bedrock is either 1D anisotropic or features a 2D conductivity distribution. On the other hand, at 
$4,\!000\,\mathrm{s}$ and $40,\!000\,\mathrm{s}$, the Amplitude Tensor scale parameter has two areas of distinct values, in one area it lies between $10\,\Omega m$ and $20\,\Omega m$ and in the other it is between $60\,\Omega m$ and $100\,\Omega m$, clearly distinguishing the resistive continental and the conductive marine regions of the model. Further, at these periods, the Amplitude Tensor still shows the dominant features of the model. These are some conductive anomalies located between A17-21 and between A25-27 and E-25-27, which are visible through the low value of the amplitude scale parameter. We first detected them at the period of $40\,\mathrm{s}$ and therefore they are shallower than the depth that is penetrated at the present periods. Thus, it is remarkable that they are still visible even at $40,\!000\,\mathrm{s}$ where almost all Phase Tensor scale values are homogenous indicating the homogeneous and conductive bedrock. 
The dominant character of these conductive anomalies in the impedance responses could be associated to relatively high conductivity or due to large thickness. 

\subsubsection{Anisotropic Model}
\begin{figure}[h]
	\centering
	\includegraphics[width=1\textwidth]{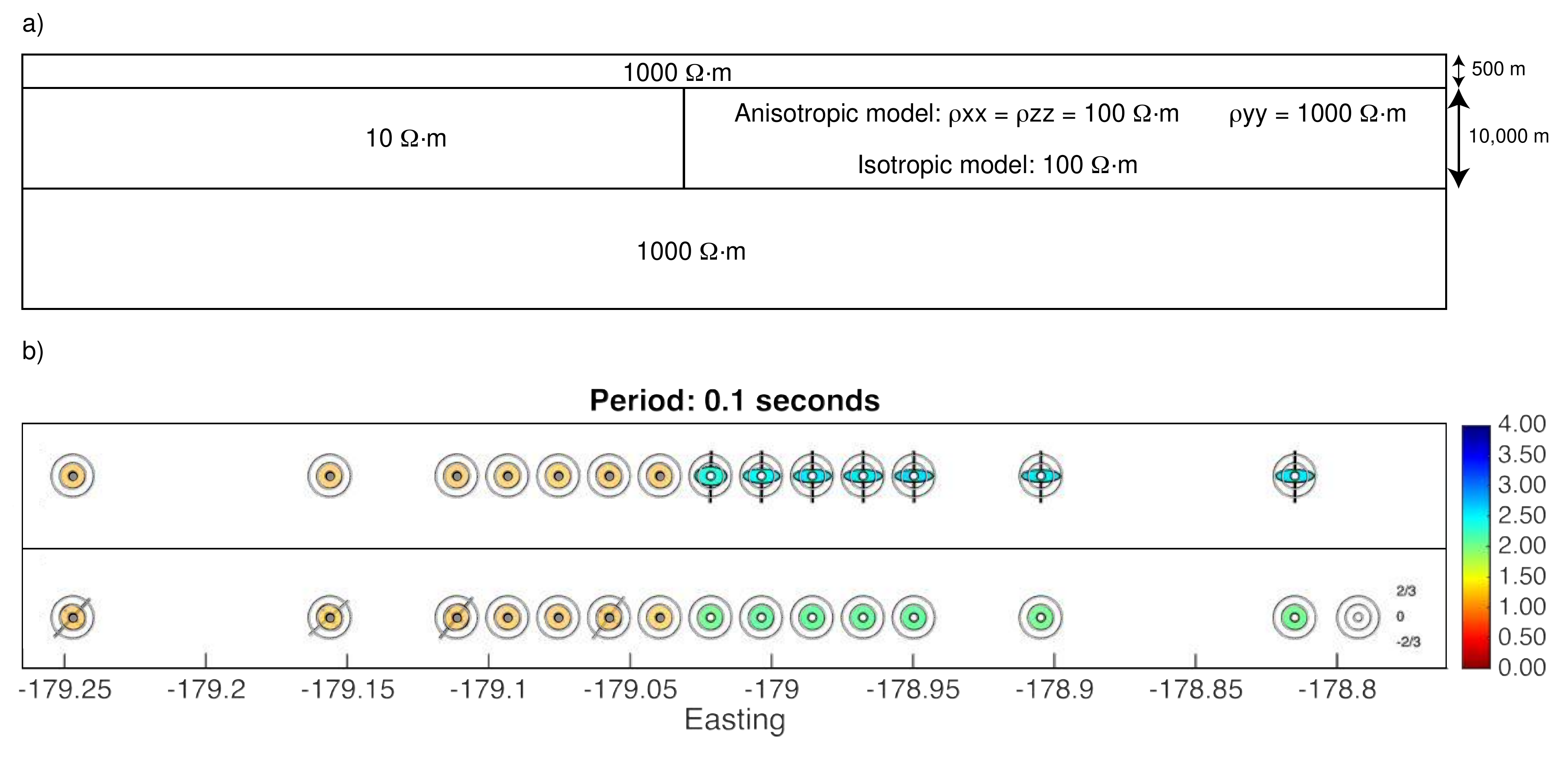}
	\caption{a) Model used to test anisotropy ($x$: into plane, $y$: along profile, and $z$: depth). b) Amplitude Tensor parameter map of the scale (colour) and anisotropy (ellipse) from the anisotropic model (top), compared to the isotropic model (bottom) for a period of $0.1\,\mathrm{s}$. }
	\label{ANI}
\end{figure}
Detection of anisotropy is an important factor in MT data \citep{Marti:2014}. To test the potential
of the Amplitude Tensor to detect microscopic anisotropy we have extracted its parameters
and focused on this feature. Figure \ref{ANI}a shows the anisotropic model we have used for this purpose.
It consists of a $1,\!000\,\Omega \mathrm{m}$ resistive top layer of $500\,\mathrm{m}$ thickness,
followed by two quarter space structures of $10,\!000\,\mathrm{m}$ thickness on top of a $1,\!000\,\Omega \mathrm{m}$ half-space. One of the middle quarter
spaces is $10\,\Omega \mathrm{m}$ isotropic and the second one is anisotropic with
resistivities $\rho_{xx} = 100\,\Omega \mathrm{m}$ and $\rho_{yy} =
1,\!000\,\Omega \mathrm{m}$, with $y$ pointing East. 
We also considered a second model with the same resistivity structures but substituting the anisotropic quarter space with an isotropic one with a resistivity of $100\Omega \mathrm{m}$.

In Figure \ref{ANI}b we show the Amplitude Tensor parameters from the responses of these two models at a frequency of $10\,\mathrm{Hz}$. Observing the parameters associated to each model we see that the parameter values differ on both sides of the vertical contrast line (at approximately $-179.03$ East) but are almost equal within the domain of either side, thus the quarter spaces of both models are well resolved and easy to distinguish for both, the isotropic and the anisotropic models. 
Comparing the top row of parameters (anisotropic model) to the lower one (isotropic model) in Figure \ref{ANI}b, it is clear which parameters correspond to which model (anisotropic or isotropic). The top row parameters show elongated ellipses in the position where the anisotropic quarter space is located in this model, indicating the anisotropy. Contrarily, in the bottom row, the anisotropic parameter in the same positions is null (circles), indicating isotropy in the quarter space of the second isotropic model. 
This technique to pin-point anisotropy was introduced by \citet{Hamilton.ea:2006} within the framework of the SAMTEX experiment.

\subsubsection{Real Data}
\begin{sidewaysfigure}[p]
	\centering
	\includegraphics[height=300bp]{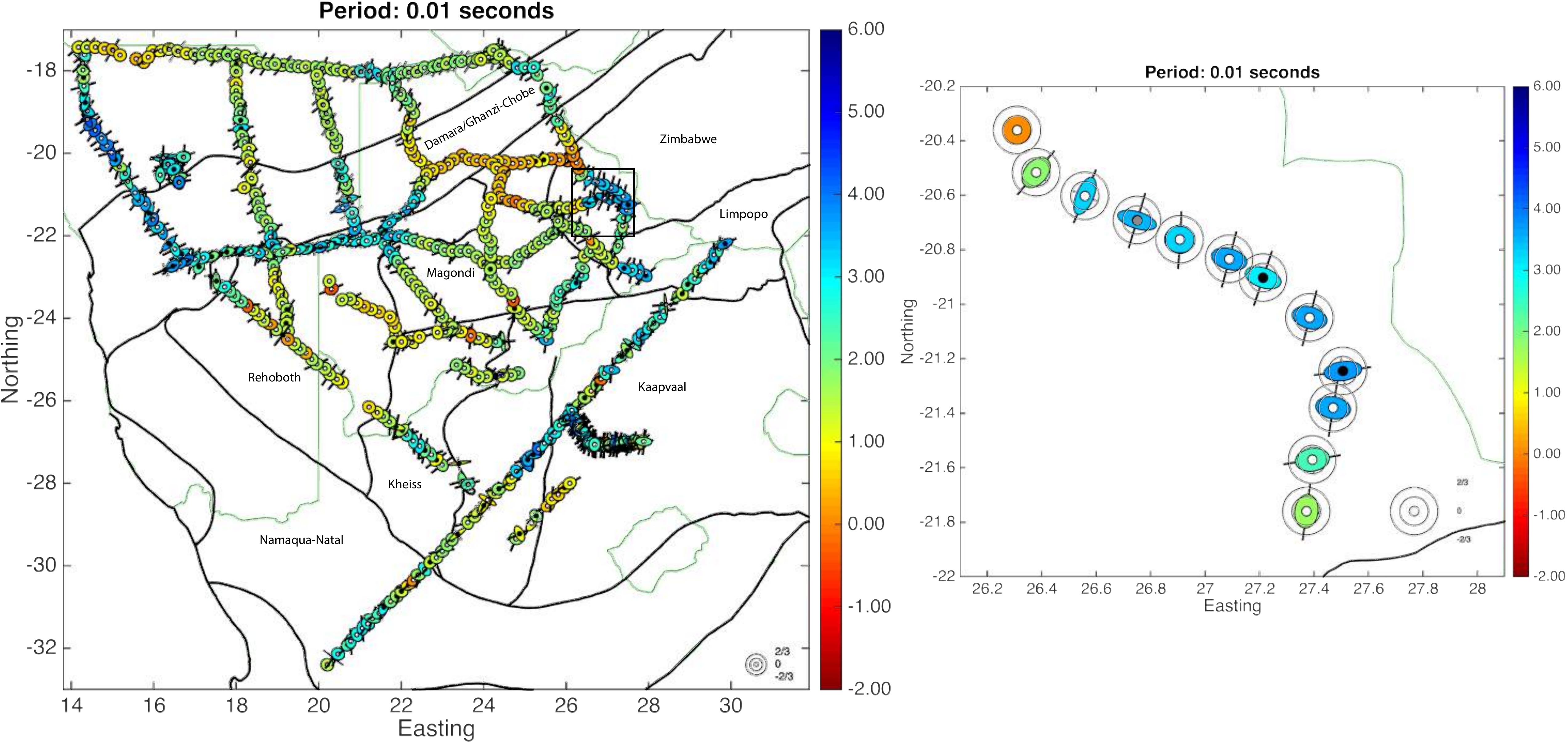}
	\caption{Amplitude Tensor map of the SAMTEX experiment illustrated by scale (logarithmic apparent resistivity in colour), skew (circle fill - black: 3D, grey: quasi 2D, white: 2D), anisotropy and strike angle (ellipse). On the right, there is a zoom of a few sites of an area located in the 
	rectangle on the left plot.}
	\label{samtex}
\end{sidewaysfigure}
Finally, we analyse a field example to show the potential of the Amplitude Tensor for a qualitative interpretation of a real data set. We consider a dataset from the South African {MagnetoTelluric} {EXperiment} ({SAMTEX}).
The {SAMTEX} was an international MT study to delineate and map the
Kaapvaal craton and adjacent terranes. The whole dataset consists of over 740
sites \citep{Jones:2009}, ranging from some $1,\!000\,\mathrm{Hz}$ up to $10,\!000\,\mathrm{s}$. Details of the survey and interpretation have been published elsewhere and we
refer to \citet{Jones:2009} for more information.

Figure \ref{samtex} shows maps of Amplitude Tensor parameters at $100\,\mathrm{Hz}$ derived from this dataset. The amplitude of the apparent resistivities, represented by the
colour of the ellipses, shows a good correlation between Kaapvaal and Zimbabwe Archaean
terranes and Ghanzi-Chobe belt and higher resistivity values (cooler colours) and lower
resistivity values (warmer colours) for some of the Phanerozoic terranes, like the Rehoboth
terrane, already pointed out by \citet{evans2011electrical}. In particular, the
northern Rehoboth terrane border and the data collected in Namibia, on the
Damara and Ghanzi-Chobe belts, are marked by high resistivities, which
correlates well with inversion results from \citet{Muller.ea:2009} and \citet{Khoza:2013}. As
mentioned in the previous paragraph, \citet{Hamilton.ea:2006} plotted the most
conductive direction for the SW-NE profile crossing the Kaapvaal craton. A
direct comparison of the results from \citet{Hamilton.ea:2006} and ours is not
possible since these authors averaged the most conductive directions for crustal
and lithospheric depths, while we present the most conductive direction
(anisotropy) for a specific very short period that is only sampling the upper and middle crust.

Focusing now on the eastern edge of the area covered by the experiment, we
inspect data analysis results from the southwestern border of the Zimbabwe
craton (Figure \ref{samtex} and zoom image). The Magondi terrane and the
northwestern edge of the Zimbabwe cratons are marked by high conductivities,
which have been associated with saline fluids \citep{Miensopust.ea:2011}. Another
of the tectonic features at this location that has a great impact on MT data is the Okavango dike
swarm. \citet{Miensopust.ea:2011} pointed out a couple of effects of the
presence of the dike swarm, first the high resistivities associated to the
dolorites that make the dikes ($30,\!000\,\Omega\mathrm{m}$), which is also observed in our
analysis, that shows a very conductive upper crust to the north of the Magondi
terrane and adjacent terranes. Second, the authors suggest that there is an
anisotropic structure at middle to low crustal depths. Our analysis (see zoom
image in Figure \ref{samtex}) shows that the circles become ellipses, indicative
of anisotropy, at the same location. Our analysis at $100\,\mathrm{Hz}$, which corresponds
to a skin depth of about $5\,\mathrm{km}$, already can sense the anisotropy caused by the
dolorite intrusions (see zoom plot in Figure \ref{samtex}), and analysis at
longer periods (see supplements of the zoom image) show that the anisotropy extends to the lower
crust. As seen in the previous example, according to the orientation of the
ellipses, the most conductive direction of the anisotropy is approximately NS (cp.~zoom image),
except to the Northwest where the direction changes to a more EW angle.

To the West, along the boundary between the Damara/Ghanzi-Chobe belt and the
Rehoboth terrane the data show anisotropy (as mentioned earlier, circles
become ellipses) with an approximate EW trend. Any modelling of these data should
take into account this anisotropy to avoid interpreting artifacts, as the ones pointed out
by \citet{Miensopust.ea:2011}.

\section{Conclusions}
We propose a new MT impedance tensor decomposition into the known Phase Tensor and a newly defined Amplitude Tensor. The tensor decomposition, and in particular the information contained in the Amplitude Tensor, is presented and discussed with synthetic and real data. Analogies to the Phase Tensor are drawn with respect to dimensionality and it is laid out
that the Amplitude Tensor holds information about both galvanic and inductive
effects. 
Especially, it is shown that galvanic effects are present at all period ranges in the Amplitude Tensor but only the galvanic part from the shortest periods should be considered as distortion. 
The galvanic effect at longer periods represents interpretable information about subsurface structure, which is not present in the Phase Tensor and is lost when the Phase Tensor is used solely for interpretation. 
Therefore, the formulation of the Amplitude Tensor bridges the gap between Phase and Impedance Tensors and enables the possibility to interpret Phase Tensor and impedance amplitude information without redundancy, which is present, i.e.~by instead complimenting the Phase Tensor with apparent resistivity data. 

Since the Phase Tensor is unaffected by galvanic distortion of the electric field, all such distortion must be present in the Amplitude Tensor. We demonstrate that the Amplitude Tensor itself also contains inductive information about the subsurface, just as the Phase Tensor, and that, because this information represents the same subsurface geometry, it must be similar to one contained in the Phase Tensor. Through that coupling, we postulate the hypothesis that the galvanic amplitude can be separated from the inductive effects present in the Amplitude Tensor and therewith, the galvanic distortion matrix can be estimated up to a single constant, usually denoted as galvanic shift or site gain $g$. We show exemplarily that, by using the Amplitude Tensor, the electric galvanic distortion anisotropy can be estimated explicitly for 2D subsurfaces, since anisotropy of inductive phase and amplitude must correspond to the same subsurface but the tensor singular values may differ greatly due to anisotropic galvanic distortion. 

For the proposed approximation of the inductive Amplitude Tensor and therewith the estimation of the galvanic electric distortion matrix from the galvanic Amplitude Tensor, we break with the principal assumptions of nearly all distortion analysis methods and do not constrain the sought after impedance to any dimensionality. Therefore, we offer a new way to explore galvanic effects (including electric galvanic distortion) in the impedance data for general subsurfaces. This shows that the introduction of the Amplitude Tensor opens a new perspective on distortion analysis. 

The Amplitude Tensor presented in this work is clearly complementary to the Phase Tensor as it contains integrated information about inductive and galvanic effects of the subsurface. Therewith, it can also be used as an additional source of independent information for Phase Tensor inversions. 
Furthermore, we demonstrate that map plots of the Amplitude Tensor and of the Phase Tensor parameters at different periods can be directly interpreted to obtain qualitative information of the subsurface structure. Moreover, we show that the information added by the Amplitude Tensor can complement the information of the Phase Tensor for a more accurate interpretation. Thus, 
map plots of the Amplitude Tensor can be useful, for instance, to design a preliminary model before starting a computationally expensive 3D inversion.

\section{Supplements}
\label{sec:supp}
A MatLab executable to compute Amplitude and Phase Tensor parameters is available  at:
\begin{center}
\url{https://www.dropbox.com/sh/16ldi7omr99uvvd/AADNnjL8RuugEWiNupaKgAh3a?dl=0}
\end{center}
\noindent for academic use. At the same location, parameter maps of the models discussed in this work are available for a longer period range for the interested reader.


\section*{Acknowledgments}
This work was funded by Repsol under the framework of the CO-DOS project. 
All MT3D workshop datasets are publicly available at \begin{center} \url{http://www.complete-mt-solutions.com/mtnet/workshops/em_workshops.html#3DMTINV}\end{center} for download and further details. For this we especially thank the organisers of the workshops and the creators of the models, Alan Jones, Marion Miensopust, Pilar Queralt and Jan Vozar. The constructive criticism of the reviewers John Booker, Alan Jones, Anna Avdeeva and two anonymous reviewers, and of the editors Gary Egbert and Ute Weckmann have greatly improved earlier versions of this paper. 

\appendix

\section{Proofs for the Amplitude-Phase Decomposition}
\label{app:proofs}
Let $\mathbf{M}\in\mathbb{C}^{n\times n}$ with $\mathbf{M}=\mathbf{X}+i\mathbf{Y}$, where
 $\mathbf{X},\mathbf{Y}\in \mathbb{R}^{n\times n}$ are invertible.
 Set $\mathbf{\Phi} = \mathbf{X}^{-1} \mathbf{Y}$ and:
  \begin{align*}
  \mathbf{C}_0 & = (\mathbb{I}+\mathbf{\Phi \Phi}^T)^{-1/2},\\
  \mathbf{S}_0 & = \mathbf{C}_0 \mathbf{\Phi},\\
  \mathbf{P}_0 & = \mathbf{X C}_0^{-1} = \mathbf{Y S}_0^{-1},
 \end{align*} 
  where $\mathbb{I}\in \mathbb{R}^{n\times n}$ denotes the identity matrix.

\begin{thm}  \label{thm}
 For any orthogonal $\mathbf{U}\in \mathbb{R}^{n\times n}$ the matrices,
   \[
  \mathbf{C} = \mathbf{U C}_0,\quad \mathbf{S}= \mathbf{U S}_0,\quad \mathbf{P}= \mathbf{P}_0 \mathbf{U}^T
 \] 
 satisfy:
 \begin{align}
\label{1}
      \mathbf{S} & = \mathbf{C\Phi},\\
      \label{2}
  \mathbf{M} & = \mathbf{P}(\mathbf{C}+i\mathbf{S}),\\
  \label{3}
  \mathbb{I} &= \mathbf{CC}^T + \mathbf{SS}^T.
\end{align} 
\end{thm}

\begin{proof}
 By $\mathbf{S}_0 = \mathbf{C}_0 \mathbf{\Phi}$ we obviously have \eqref{1}. Equation \eqref{3} follows from
 $\mathbf{C}_0 = \mathbf{C}_0^T$, $\mathbf{UU}^T = \mathbb{I}$ and
  \[
  \mathbf{CC}^T + \mathbf{SS}^T
= \mathbf{UC}_0(\mathbb{I}+\mathbf{\Phi \Phi}^T) \mathbf{C}_0^T \mathbf{U}^T = \mathbb{I}.
 \] 
 Since $\mathbf{U}^T \mathbf{U}=\mathbb{I}$ we obtain \eqref{2} by:
  \[
 \mathbf{P}(\mathbf{C}+i\mathbf{S})= \mathbf{P}_0 \mathbf{U}^T \mathbf{U C}_0 (\mathbb{I}+i\mathbf{\Phi}) = \mathbf{X}( \mathbb{I} + i \mathbf{X}^{-1}\mathbf{Y}) = \mathbf{M}.\qedhere
 \] 
 \end{proof}

\begin{thm}  \label{thm2}
 For matrices $\mathbf{C},\mathbf{S},\mathbf{P}\in \mathbb{R}^{n\times n}$ which satisfy
 \eqref{1}, \eqref{2} and \eqref{3},
 there exists an orthogonal matrix $\mathbf{U}\in \mathbb{R}^{n\times n}$ so that:
  \[
  \mathbf{C} = \mathbf{U C}_0,\quad \mathbf{S}= \mathbf{U S}_0,\quad  \mathbf{P}= \mathbf{P}_0 \mathbf{U}^T.
 \] 
\end{thm}
\begin{proof}
 We set $\mathbf{U} = \mathbf{C C}_0^{-1}$ and verify that
 $\mathbf{U}$ is an orthogonal matrix with
  $\mathbf{C} = \mathbf{U C}_0$, $ \mathbf{S}= \mathbf{U S}_0$ and $ \mathbf{P}= \mathbf{P}_0  \mathbf{U}^T.$
 We have by $\mathbf{C}_0^{-1}= (\mathbf{C}_0^{-1})^T$, \eqref{1} and \eqref{3} the orthogonality:
  \[
 \mathbf{U  U}^T =  \mathbf{C C}_0^{-1} \mathbf{C}_0^{-1}  \mathbf{C}^T
 =  \mathbf{C} (\mathbb{I} + \mathbf{\Phi \Phi}^T)  \mathbf{C}^T =  \mathbf{C  C}^T +  \mathbf{S  S}^T = \mathbb{I}.
 \] 
  Then, clearly $\mathbf{C} =  \mathbf{U C}_0$ and by \eqref{1} holds:
   \begin{align*}
   \mathbf{U S}_0 =  \mathbf{C C}_0^{-1} \mathbf{S}_0 =  \mathbf{C \Phi} = \mathbf{S}.
  \end{align*} 
 Finally, from \eqref{2} follows $ \mathbf{P  C}= \mathbf{X}$ and
  \[
  \mathbf{P}_0 = \mathbf{X C}_0^{-1} =  \mathbf{P  C C}_0^{-1} =  \mathbf{P  U},
 \] 
 which leads by the orthogonality of $ \mathbf{U}$ to $\mathbf{P}_0  \mathbf{U}^T =  \mathbf{P}$.
\end{proof}

Theorem~\ref{thm} and Theorem~\ref{thm2} show that the solutions of \eqref{1}, \eqref{2} and \eqref{3}
are given up to orthogonal transformations by $\mathbf{C}_0$, $\mathbf{S}_0$ and $\mathbf{P}_0$. An advantage of $\mathbf{C}_0$ and $\mathbf{S}_0$ compared to orthogonal transformations of those matrices lies in the fact that the singular value decomposition (SVD) of $\mathbf{\Phi}$ carries to some extent over to $\mathbf{C}_0$ and $\mathbf{S}_0$. To see this, let us denote the SVD of $\mathbf{\Phi}$ by:
  \[
  \mathbf{\Phi} = \mathbf{V}_\Phi \mathbf{\Sigma}_\Phi \mathbf{W}_\Phi^T,
 \] 
 where $\mathbf{V}_\Phi,\mathbf{W}_\Phi\in\mathbb{R}^{n\times n}$ are orthogonal matrices
 and $\mathbf{\Sigma}_\Phi= \rm diag(\sigma_1,\dots,\sigma_n)$ is a diagonal matrix, containing the singular values $\sigma_i \geq 0$ for $i=1,\dots,n$ of $\mathbf{\Phi}$.
 Then, one easily recovers the following expressions due to the matrix square root properties as discussed below:
  \begin{align*}
  \mathbf{C}_0 & = \mathbf{V}_\Phi (\mathbb{I}+\mathbf{\Sigma}_\Phi^2)^{-1/2} \mathbf{V}_\Phi^T, \\
  \mathbf{S}_0 & = \mathbf{V}_\Phi(\mathbb{I}+\mathbf{\Sigma}_\Phi^{-2})^{-1/2} \mathbf{W}_\Phi^T.
 \end{align*} 

In the expressions above we have used that the matrix square root applied to the singular value decomposition of any symmetric matrix $\mathbf{C}\in\mathbb{R}^{n\times n}$ yields:
 \begin{equation}
\sqrt{\mathbf{C}}=\sqrt{\mathbf{U}_C  \mathbf{\Sigma}_C   \mathbf{U}_C^T}=\mathbf{U}_C  \sqrt{\mathbf{\Sigma}_C}   \mathbf{U}_C^T,
\end{equation} 
with $\mathbf{U}_C\in\mathbb{R}^{n\times n}$ being the orthogonal SVD matrix and the diagonal matrix $\mathbf{\Sigma}_C$ containing the singular values of $\mathbf{C}$. This result can be readily corroborated by: 
 \begin{equation}
\mathbf{C}=\sqrt{\mathbf{C}}\sqrt{\mathbf{C}}=\mathbf{U}_C  \sqrt{\mathbf{\Sigma}_C}   \mathbf{U}_C^T\mathbf{U}_C  \sqrt{\mathbf{\Sigma}_C}   \mathbf{U}_C^T=\mathbf{U}_C \mathbf{\Sigma}_C \mathbf{U}_C^T.
\end{equation} 
Due to this property, the singular values of $\mathbf{C}=c(\mathbf{\Phi})$, $\mathbf{S}=s(\mathbf{\Phi})$ and $\mathbf{\Phi}$ can be linked directly and independently by trigonometric identities that relate in squares and square roots such as half-angle and double-angle formulae, and thus, allows to formulate analogies to these identities for the phase matrix $\mathbf{\Phi}$. 

\section{Logarithmic Amplitude and Phase Anisotropy}
\label{sec:logAniso}
\begin{figure}[t]
		\centering
		a) Relative anisotropy\\
		\includegraphics[width=.58\textwidth]{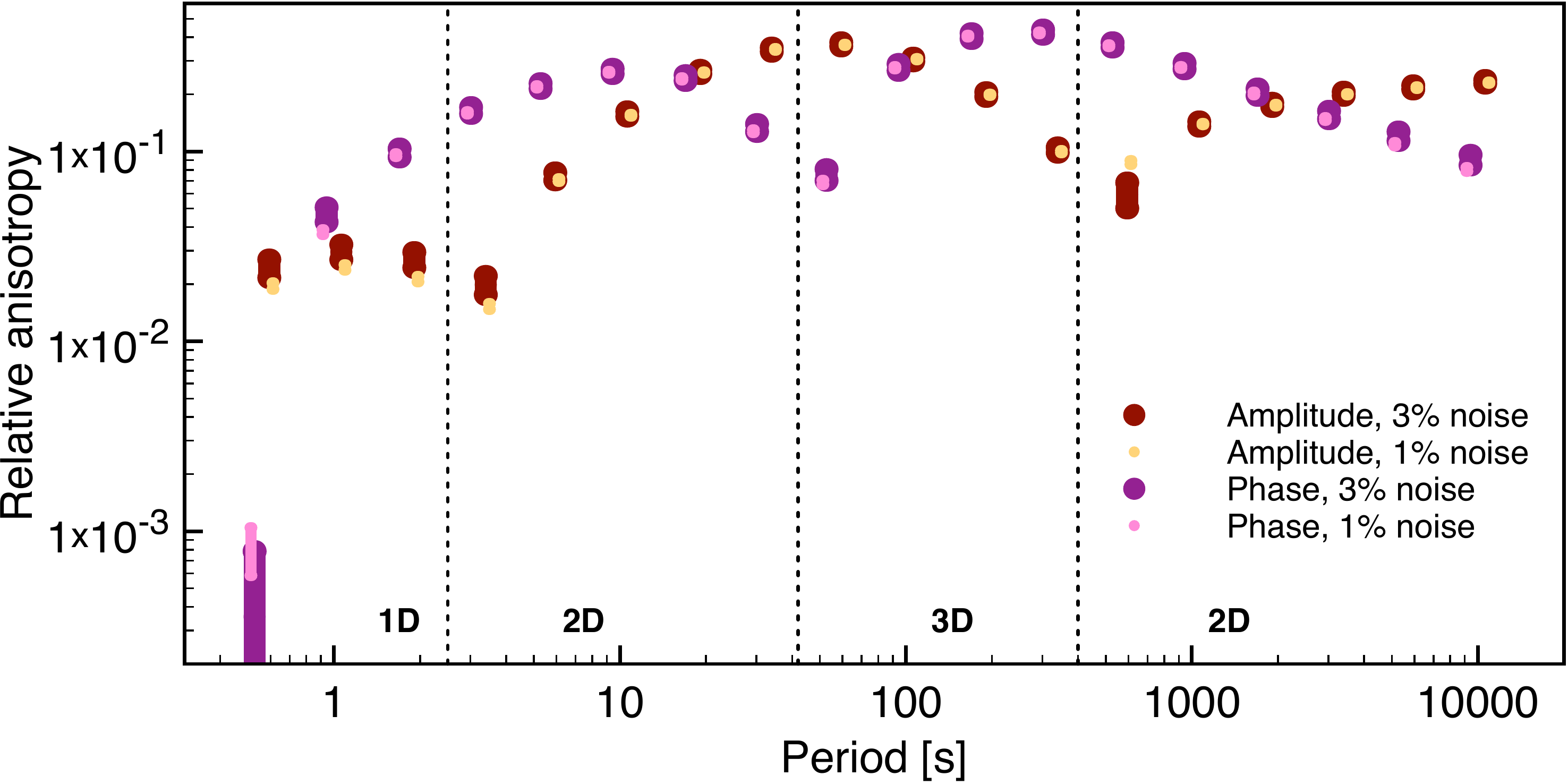}\\
		b) Moving average of relative anisotropy\\
		\includegraphics[width=.58\textwidth]{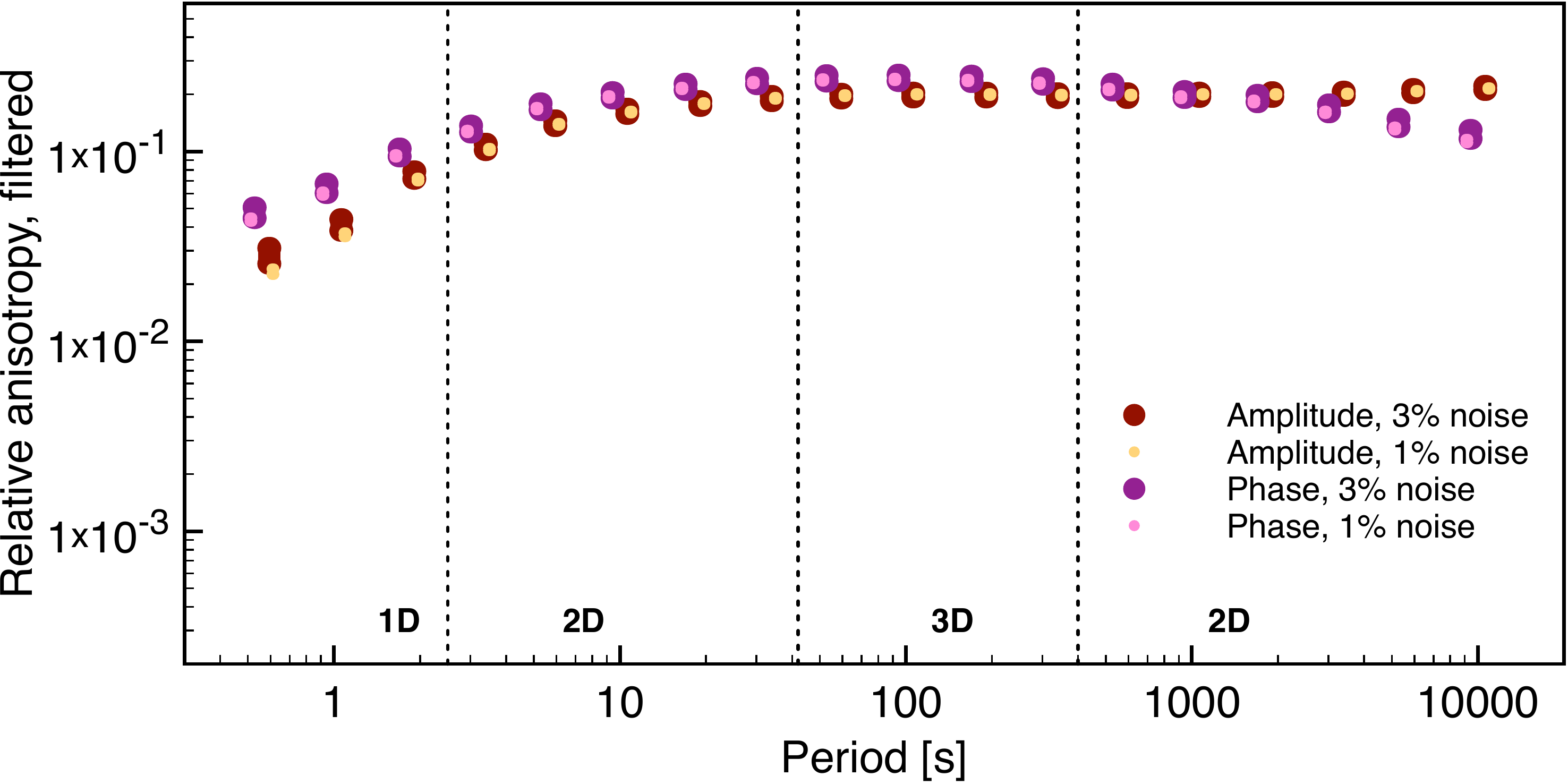}
	\caption{Impedance data of site \emph{A09} of the DSM1 \citep{Miensopust:2013} is decomposed into Phase and Amplitude Tensor singular values. From these singular values, the relative anisotropies are computed according to \eqref{eq:rPArAA}, and the moving average results from a two-decade gaussian moving window filter. To illustrate the statistical stability of the decomposition, a confidence limit of $1\%$ and $3\%$ was added to the impedance based on the impedance determinants square root. Periods of different data are slightly shifted for comparison's sake.}
	\label{fig:exampleAmpPhaseParam2}
\end{figure}
We refer to anisotropy as macroscopic anisotropy as it is observed in measured data at a certain period  
in contrast to microscopic anisotropy that can be considered in the subsurface model parameters. 
Let us define two new parameters, the phase anisotropy, $\phi_a$, and the logarithmic amplitude anisotropy, $\rho_a$, which describe the anisotropy in the measured Phase and Amplitude Tensors, respectively:
 \begin{equation}
  \phi_a =\frac{1}{2}(\arctan\phi_1-\arctan\phi_2) \qquad \mathrm{and}\qquad
  \rho_a = \frac{1}{2}(\ln\rho_1-\ln\rho_2).
  \label{eq:rPArAA}
\end{equation} 
With these definitions, the parameters can be compared like the real and imaginary parts of the logarithm of a complex number $z=\varrho\exp(i\varphi)$, that is 
  $\ln z=ln (\varrho\exp(i\varphi))=\ln\varrho+i\varphi$.
The motivation to define comparable parameters of the macroscopic anisotropy is that 
we want to approximate the logarithmic amplitude anisotropy of the inductive Amplitude Tensor with the phase anisotropy of the Phase Tensor as described in section \ref{sec:ApproxIndAmp}.

In order to illustrate the comparison of the anisotropy parameters between the Phase and the Amplitude Tensors, let us return to the previous example of site $A09$ from the DSM1 data set (cp.~Figure \ref{fig:exampleAmpPhaseParam})
, for which we display the phase and the logarithmic amplitude anisotropy parameters in Figure \ref{fig:exampleAmpPhaseParam2}a.  
The plotted phase and amplitude anisotropy parameters display the same behaviour over period, i.e.~the maxima, minima and general shape of the curves are at the same order of magnitude, but one curve appears to be constantly offset with respect to the other by approximately one decade.
From these parameters we cannot clearly associate any particular behaviour to a certain subsurface dimensionality as we did for the other parameters of the Phase and Amplitude Tensors (cp.~section \ref{sec:AmpDim}). Certainly, if the subsurface is absolutely 1D isotropic, the parameters defined above would be zero, but for any other case it is impossible to predict any particular value. 
Based on the observed offset of one decade, we define a two decade Gaussian moving window 
average for both anisotropy parameters, which shows that both averages are nearly identical in Figure \ref{fig:exampleAmpPhaseParam2}b. Note that due to the window width the period where the window is well-defined only ranges from $3$ to $2,\!000\,\mathrm{s}$. From the observed similarity of these curves we conclude that the average logarithmic amplitude anisotropy can be approximated with the  average phase anisotropy and propose that such procedure can be used to estimate the galvanic logarithmic amplitude anisotropy as we discuss in section \ref{sec:RecAniDis2D}. 

Consider the principal values, $\rho_1$, $\rho_2$ and $\phi_1$, $\phi_2$ of the Amplitude
and Phase Tensors, respectively. Defining two new variables $\rho$ and $\phi$ as:
 \begin{equation}
\rho=\sqrt{\rho_1\rho_2}=\exp\left(\frac{1}{2}(\ln\rho_1+\ln\rho_2)\right) \quad \mathrm{and} \quad
\phi=\frac{1}{2}(\arctan\phi_1+\arctan\phi_2),
\label{eq:detRhoPhi}
\end{equation} 
and considering them together with \eqref{eq:rPArAA}, we obtain $\rho_1$, $\rho_2$, $\phi_1$ and $\phi_2$ in terms of $\rho$, $\rho_a$, $\phi$ and $\phi_a$: 
 \begin{equation}
\ln\rho_1=\ln\rho+\rho_a,\quad
\ln\rho_2=\ln\rho-\rho_a,\quad
\phi_1=\tan\left(\phi+\phi_a\right)\quad\mathrm{and}\quad
\phi_2=\tan\left(\phi-\phi_a\right).
  \label{eq:impedanceSV}
\end{equation} 
Both Phase Tensor parameters, $\phi$ and $\phi_a$, are generally unaffected by galvanic electric field distortion, whereas $\rho$ and $\rho_a$ can include galvanic shift 
and anisotropic distortion, respectively. 

\subsection{Galvanic and Inductive Amplitude Singular Values for 2D Impedances}
\label{sec:AmpSV}
For the case of a regional 2D impedance, we can derive an explicit formula for the galvanic amplitude singular values, which we can use to estimate the anisotropic distortion parameter. For this, consider an impedance tensor in the coordinate system of the regional strike, with the two principal modes:
 \begin{equation}
  Z_1=\exp(\ln\rho_1+i\arctan\phi_1)
  \qquad\mathrm{and}\qquad
  Z_2=\exp(\ln\rho_2-i\arctan\phi_2),
  \label{eq:ImpPrincVal}
\end{equation} 
where $\rho_1$, $\rho_2$ and $\phi_1$, $\phi_2$ are the singular values of the Amplitude and Phase Tensors, respectively. Then, the logarithmic amplitude anisotropy parameter is a sum of inductive and galvanic logarithmic amplitude anisotropy parameters $\rho_a=\rho_a^{gal}+\rho_a^{ind}$ as a consequence of substituting the expressions of $\ln\rho_{1,2}$ from \eqref{eq:impedanceSV} in $\rho_{1,2}=\rho_{1,2}^{gal}\rho_ {1,2}^{ind}$, obtained from \eqref{eq:AmpGalInd} for 2D subsurfaces:
 \begin{equation}
\ln\rho_{1,2}=\ln\rho_{1,2}^{gal}+\ln\rho_{1,2}^{ind}=\ln\rho^\mathrm{gal}\pm\rho_a^\mathrm{gal}+\ln\rho^\mathrm{ind}\pm\rho_a^\mathrm{ind}=\ln\rho\pm(\rho_a^{gal}+\rho_a^{ind}),
  \label{eq:AmplitudeSVrelation}
\end{equation} 
where in the last equality we have combined the indeterminable $\rho^\mathrm{gal}$ and $\rho^\mathrm{ind}$ in $\rho$. With the approximation $\rho_a^{ind}\approx\phi_a$, based on the idea that the inductive subsurface response is sensed equally by the phase and inductive amplitude as discussed in section \ref{sec:ApproxIndAmp}, we can define the inductive amplitude singular values by:
 \begin{equation}
  \rho_1^{\mathrm{ind}}=\exp\left(\ln\rho+\phi_a\right)
  \qquad\mathrm{and}\qquad
  \rho_2^{\mathrm{ind}}=\exp\left(\ln\rho-\phi_a\right),
  \label{eq:indAmplitudeSV}
\end{equation} 
where we have identified the inductive part in \eqref{eq:AmplitudeSVrelation}, including the indeterminable galvanic shift $\ln\rho^\mathrm{gal}$. 
The expressions for the galvanic amplitude singular values can be derived similarly, substituting \eqref{eq:impedanceSV} and \eqref{eq:indAmplitudeSV} in \eqref{eq:AmplitudeSVrelation}, resulting in: 
 \begin{equation*}
  \rho_1^{\mathrm{gal}}=\exp\left(\rho_a-\phi_a\right)
  \qquad\mathrm{and}\qquad
  \rho_2^{\mathrm{gal}}=\exp\left(-\left(\rho_a-\phi_a\right)\right).
\end{equation*} 
Note that the estimated galvanic amplitude only depends on the anisotropy parameters, since the galvanic shift has been absorbed by $\ln\rho$ in the inductive amplitude. Finally, the anisotropic distortion parameter can be estimated from the relations above as outlined on the example given in section \ref{sec:RecAniDis2D}.


%

\end{document}